\documentclass[journal]{IEEEtran}
\newcommand{\conf}{true} 
\usepackage{amssymb}
\usepackage{amsmath}
\usepackage{latexsym}
\usepackage{psfrag}
\usepackage{color,soul}
\usepackage{graphicx,subfigure}
\usepackage{ifthen}
\newboolean{Journal}
\setboolean{Journal}{\conf}
\newcommand{\NotForConf}[1]{\ifthenelse{\boolean{Journal}}{}{#1}}
\newcommand{\IfConf}[2]{\ifthenelse{\boolean{Journal}}{#1}{#2}}

\usepackage{LouStyle}
\usepackage{rgsEnvironments}
\usepackage{hyperref}

\graphicspath{{figures/}}

\newcommand{\techreport}{false} 

\newcommand{\Figure}{\IfConf{Fig.~}{Fig.~}}
\IfConf{\newcommand{\figwidth}{0.3\textwidth}
}{\newcommand{\figwidth}{0.6\textwidth}}

\renewcommand{\eqref}[1]{(\ref{#1})}
\renewcommand{\bx}{x}

\renewcommand{\MC}{C}
\renewcommand{\MD}{D}

\renewcommand{\leq}{\leqslant}
\renewcommand{\geq}{\geqslant}

\newcommand{\X}{M}
\newcommand{\sol}{\phi^{f}}

\newcommand{\reals}{\mathbb{R}}
\newcommand{\ball}{\mathbb{B}}
\newcommand{\rb}{b}
\newcommand{\rc}{c}

\def \SMF {}
\def \STF {}  

\def \SMR {\color{black}}
\def \STR {\color{black}\normalcolor}  
\def \SM {\color{black}}
\def \ST {\color{black}\normalcolor}  

\newcommand{\Mag}[1]{}
\renewcommand{\Cred}[1]{{#1}}
\renewcommand{\Cblue}[1]{{#1}}
\newcommand{\Cblack}[1]{{#1}}

\newcommand{\Blue}[1]{{\color{black}#1}} 
\newcommand{\Red}[1]{{\color{black}#1}} 
\newcommand{\CBlue}[1]{{\color{black}#1}} 
\newcommand{\CRed}[1]{{\color{black}#1}} 
\renewcommand{\Cmag}[1]{} 

\newcommand{\RRed}[1]{{\color{black}#1}} 

\def \SMB {\color{black}}
\def \STB {\color{black}\normalcolor}  

\IfConf{
\setlength{\abovedisplayskip}{5pt}
\setlength{\belowdisplayskip}{5pt}
\setlength{\skip\footins}{0.1cm}

\setlength{\intextsep}{2.5pt}
\setlength{\textfloatsep}{3.0pt}

\usepackage{setspace}
\setstretch{1}
}
{}

\IfConf{\newcommand{\eqrefTCP}{(2)~}}{\newcommand{\eqrefTCP}{(2.1)~}}
\IfConf{\newcommand{\eqrefIzhi}{(5)~}}{\newcommand{\eqrefIzhi}{(2.4)~}}

\newboolean{TechReport}
\setboolean{TechReport}{\techreport}
\newcommand{\IfTAC}[2]{\ifthenelse{\boolean{TechReport}}{{#1}}{#2}} 

\newcommand{\mtau}{{\tau(t)}}
\newcommand{\Ztau}{(t)}
\newcommand{\Ytau}{\tau(t)}

\allowdisplaybreaks
\begin{document}

\title{\LARGE \textbf{
Notions, Stability, Existence, and Robustness of Limit Cycles\\ in Hybrid Dynamical Systems}} 

\author{Xuyang Lou, Yuchun Li, and Ricardo G. Sanfelice,~\IEEEmembership{Fellow, IEEE}\thanks{X. Lou is with
the Key Laboratory of Advanced Process Control for Light Industry (Ministry of Education), Jiangnan University, Wuxi 214122, China.
Email: {\tt Louxy@jiangnan.edu.cn}
Y. Li and R. G. Sanfelice are with
the Department of Computer Engineering, University of California, Santa Cruz, CA 95064, USA.
      Email: {\tt yuchunli,ricardo@ucsc.edu.}
Research by X. Lou has been supported by
Jiangsu Provincial Natural Science Foundation of China Grants no. BK20201340 and the 111 Project Grants no. B23008.
Research by R. G. Sanfelice partially supported by NSF Grants no. ECS-1710621 and CNS-1544396, by AFOSR Grants no. FA9550-16-1-0015 and FA9550-19-1-0053, and by CITRIS and the Banatao Institute at the University of California.}}

\date{}
\maketitle

\begin{abstract}
This paper deals with existence
and robust stability of hybrid limit cycles for a class of hybrid systems
given by the combination of continuous dynamics on
a flow set and discrete dynamics on a jump set.
For this purpose, the notion of Zhukovskii stability,
typically stated for continuous-time systems,
is extended to the hybrid systems.
Necessary conditions, particularly, a condition using a forward invariance notion,
for existence of hybrid limit cycles are first presented.
In addition, a sufficient condition, related to Zhukovskii stability,
for the existence of (or lack of) hybrid limit cycles
is established.
Furthermore, under mild assumptions, we show that asymptotic
stability of such hybrid limit cycles is not only equivalent to
asymptotic stability of a fixed point of the associated Poincar\'{e}
map but also robust to perturbations.
Specifically, robustness to
generic perturbations, which capture state noise and unmodeled
dynamics, and to inflations of the flow and jump sets are established in terms of $\mathcal{KL}$ bounds.
Furthermore, results establishing relationships between the properties of a computed Poincar\'{e} map,
which is necessarily affected by computational error,
and the actual asymptotic stability properties of a hybrid limit cycle are proposed.
In particular, it is shown that
asymptotic stability of the exact Poincar\'{e} map is preserved
when computed with enough precision.~\IfTAC{Two}{Several}~examples, including a congestion control system\IfTAC{,}{~and spiking neurons,} are presented to illustrate the notions and
results throughout the paper.
\end{abstract}

\IEEEpeerreviewmaketitle

\IfTAC{\vspace{-2mm}}{}
\section{Introduction}

\IfTAC{
\subsection{Motivation and Related Work}

}
{

\subsection{Motivation}

Due to recent technological advances requiring
advanced \IfConf{mathematical}{ma-thematical} models, hybrid systems have drawn considerable attention in recent years.
Hybrid systems have state variables that can evolve continuously (flow) and/or discretely (jump),
leading to trajectories of many types, such as those with finitely many jumps and infinite amount of flow,
with an infinite number of intervals of flow with finite (nonzero) length followed by a jump
-- potentially defining a limit cycle with a jump --
and with an infinite number of intervals of flow with finite (but decreasing) length followed by a jump, which are
the so-called Zeno solutions.
There exist several frameworks capable of modeling such systems
as well as tools for their analysis and design
\cite{hybrid:automata,vanderSchaftSchumacher00,Borrelli.03,HaddadChellaboinaNersesov06}.
Recent progress in the development of a robust stability theory for
hybrid dynamical systems has led to a
new framework, known as hybrid inclusions.
These developments appear in \cite{Goebel:book}
and include results to assure existence of solutions,
as well as to certify asymptotic stability of closed sets
and robustness to perturbations.
In spite of these advances,
the study of existence and robust stability of limit cycles for such systems
has not received much attention,
even though numerous applications in robotics \cite{Grizzle:2001},
mechanical systems \cite{Bullo:2002}, genetic regulatory networks \cite{Shu:2014},
and neuroscience \cite{Izhikevich:2003}
would benefit from results guaranteeing
such properties.

\subsection{Related Work}
}

Nonlinear dynamical systems with periodic solutions are found in many areas, including biological dynamics \cite{Adimy:2006}, neuronal systems \cite{Rodrigues:2007}, and population dynamics \cite{Gonzalez-Olivares}, to name just a few.
%
%
In recent years, the study of limit cycles in hybrid systems has received \Blue{renewed attention, mainly due to} the existence of hybrid limit cycles in many engineering applications, such as walking robots \cite{Grizzle:2001},
genetic \IfTAC{}{regulatory} networks \cite{Shu:2014},
\Blue{holonomic mechanical systems subject to impacts \cite{Slynko},}
among others. \Blue{Theory for the study of such periodic behavior dates back to the work Andronov et al. in 1966 \cite{Andronov},
where self-oscillations (limit cycles) and discontinuous oscillations
were studied.} \Blue{Limit cycles have been studied within the impulsive differential equations framework \cite{Samoilenko:1987,Lakshmikantham:1989,Samoilenko:1995},}
\Blue{for example in strongly nonlinear impulsive systems \cite{Akhmetov:1992,Perestjuk:1974},
in slowly impulsive systems \cite{Perestjuk:1974b}, in the Van der Pol equation \cite{Akhmet:2013}, in a holonomic mechanical system
subject to impacts \cite{Slynko}, and in a weakly nonlinear two-dimensional impulsive system \cite{Yu:2013}.} \Blue{These early developments pertain to
nominal systems given in the form of impulsive differential equations,
leaving the question of whether it is possible to
handle more general models, such as hybrid system models, and guarantee robustness to generic perturbations
wide open.}

As a difference to general continuous-time systems, for which the Poincar\'{e}-Bendixson theorem
\SM uses the topology of $\BR^2$ to rule out chaos and \ST
offers
criteria for existence of limit cycles/periodic orbits,
the problem of identifying the existence of limit cycles for hybrid systems has been
studied for specific classes of hybrid systems.
Specific results for existence
of hybrid limit cycles include \IfTAC{\cite{Grizzle:2001}-\cite{Colombo}.}{\cite{Grizzle:2001}, \cite{Matveev:Savkin:1999}-\cite{Colombo}.}
~In particular,
Grizzle et al. establish the existence and stability properties of a
periodic orbit of nonlinear systems with impulsive effects via the method of Poincar\'{e} sections \cite{Grizzle:2001}.
Using the transverse contraction framework, the existence and orbital stability of nonlinear hybrid limit cycles
are analyzed for a class of autonomous hybrid
dynamical systems with impulse in \cite{Tang:Manchester:2014}.
In \cite{reset}, the existence and stability of limit cycles in reset control systems
are investigated via techniques that rely on the linearization of the Poincar\'{e} map about its fixed point.
In \cite{Lou:ACC17}, we analyze the existence of hybrid
limit cycles in hybrid dynamical systems and establish necessary conditions for the existence of
hybrid limit cycles.
Clark et al. prove a version of the Poincar\'{e}-Bendixson theorem for planar hybrid dynamical systems with
empty intersection between the flow set and the jump set \cite{Clark:Bloch:2019a}, and
extend the results to the case of an arbitrary number of state spaces (each of which is a subset of $\BR^2$) and impacts in \cite{Clark:Bloch:2019b}.
More recently, Goodman and Colombo propose necessary conditions for existence of a periodic orbit related to the Poincar\'{e} map and
sufficient conditions for local conjugacy between two Poincar\'{e} maps
in systems with prespecified jump times evolving on a differentiable manifold \cite{Colombo}.
We believe that
conditions for existence of hybrid limit cycles in general hybrid systems
should play a more prominent role in analysis and control of hybrid limit cycles.
To the best of our knowledge,
tools for the analysis of
existence \SM or nonexistence \ST of hybrid limit cycles
for the class of hybrid systems in \cite{Goebel:book,Ricardo:book:2021}
are still not available in the literature.

Stability issues of hybrid limit cycles are currently a major focus in studying
hybrid systems for their practical value in applications.
Due to the complicated behavior caused by interaction
between continuous change and instantaneous change,
the study of stability of limit cycles in hybrid systems is more difficult than
the study in continuous systems or discrete systems,
and so becomes a challenging issue.
In this respect, the Poincar\'{e} map and its variations or generalizations
still play a dominating role; see, e.g., \cite{HaddadACC:2002}-\cite{Znegui:2020}.
For instance,
Nersesov et al. generalize the Poincar\'{e}  method
to analyze limit cycles for left-continuous hybrid impulsive dynamical systems \cite{HaddadACC:2002}.
Gon\c{c}alves analytically develops the local stability of limit cycles
in a class of switched linear systems when a limit cycle exists \cite{Goncalves:2000}.
The authors in \cite{Benmiloud:2019}
analyze local stability of a predefined limit cycle for switched affine systems
and
design switching surfaces by computing eigenvalues of the Jacobian of the Poincar\'{e} map.
\IfTAC{}{The trajectory sensitivity approach in \cite{Hiskens:2007} is employed
to develop sufficient conditions for stability of limit cycles in switched differential-algebraic
systems.}
Motivated by robotics applications,
the authors in \IfTAC{\cite{Grizzle:2009}-\cite{Znegui:2020}}{
\cite{Grizzle:2009}, \cite{WangHelin:2020}, \cite{Hamed:CST:2020,Znegui:2020}} analyze
the stabilization of periodic orbits in systems with impulsive effects
using the Jacobian linearization of the Poincar\'{e} return map
and the relationship between the stability of the return map and 
the stability of the hybrid zero dynamics.
To the best of our knowledge, all of the aforementioned results about limit cycles
are only suitable for hybrid systems that
have jumps on switching surfaces
and under nominal/noise-free conditions. In fact, the results therein do not characterize the robustness properties to perturbations
of stable hybrid limit cycles,
which is a very challenging problem due to the impulsive behavior
in such systems.

Besides our preliminary results in \cite{Lou:ACC17,Lou:CDC15,lou.li.sanfelice16:TAC},
results for the study of existence and robustness of limit cycles in hybrid systems
are currently missing from the literature,
being perhaps the main reason that a robust stability theory for such systems has only been \IfTAC{}{recently}
developed in \cite{Goebel:book,Ricardo:book:2021}.
In fact, all of the aforementioned results about limit cycles are formulated for hybrid systems
operating in nominal/noise-free conditions.
The development of tools that characterize the existence of hybrid limit cycles
and the robustness properties
to perturbations of stable hybrid limit cycles is very challenging and
demands a modeling framework that properly handles time and the
complex combination of continuous and discrete dynamics.

\subsection{Contributions}

Tools for the analysis of existence of limit cycles
and robustness of asymptotic stability of limit cycles
in hybrid systems are not yet available in the literature.
In this paper, we propose such tools for
hybrid systems given as hybrid inclusions \cite{Goebel:book},
which is a broad modeling framework for hybrid systems as
it subsumes hybrid automata, impulsive systems, reset systems, among others; see \cite{Goebel:book,Ricardo:book:2021} for more details.
We introduce a notion of hybrid limit cycle
for hybrid systems modeled as hybrid equations, which are given by
 \begin{equation}\label{sec2:eq1}
  \MH \left\{
  \begin{array}{cccclcccc}
    \dot{x} & = & f(x) && x \in \MC, \\
    x^{+}   & = & g(x) && x \in \MD, \\
  \end{array}
  \right.
\end{equation}
where $x\in\BRn$ denotes the state of the system,
$\dot{x}$ denotes its derivative with respect to time, and $x^{+}$ denotes its value
after a jump. The state $x$ may have components that correspond
to physical states, logic variables, timers, memory states, etc.
The map $f$ and the set $\MC$ define the continuous dynamics (or flows),
and the map $g$ and the set $\MD$ define the discrete dynamics (or jumps).
In particular, the function $f:\BRn\RA\BRn$ (respectively, $g:\BRn\RA\BRn$) is a single-valued map describing
\IfTAC{the continuous (respectively, discrete) evolution}{the continuous evolution (respectively, the discrete evolution)}
while $C\subset\BRn$ (respectively, $D \subset \BRn$) is the set on which the flow map $f$ is effective (respectively, from which jumps can occur).

For this hybrid systems framework, we develop tools for
characterizing existence of hybrid limit cycles
and robustness properties
to perturbations of stable hybrid limit cycles.\IfTAC{\footnote{
Preliminary version of the results in this paper appeared without proof
in the conference articles \cite{Lou:ACC17} and \cite{Lou:CDC15}.
}}{}~The contributions of this paper include the following:

\begin{itemize}
\item  We introduce a notion of hybrid limit cycle (with one jump per period\footnote{Here, we mainly focus on hybrid limit cycles with ``one jump per period." The case of multiple jumps per period can be treated similarly; see \IfTAC{\cite{lou.li.sanfelice16:TAC}}{\cite{lou.li.sanfelice16:TAC,Lou:ADHS15}}.}) for the class of hybrid systems in \eqref{sec2:eq1}.
   Also, we define the notion of flow periodic solution and asymptotic stability of the hybrid limit cycle for such hybrid systems.\footnote{In this work, a hybrid limit cycle is given by a closed set,
while the limit cycle defined in \cite{Grizzle:2001,Grizzle:2009,Hamed:Gregg:2017} 
is given by an open set due to
the right continuity assumption in the definition of solutions.}

\item
We present necessary conditions for existence of hybrid limit cycles, including compactness,
\IfTAC{transversality of the limit cycle,
and a continuity of the so-called time-to-impact function.}{finite-time convergence of the jump set, and transversality of the limit cycle,
as well as a continuity of the so-called time-to-impact function.}
Particularly, a condition using a forward invariance notion for existence of hybrid limit cycles is first presented.

\item
Motivated by the use of Zhukovskii stability methods for periodic orbits in continuous-time systems,
as done in \cite{Ding:2004,Yang:2000,Leonov:2006},
we introduce this notion
for the class of hybrid systems introduced in \eqref{sec2:eq1} and provide a sufficient condition for
the incremental stability introduced in \cite{Li:Sanfelice:CDC15} that involves the Zhukovskii stability notion.

\item
By assuming that the state space contains no isolated equilibrium point for the flow dynamics, we establish
a sufficient condition for the existence of hybrid limit cycles based on Zhukovskii stability.
In addition,
based on
an incremental graphical stability notion introduced in \cite{Li:Sanfelice:CDC15},
an approach to rule out existence of hybrid limit cycles in some cases
is proposed.

\item
We establish sufficient and necessary conditions for guaranteeing
(local and global) asymptotic stability of hybrid limit cycles for a class of hybrid systems.
In the process of deriving these results, we construct 
time-to-impact functions and Poincar\'{e} maps that
cope with one jump per period of a hybrid limit cycle.
\item  Via perturbation analysis for hybrid systems, we
propose a result on robustness to generic perturbations of asymptotically stable hybrid limit cycles, which allows for state noise
and unmodeled dynamics, in terms of $\KL$ bounds.

\item Due to the wide applicability of the Poincar\'{e} section method,
we present results that relate the properties of a computed Poincar\'{e} map,
which is necessarily affected by computational error,
to the actual asymptotic stability properties of hybrid limit cycles.

\end{itemize}

\IfTAC{
\vspace{-2mm}
\subsection{Notation}
}
{
\subsection{Organization and Notation}

The organization of the paper is as follows.
\begin{itemize}
\item
Section~\ref{sec:motive_ex} presents
two motivational examples.
\item
Section~\ref{sec:prelim} gives some preliminaries on hybrid systems
 and basic properties of hybrid limit cycles.
\item
Section~\ref{sec:necess} presents several necessary conditions for
existence of hybrid limit cycles. 
\item
In Section~\ref{sec:notions}, the
Zhukovskii stability notion and incremental graphical stability
notion are introduced.
Moreover, the relationship between these two notions is studied.
\item
In Section~\ref{sec:existence},
a sufficient condition for existence of hybrid limit cycles
is established including a result 
on nonexistence of hybrid limit cycles.
\item
With the hybrid limit cycle definition, Section~\ref{sec:stability} establishes
sufficient conditions
for stability of hybrid limit cycles. 
\item
Section~\ref{sec:robustness}
provides results on general
robustness of stability to perturbations.
In addition, several examples are presented throughout the paper, including
spiking neurons, which exhibit regular spiking,
and a congestion control system, which exhibits periodic solutions. 
\item
Section~\ref{sec:conclu} concludes the paper.
\end{itemize}

\textbf{Notation.}}
The set $\BRn$ denotes the $n$-dimensional Euclidean space,
$\BR_{\geq 0}$ denotes the set of nonnegative real numbers, i.e., $\BR_{\geq 0}:=[0,+\infty),$
and
$\BN$ denotes the set of natural numbers including $0$, i.e., $\BN:=\{0,1,2,\cdots\}.$
Given a vector $x\in\BRn,$ 
$|x|$ denotes its Euclidean norm. 
Given a set $S,$ $S^n$ denotes $n$ cross products of $S$,
namely $S^n=S\!\times\! S\!\times\! \cdots\!\times\! S$.
Given a continuously differentiable function $h\!:\! \BRn\RA\BR$ and
a function $f\!:\! \BRn\RA\BRn,$ the Lie derivative of $h$ at $x$
in the direction of $f$ is denoted by $L_{f}h(x):=\langle \nabla h(x),f(x) \rangle$. 
Given a function $f: \BR^{m}\RA\BRn,$ its domain of definition is denoted by
$\dom f$, i.e., $\dom f:= \{ x\in\BR^{m} : f(x) \text{ is defined} \}$.
The range of $f$ is denoted by $\text{rge}f$, i.e., $\text{rge}f:=\{f(x): x\in\dom f\}$.
Given a closed set $\MA\subset\BRn$ and a point
$x\in\BRn,$ $|x|_{\MA}:=\inf_{y\in\MA} |x-y|$.
Given a set $\MA\subset\BRn,$
$\overline{\MA}$ (respectively,
$\overline{\mathrm{con}}\ \MA$) denotes its closure (respectively, its closed
convex hull) and
$\MA^{\circ}$ denotes its interior.
Given an open set $\mathcal{X}\subset \BRn$ containing a compact set $\MA$,
a function $\omega: \mathcal{X}\RA\BR_{\geq0}$ is a
{\it proper indicator} for $\MA$ on $\mathcal{X}$ if $\omega$ is continuous, $\omega(x)=0$ if and only if $x\in\MA$, and $\omega(x)\RA\infty$ as $x$
approaches the boundary of $\mathcal{X}$ or as $|x|\RA\infty$.
\Cred{Given a sequence of set $\mathcal{X}_{i}$, $\limsup_{i\rightarrow\infty} \mathcal{X}_{i}$ denotes the outer limit of $\mathcal{X}_{i}$.}
The set $\BB$ denotes a closed unit ball in Euclidean space (of appropriate dimension)
centered at zero.
Given $\delta>0$ and $x\in \BRn$, $x+\delta\BB$ denotes a closed ball centered at $x$ with radius $\delta$.
A function $\alpha: \BR_{\geq 0}\RA\BR_{\geq0}$ belongs to class-$\MK$ $(\alpha\in\MK)$ if it is
continuous, zero at zero, and strictly increasing; it belongs to class-$\Kinf$ $(\alpha\in\Kinf)$ if, in addition, is unbounded.
A function $\beta: \BR_{\geq 0}\times\BR_{\geq 0} \RA \BR_{\geq 0}$ belongs
to class-$\KL$ $(\beta\in\KL)$ if, for each $t \geq 0,$ $\beta(\cdot, t)$ is nondecreasing and $\lim_{s\RA 0^{+}}\beta(s,t)=0$
and, for each $s \geq 0,$ $\beta(s,\cdot)$ is nonincreasing and $\lim_{t\RA\infty}\beta(s,t)=0.$

\IfTAC{
\vspace{2mm}
\section{Motivational Example}\label{sec:motive_ex}

}
{
\section{Motivational Examples}\label{sec:motive_ex}

The following examples motivate the need of tools for the study of
existence and robust stability of limit cycles in hybrid systems.
}

\IfTAC{}{
\begin{example}\label{exam:TCP}}
Consider the hybrid model for a congestion control mechanism in TCP proposed in \cite{Hespanha:Bohacek:2001}. 
The hybrid model in congestion avoidance mode can
be described as follows:
\begin{itemize}
\item when $q\in [0, q_{\max}]$
\addtocounter{equation}{1}
\begin{equation}
\begin{bmatrix}
\dot{q}\\
\dot{r}
\end{bmatrix} =
\left\{
\begin{aligned}
&\begin{bmatrix}
\max\{0, r-B\} \\
a
\end{bmatrix}  & \; \mathrm{if }\;\; q=0 \\
&\begin{bmatrix}
r-B\\
a
\end{bmatrix}  & \; \mathrm{if }\;\; q>0
\end{aligned}
\right. \tag{\theequation a}
\end{equation}
\item
when $q=q_{\max}, r\geq B$
\IfTAC
{
\begin{equation}
 (q^{+}, r^{+}) = (q_{\max},  mr) \tag{\theequation b}
    \label{eq:TCP1}
\end{equation}\vspace{-4mm}
}
{
\begin{equation}
    \begin{bmatrix}
     q^{+}\\
     r^{+}
    \end{bmatrix}
    =
    \begin{bmatrix}
    q_{\max}\\
    mr
    \end{bmatrix} \tag{\theequation b}
    \label{eq:TCP1}
\end{equation}}
\end{itemize}
where $q\in [0, q_{\max}]$ denotes the current queue size, $q_{\max}$ is the maximum
queue size, $r\geq 0$ is the rate of incoming data packets,
and $B\geq 0$ is the rate of outgoing packets.
The constant $a\geq 1$ 
reflects the rate of growth of incoming data packets $r$
while $m\in (0, 1)$ 
reflects the factor that makes the rate of incoming packets decrease;
see \cite{Hespanha:Bohacek:2001} for details.
The model in \eqrefTCP reduces the rate of incoming packets $r$
by the factor $m$ if
the queue size $q$ equals the maximum value
$q_{\max}$ with rate larger than or equal to $B$.
\begin{figure}[!ht]
\psfrag{MT}[][][0.8]{{\color{magenta}$M_{\mathrm{T}}$}}
\psfrag{M1}[][][0.7]{{\color{magenta}$M_1$}}
\psfrag{Q}[][][0.7]{{\color{magenta}$M_2$}}
\psfrag{qq}[][][0.7]{$q$}
\psfrag{rr}[][][0.7]{$r$}
\psfrag{P1}[][][0.6]{\Cblue{$P_{1}$}}
\psfrag{P2}[][][0.6]{\Cblue{$P_{2}$}}
\psfrag{P3}[][][0.6]{\Cblue{$P_{3}$}}
\psfrag{P4}[][][0.6]{\Cblue{$P_{4}$}}
\psfrag{P5}[][][0.6]{\Cblue{$P_{5}$}}
\psfrag{Q1}[][][0.6]{$(q_{\max}, B\!+\!\sqrt{2aq_{\max}})$}
\psfrag{Q2}[][][0.6]{\sethlcolor{white}\hl{$(0,B)$}}
\psfrag{Q3}[][][0.6]{\sethlcolor{white}\hl{$(\frac{B^2}{2a}, 0)$}}
\centering{\includegraphics[width=\figwidth]{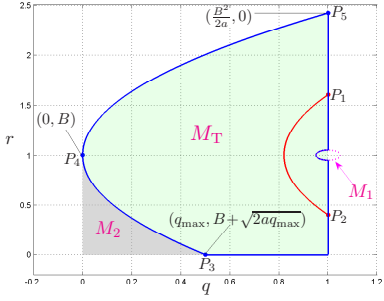}}
\caption{
Diagram of the compact set $M_{\mathrm{T}}$ denoted in the region with
light green filled 
pattern.    
Parameters used in the plot are $B = 1,a = 1,{q_{\max }} = 1,$ and $m = 0.25$.
The points $P_{1}$ and $P_{2}$ correspond to state values in a limit cycle just before and right after each jump, respectively.
The points $P_{1}$ corresponds to $(q,r)=(q_{\max}, 2B/(1+m))$,
the point $P_{2}$ corresponds to $(q,r)=(q_{\max}, 2Bm/(1+m))$,
the point $P_{3}$ corresponds to $(q,r)=(B^2/(2a), 0)$,
the point $P_{4}$ corresponds to $(q,r)=(0, B)$,
and the point $P_{5}$ corresponds to $(q,r)=(q_{\max}, B+\sqrt{2aq_{\max}})$.
}
\label{tcp:region}
\end{figure}

We are interested in the \IfTAC{}{hybrid} system \eqrefTCP restricted to
the region
\begin{equation}\label{eqMT}
{\footnotesize
\! M_{\mathrm{T}}\!\!:=\!\!
\{(q,r)\!\in\!\BR_{\geq 0}\!\!\times\!\! \BR_{\geq 0}\!:
q\!\leq\! q_{\max}, aq\!\geq\! \frac{1}{2}r^2\!\!-\!\!Br\!+\!\frac{B^2}{2} \}\!\!}
\end{equation}
for given parameters $a$, $m$, $q_{\max}$
and $B$ (later, the set $M_{\mathrm{T}}$ will be part of our analysis); see \Figure~\ref{tcp:region}.
From the first piece in the definition in \eqrefTCP with $a>0$,
for any maximal solution with initial condition with zero
$q$ and $r$ less than $B$, $q$ remains at zero until $r>B$.
\Figure~\ref{tcp:region} is shown to analyze
how we get a region from which a limit cycle with one jump exists.
The points in the curve $P_{3}\RA P_{4}\RA P_{5}$
satisfy $aq=\frac{1}{2}r^2-Br+\frac{B^2}{2}$.
Solutions from the region $M_2$ 
result in solutions such that $q$ reaches zero and remains at zero 
until $r=B$ (point $P_{4}$).
The open set
$M_{1}:=\{(q_{\max},B)\}+\varepsilon\BB^{\circ}$
with $\varepsilon>0$ small enough,
will be part of our analysis in Example~\ref{exam:TCP3}
and 
be ruled out to ensure the transversality of the limit cycle.
We are not interested in the region $M_2$
with
gray filled pattern
as it \Blue{leads to} a complex hybrid model
which \Blue{might} be hard to be analyzed.
The compact set $M_{\mathrm{T}}$ is marked by the region with light green filled pattern.
Hence, the set $M_{\mathrm{T}} \setminus M_{1}$ (the region surrounded by blue line)
is the region of the state space that we are interested in. Note that if the value of $r$ after a jump from the point $P_{5}$
is larger than $B$ (for instance, point $P_{5}$ jumps to point $P_{1}$),
a consecutive jump will happen.
Therefore, to avoid this case, 
we impose the
condition $m(B+\sqrt{2aq_{\max}})<B$.

From points in the set $M_{\mathrm{T}}$, solutions approach a limit cycle.
 On $M_{\mathrm{T}}$ and for parameters satisfying the conditions above,
the resulting system with $(q,r)\in M_{\mathrm{T}}$
can be described as a hybrid system $\MH_{\textrm{\tiny TCP}}$
on $M_{\mathrm{T}}$ with data
\begin{align}
\MH_{\textrm{\tiny TCP}}
\left\{
\begin{aligned}
\dot{x}&= f_{\textrm{\tiny TCP}}(x) :=
\left[
  \begin{array}{cc}
    r-B\\
    a
  \end{array}
\right] &\quad
x\in  C_{\textrm{\tiny TCP}},\\
x^{+}&= g_{\textrm{\tiny TCP}}(x): = \left[
  \begin{array}{cc}
    q_{\max}\\
    mr
    \end{array}
\right] &\quad
x\in  D_{\textrm{\tiny TCP}},
\end{aligned}
\right.
\label{eq:TCP3}
\end{align}
where $x=(q,r),$ $ C_{\textrm{\tiny TCP}}=\{x\in\BR^{2}: q\leq q_{\max} \},$
$ D_{\textrm{\tiny TCP}}=\{x\in\BR^{2}: q=q_{\max}, r\geq B\}.$

A limit cycle of the system in \eqref{eq:TCP3}
with parameters $B=1, a=1, m=0.25$, and $q_{\max}=1$
\IfTAC{
{is depicted in~\Figure~\ref{tcp:region}.
This figure shows
in red a limit cycle denoted as $\MO$ and}
}
{is depicted in~\Figure~\ref{tcp:fig1}.
This figure shows
in blue a limit cycle denoted as $\MO$ and}
defined by the solution to the congestion control system
with initial condition $P_{2}=\{(1,0.4)\}$. 
This solution flows to the point $P_{1}$,
jumps to the point $P_{2}$,
and then flows back to $P_{1}$.
The interest in this paper is to find conditions under which
such limit cycles may exist.
\IfTAC{}{
\begin{figure}[!ht]
\psfrag{x1}[][][1.0]{$q$}
\psfrag{x2}[][][1.0]{$r$}
\psfrag{A}[][][1.0]{$P_{1}$}
\psfrag{B}[][][1.0]{$P_{2}$}
\psfrag{O}[][][0.9]{$\MO$}
\centering{\includegraphics[width=\figwidth]{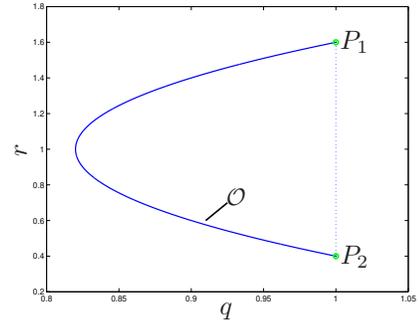}}
\caption{
A limit cycle of the congestion control system in \eqref{eq:TCP3}.
The point $P_{1}=\{(1, 1.6)\}$ ($P_{2}=\{(1, 0.4)\}$, respectively) corresponds to the value of the state
before the jump (after the jump, respectively).
}
\label{tcp:fig1}
\end{figure}
}
\IfTAC{}{
\end{example}
}

\IfTAC{}{
\begin{example}\label{exam:Izhi}
Consider the Izhikevich neuron model \cite{Izhikevich:2003}
given by
\addtocounter{equation}{1}
\begin{equation}
\left\{
\begin{array}{l c l}
\dot{v}=0.04 v^{2}+5v+140-w+I_{\text{ext}},\\
\dot{w}=a(bv-w),
\end{array}
\right.
 \tag{\theequation a}
\end{equation}
%
where $v$ is the membrane
potential, $w$ is the recovery
variable, and $I_{\text{ext}}$ represents
the synaptic (injected) DC current. 
When the membrane voltage of a neuron increases
and reaches a threshold (which in \cite{Izhikevich:2003} is
equal to 30 millivolts),  
the membrane voltage and the recovery variable
are instantaneously reset via the  
following rule:
\begin{equation}\label{eq:Izhi1ab}
\text{when } v\geq 30 \text{, then}
\left\{
\begin{array}{l c l}
v^{+}=c,\\
w^{+}=w+d.
\end{array}
\right. \tag{\theequation b}
\end{equation}
The value of the input $I_{\text{ext}}$ and the model parameters $a, b, c$, and $d$ are used to determine the
neuron type. In fact,
the model in \eqrefIzhi can exhibit a specific firing pattern (of all known types)
of cortical neurons when these parameters are appropriately chosen \cite{Izhikevich:2003}. 
For instance, when 
the parameters are chosen as
\begin{equation}\label{eq:Izhi:param}
a=0.02, b=0.2, c=-55, d=4, I_{\text{ext}}=10,
\end{equation}
the neuron model exhibits intrinsic bursting behavior.
This corresponds to a limit cycle, which is denoted as $\MO$,
and is defined by the solution to \eqrefIzhi that jumps from point $P_{1}$ to point $P_{2}$
and then flows back to $P_{1}$; see~\Figure~\ref{fig:1}.

As suggested in~\Figure~\ref{fig:1}, the limit cycle $\MO$ is
asymptotically stable.
In particular,
solutions initialized close to the set $\MO$ stay close for all time and
converge to the set $\MO$ as time gets large.
For instance,
the trajectory (black line) of a solution starting from $(-54.76, -3.5)$ (the point $P_{3}$ in the subfigure),
which is close to the point $P_{2}$,
remains close to the limit cycle $\MO$ and approaches it eventually.
However, solutions initialized relatively far away from the set $\MO$
may not stay close for all time.
For instance, as shown in~\Figure~\ref{fig:1},
the trajectory (red line) of a solution starting from $(-54.5, -3.5)$ (the point $P_{4}$ in the subfigure)
that is close to the point $P_{2}$
first goes far away from the limit cycle $\MO$ and approaches it eventually.

Interestingly, solutions to the neuron model with state perturbations, in particular,
solutions to \eqrefIzhi with  an admissible state
perturbation\footnote{A mapping $e$ is an admissible state perturbation if $\dom e$ is a hybrid time domain and the function $t\mapsto e(t,j)$ is measurable on $\dom e \cap (\BR_{\geq 0}\times
\{j\})$ for each $j\in\BN$. See \cite[Definition 4.5]{Goebel:book} for more details.}, may not be always close to the nominal solutions. For instance,
an additive perturbation $e=(0.24, 0)$ (or $e=(0.5, 0)$, respectively) to $(v^{+},w^{+})$ after a jump from the point $P_{1}$ would result in
a state value equal to the point $P_{3}$ (or
to the point $P_{4}$, respectively) instead of the point $P_{2}$.
As shown in~\Figure~\ref{fig:1}, the trajectory (black line) from the point $P_{3}$
remains close to the limit cycle $\MO$ and approaches it eventually,
while the trajectory (red line) from the point $P_{4}$ does not stay close to the one
from the point $P_{2}$. Since the points $P_{3}$ and $P_{4}$ are close to each other
but the trajectories from those points are not,
the stability property of the limit cycle $\MO$
has a small margin of robustness to perturbations.
\begin{figure}[!ht]
\vspace{3mm}
\centering
\psfrag{x1}[][][0.8]{\hspace{0.15in}Membrane potential, $v$ (mV)}
\psfrag{y1}[][][0.8]{\hspace{0.15in}Recovery variable, $w$}
\psfrag{A}[][][0.7]{$P_{1}$}
\psfrag{B}[][][0.7]{$P_{2}$}
\psfrag{C}[][][0.7]{$P_{3}$}
\psfrag{D}[][][0.7]{$P_{4}$}
\includegraphics[width=\figwidth]{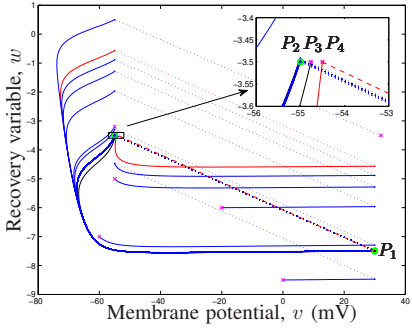}
\vspace{-3mm}
\caption{
Phase plot of several solutions  
to the Izhikevich neuron model
in \eqrefIzhi
with different conditions.
The point $P_{1}$ corresponds to $(v,w)=(30, -7.5)$,
the point $P_{2}$ corresponds to $(v,w)=(-55, -3.5)$,
the point $P_{3}$ corresponds to $(v,w)=(-54.76, -3.5)$,
and the point $P_{4}$ corresponds to $(v,w)=(-54.5, -3.5)$.
\Cred{The limit cycle of the Izhikevich neuron model is defined by the solution to \eqrefIzhi from the point $P_2$ that jumps at the
point $P_1$.}}
\label{fig:1}
\end{figure}
\end{example}
}

\IfTAC{}
{Motivated by these examples, our interest is in developing constructive analysis tools that
can be applied to such systems so as to
guarantee the existence, stability and robustness properties of hybrid limit cycles.
which are missing in the literature. 
}

\IfTAC{
\section{Definitions and Basic Properties}\label{sec:prelim}
}{
\section{Definitions and Basic Properties of Hybrid Systems with Hybrid Limit Cycles}\label{sec:prelim}
}

\subsection{Hybrid Systems}
\label{subsec:hybrid}

We consider hybrid systems $\MH$ as in \cite{Goebel:book}, which can be written as
 in \eqref{sec2:eq1}.
The data of a hybrid system $\MH$ is given by $(C, f, D, g )$. The restriction of $\MH$ on a set $M$ is defined as $\MH|_{M} = (M\cap C, f, M\cap D, g)$.
A solution to $\MH$ is parameterized by ordinary time $t$ and a counter $j$ for jumps. It is given by a hybrid arc\footnote{A hybrid arc is a function $\phi$ defined on a hybrid time domain and
for each $j\in\BN$, $t\mapsto \phi(t,j)$
is locally absolutely continuous. 
A \emph{compact hybrid time domain} is a set $\ME\subset \BR_{\geq 0}\times \BN$ of the form
$\ME=\bigcup_{j=0}^{J-1}([t_{j},t_{j+1}],j)$
for some finite sequence of times $0=t_{0}\leq t_{1}\leq \cdots\leq t_{J}$;
the set $\ME$ is a \emph{hybrid time domain} if for all $(T, J)\in \ME,$ $\ME \cap ([0, T]\times \{0,1,\cdots,J\})$
is a compact hybrid time domain. Denote $I^{j}:=\{t: (t,j)\in \mathcal{E}$.
}
$\phi: \dom\phi\RA\BRn$ \SM that satisfies the dynamics of
$\MH$; see \cite{Goebel:book} for more details. \ST A solution $\phi$ to $\MH$ is said to be
complete if $\dom\phi$ is unbounded.
It is Zeno if it is complete and the projection
of $\dom\phi$ onto $\BR_{\geq 0}$ is bounded.
It is discrete if $\dom\phi\subset \{0\}\times\BN$.
It is said to be maximal if it is not a (proper) truncated
version of another solution.
The set of maximal solutions to $\MH$ from the set $K$ is denoted as
$$
\MS_{\MH}(K)\!:=\!\{\phi\!:\! \phi \text{ is a maximal solution to } \!\MH\! \text{ with }\! \phi(0,0)\!\in\! K\}.
$$

We define $t\mapsto \sol(t,x_{0})$ as a solution of the flow dynamics~\IfTAC{$\dot{x}=f(x)\quad x\in C$}{$$\dot{x}=f(x)\quad x\in C$$}
from $x_{0}\in \SM \overline{C} \ST $. A hybrid system $\MH$ is said to be well-posed if it satisfies the {\em hybrid basic conditions}, \SM namely, \ST
\begin{enumerate}
\item [A1)] The sets $C, D\subset\BRn$ are closed. 
\item [A2)] The flow map  $f:C\RA\BRn$ 
and the jump map $g:D\RA\BRn$ are continuous.
\end{enumerate}

The following notion of $\omega$-limit set of a hybrid arc
is used in Section \ref{sec:existence} to formulate
sufficient conditions for the existence of hybrid limit cycles.
\begin{definition}\label{def:LS1}
\cite[Definition 6.17]{Goebel:book}
The $\omega$-limit set of a hybrid arc $\phi : \dom\phi\RA\BRn$, denoted
$\Omega(\phi)$, is the set of all points $x\in\BRn$ for which there
exists a sequence $\{ (t_{i}, j_{i}) \}_{i=1}^{\infty}$ of points $(t_i, j_i) \in \dom\phi$ with $\lim_{i\RA\infty} t_i + j_i=\infty$
and $\lim_{i\RA\infty} \phi(t_i, j_i) = x$. Every such point $x$ is an $\omega$-limit point of $\phi$.
\end{definition}

For more details about this hybrid systems framework, we refer the readers to \cite{Goebel:book}.

\IfTAC{\vspace{-2mm}}{}
\subsection{Hybrid Limit Cycles}

Before revealing their basic properties,
we define hybrid limit cycles.
For this purpose,
we consider the following notion of flow periodic solutions.

\begin{definition}
\label{def:periodic_sol}
{(flow periodic solution)}
A complete solution $\phi^{*}$ to $\MH$ is {\em flow periodic with period $T^{*}$ and one jump in each
period} if $T^{*}\in (0, \infty)$ is the smallest number such that 
$\phi^{*}(t+T^{*},j+1)=\phi^{*}(t,j)$ for all $(t,j)\in\dom\phi^{*}$.
\end{definition}

The definition of a flow periodic solution $\phi^{*}$ with period $T^{*}>0$ 
above implies that if $(t,j)\in\dom\phi^{*}$, then $(t+T^{*},j+1)\in\dom\phi^{*}$.
For a notion allowing for multiple jumps in a period,
see \IfTAC{\cite{lou.li.sanfelice16:TAC}}{\cite{lou.li.sanfelice16:TAC,Lou:ADHS15}}.
A flow periodic solution to $\MH$ as in Definition~\ref{def:periodic_sol} generates a hybrid limit cycle.
\begin{definition}
\label{def:periodic_orbit}
{(hybrid limit cycle)}
A flow periodic solution $\phi^{*}$ with period $T^{*}\in (0, \infty)$ and one jump in each period defines a {\em hybrid limit cycle}\footnote{
\SM Alternatively, the hybrid limit cycle $\MO$
\ST can be written as
$\{x\in\BRn: x=\phi^{*}(t,j), t\in [t_{s},t_{s}+T^{*}], (t,j)\in\dom\phi^{*} \} $
for some $t_{s}\in\BR_{\geq 0}$.
}
$\MO:=\{ x\in\BRn: x=\phi^{*}(t,j), (t,j)\in\dom\phi^{*}\}$.
\end{definition}

\IfTAC{}{
Perhaps the simplest hybrid system with a hybrid limit cycle is the scalar system
capturing the dynamics of a timer with resets, namely
\begin{equation}\label{eq:timer} 
 \MH_{\mathrm{T}} \ \left\{
  \begin{array}{cccclcccc}
    \dot{\chi} & = & 1 && \chi \in [0,1], \\
    \chi^{+}   & = & 0 && \chi = 1, \\
  \end{array}
  \right.
\end{equation}
where $\chi\in [0,1]$.
Its unique maximal solution from $\xi\in [0, 1]$ is given by
$\phi(t,j)=  \xi+t-j$
for each $(t,j)\in \BR_{\geq 0}\times \BN$ such that $ t \in [\max \{0, j-\xi\}, j+1-\xi]$.
The hybrid limit cycle generated by $\phi$ is ${\MO} = \{\chi \in [0,1]: \chi = \phi(t,1), t\in [1-\xi, 2-\xi]\}
= [0,1]$.

\begin{remark}
The definition of a hybrid limit cycle $\MO$ with period $T^{*}\!>\!0$ 
implies that
$\MO$ is nonempty and contains more than two points; in particular,
a hybrid arc that generates $\MO$ cannot be
discrete.
A hybrid limit cycle $\MO$ is restricted to have one jump per period,
but extensions to more complex cases are possible \IfTAC{\cite{lou.li.sanfelice16:TAC}}{\cite{lou.li.sanfelice16:TAC,Lou:ADHS15}}.
\end{remark}
}

\IfTAC{
In \cite[Example 3.5]{lou:TAC:2022}, we revisit the example in Section~\ref{sec:motive_ex} to further illustrate the hybrid limit cycle notion in Definition~\ref{def:periodic_orbit}.}
{
Next, the more advanced examples in Section~\ref{sec:motive_ex} are revisited to further illustrate the hybrid limit cycle notion in Definition~\ref{def:periodic_orbit}.


\begin{example}\label{exam:TCP2}
Consider the hybrid congestion control system $\MH_{\textrm{\tiny TCP}}$ on $M_{\mathrm{T}}$ in \IfTAC{Section \ref{sec:motive_ex}}{Example~\ref{exam:TCP}}.
\IfTAC{As suggested by~\Figure~\ref{tcp:region},}
{As suggested by~\Figure~\ref{tcp:fig1},}
system \eqref{eq:TCP3} has
a flow periodic solution $\phi^{*}$ with period $T^{*}$. 
In fact, a solution $\phi^{*}$ to the system \eqref{eq:TCP3}
starting from the point $P_{2}$
in \IfTAC{\Figure~\ref{tcp:region}}{\Figure~\ref{tcp:fig1}},
i.e., $\phi_1^{*}(0,0) = \xi_1$ and $\phi_2^{*}(0,0) = \xi_2$
with $(\xi_1, \xi_2)=(q_{\max}, 2Bm/(1+m))\in M_{\mathrm{T}}\cap ( C_{\textrm{\tiny TCP}}\cup  D_{\textrm{\tiny TCP}})$,
is given by
\begin{align}
\phi_1^{*}(t,j)&=
    \xi_1 + (\xi_2-B)(t-jT^{*})+ \frac{a(t-jT^{*})^{2}}{2},  
  \label{eq:TCP4}\\
\phi_2^{*}(t,j)&=
    \xi_2 + a(t-jT^{*}),  
  \label{eq:TCP5}
\end{align}
for all $(t,j)\in \dom \phi^{*}$ with $t \in \left[jT^{*}, (j+1)T^{*} \right]$ and $j \geq 0,$
where
$T^{*}=2B(1-m)/(a+ma)$.
From \eqref{eq:TCP4} and \eqref{eq:TCP5},
when $t=T^{*}$ and $j=0$,
it is easy to obtain that
$\phi_1^{*}(T^{*},0)=
\xi_1 + (\xi_2-B)T^{*}+ \frac{a {T^{*}}^{2}}{2}
=q_{\max}$
and
$\phi_2^{*}(T^{*},0)=\xi_2 + aT^{*}=2Bm/(1+m)+2B(1-m)/(1+m)=2B/(1+m)>B$ (using the fact $m\in (0,1)$).
Therefore, $\phi^{*}(T^{*},0)\in D_{\textrm{\tiny TCP}}$ and the first jump happens.
Then, according to the jump map, the state is updated with $\phi_1^{*}(T^{*},1)=q_{\max}=\xi_1$
and $\phi_2^{*}(T^{*},1)=m\phi_2^{*}(T^{*},0)=2Bm/(1+m)=\xi_2$, which equals the initial value.
In addition, since $m\in (0,1)$, we have $2Bm/(1+m)<B$
implying that there is only one jump after every interval of flow.

Using \eqref{eq:TCP4} and \eqref{eq:TCP5},
for solutions initialized at $(\xi_1, \xi_2)=(q_{\max}, 2Bm/(1+m))$,
a hybrid limit cycle can be characterized as $\MO = \{x\in \mathbb{R}_{\geq 0} \times \mathbb{R}_{\geq 0}: x =  \phi^{*}(t,1), t\in [\bar t_1, \bar t_1 + T^{*}]\}.$ Furthermore, any solution $\bar{\phi} $ to $\MH_{\textrm{\tiny TCP}}$ from $\bar{\phi}(0,0) \in \MO$ satisfies $|\bar{\phi}(t,j)|_\MO=0$ for all $(t,j)\in \dom \bar{\phi}$.
In fact, using the parameters given in \IfTAC{Section \ref{sec:motive_ex}}{Example~\ref{exam:TCP}}, the solution to $\MH_{\textrm{\tiny TCP}}$ from $P_{2}=(1, 0.4)$ generates
a hybrid limit cycle as depicted in red in \IfTAC{\Figure~\ref{tcp:region}}{\Figure~\ref{tcp:fig1}}.
\end{example}
}

\IfTAC{}{
\begin{example}\label{exam:Izhi:revisit1}
Consider the Izhikevich neuron system in Example~\ref{exam:Izhi}.
This neuron system is slightly modified and
written as a hybrid system as in \eqref{sec2:eq1},
which we denote as $\MH_{\mathrm{I}}$ and is given by
\begin{equation}\label{eq:izhi2}
  \MH_{\mathrm{I}} \left\{
  \begin{aligned}
    \dot{\bx} & = f_{\mathrm{I}}(\bx):=
    \begin{bmatrix}
    f_{1}(\bx)\\
    a(bv-w)
    \end{bmatrix}
     && \bx \in \MC_{\mathrm{I}}, \\
    \bx^{+} & =g_{\mathrm{I}}(\bx):=
    \begin{bmatrix}
    c\\
    w+d
    \end{bmatrix}
     && \bx \in \MD_{\mathrm{I}}, \\
  \end{aligned}
  \right.
\end{equation}
\begin{figure}[ht]
\vspace{3mm}
\centering
\psfrag{TT}[][][1.0]{$t$}
\psfrag{VV}[][][1.0]{$v$}
\psfrag{WW}[][][1.0]{$w$}
\includegraphics[width=\figwidth]{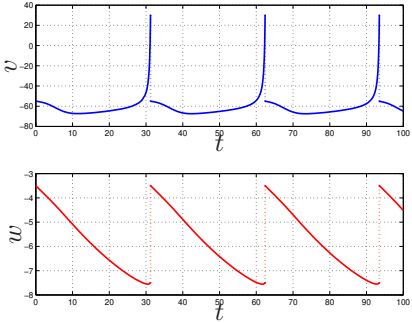}
\vspace{-3mm}
\caption{
State trajectories of a solution $\phi^{*}$ to $\MH_{\mathrm{I}}$
from $\phi^{*}(0,0)=(-55, -3.5)$.
}
\label{fig:timevw}
\end{figure}
where $x=(v,w)$, $f_{1}(\bx)=0.04 v^{2}+5v+140-w+I_{\text{ext}}$,
$\MC_{\mathrm{I}}=\{x\in\BR^{2}: v\leq 30\}$,
$\MD_{\mathrm{I}}=\{x\!\in\!\BR^{2}\!:\! v=30, f_{1}(x)\!\geq\! 0\},$
where $f_{1}(x)\!\geq \! 0$ models the fact that the
spikes occur when the membrane potential $v$ grows to the
$30 \mbox{mV}$ threshold.

The neuron system has a flow periodic solution with $T^{*}$. 
However, due to the nonlinear form of the flow map, an analytic expression of the solution is not available.
For parameters of $\MH_{\mathrm{I}}$ given by \eqref{eq:Izhi:param},
a numerical approximation of this solution is shown
in~\Figure~\ref{fig:timevw}, which is given by
the solution $\phi^{*}$ to $\MH_{\mathrm{I}}$ from
$\phi^{*}(0,0)=(-55, -3.5)$. This solution is flow periodic
with $T^{*}\approx 31.24 \mbox{ms}$.  
\end{example}
}

\IfTAC{\vspace{-2mm}}{}
\section{Necessary Conditions} 
\label{sec:necess}

\IfTAC{}{Next, we study necessary conditions for existence of hybrid limit cycles
defined in Section~\ref{sec:prelim}.
In Section~\ref{sec4:A}, under mild assumptions, we derive necessary conditions
that guarantee
compactness and transversality of the limit cycle, finite-time convergence of the jump set,
as well as a continuity of the time-to-impact function.
In Section~\ref{sec4:B},
we present a necessary condition for existence of a hybrid limit cycle with period $T^*$
using a forward invariance notion.
}

\IfTAC{
\subsection{Necessary Conditions for Existence of Hybrid Limit Cycles}
\label{sec4:A}
}{
\subsection{Necessary Conditions for Existence of Hybrid Limit Cycles in a Class of Hybrid Systems}
\label{sec4:A}
}

In this subsection, we derive several necessary conditions for the existence of hybrid limit cycles for a class of hybrid systems $\MH$ as in \eqref{sec2:eq1} satisfying the following properties.

\begin{assumption}
\label{ass:basic_data}
For a hybrid system $\MH = (C,f,D,g)$ on $\BRn$ 
and a compact set $M\subset\BRn$,
there exists a continuously differentiable function $h: \BRn\RA \BR$
such that
\begin{itemize}
\item [1)] the flow set can be written as $C=\{ x\in\BRn: h(x)\geq 0 \}$ and  the jump set as $D\!=\!\{ x\!\in\!\BRn\!: h(x)\!=\!0, L_{f}h(x)\!\leq\! 0 \}$;
\item [2)] the flow map $f$ is continuously differentiable on an open neighborhood of $M\cap C$,
and the jump map $g$ is continuous on $M\cap D$;

\item [3)]$L_{f}h(x)\!<\!0$ for all $x\!\in\! M\cap D$, and $g(M\!\cap\! D)\cap (M\!\cap\! D)\!=\!\emptyset$.

\end{itemize}
\end{assumption}

\begin{remark}\label{rem:mainAssum}
Item $1)$ in Assumption~\ref{ass:basic_data} implies that
flows occur when $h$ is nonnegative while
jumps only occur at points in the zero level set of $h$.
Note that  since $h$ is continuous and $f$ is continuously differentiable,
the flow set and the jump set are closed.
The state $x$ may include 
logic variables, counters, timers, 
etc. 
The continuity  property of $f$ in item $2)$ of Assumption~\ref{ass:basic_data}
is further required for the existence of solutions to $\dot x = f(x)$ according to \cite[Proposition 2.10]{Goebel:book}.
Moreover, item $2)$  also guarantees that solutions to $\dot x = f(x)$
depend continuously on initial conditions.
\Blue{In the upcoming results,} item $3)$ 
in Assumption~\ref{ass:basic_data} allows us to
\Blue{establish a transversality property and} 
restrict the analysis of a hybrid system $\MH$ to a region of \Cblue{a} state space $M\subset \BRn$,
leading to the restriction of $\MH$ given by
$\MH|_{M}:=(M\cap C, f, M\cap D, g).$ \Blue{The condition $g(M\!\cap\! D)\cap (M\!\cap\! D)\!=\!\emptyset$
is assumed to exclude 
discrete solutions.}
 As we will show later,
the set $M$ is appropriately chosen for each specific system such that
it guarantees completeness of maximal solutions to $\MH|_{M}$ and
the existence of flow periodic solutions.
This is illustrated in \IfTAC{Section \ref{sec:motive_ex}}{Example~\ref{exam:TCP}} with a set $M_{\mathrm{T}}$.
\IfTAC{}{See also the forthcoming Example~\ref{exam:Izhi:revisit2} and Example~\ref{exam:AS}.}
\end{remark}

\IfTAC{
\CRed{\begin{remark}
\label{rk:hybrid_basic_conds}
By  items $1)$ and $2)$ of Assumption~\ref{ass:basic_data}, the data of
$\MH|_{M}$ satisfies the hybrid basic conditions \cite[Assumption 6.5]{Goebel:book}.
Then,  using item $3)$ of Assumption~\ref{ass:basic_data},
\cite[Lemma 2.7]{Sanfelice:2007} implies that
for any bounded and complete solution $\phi$ to $\MH|_{M}$ there exists $r>0$ such that
$t_{j+1}-t_{j}\geq r$ for all $j\geq 1,$
$t_{j}=\min I^{j},$ $t_{j+1}=\max I^{j}$;
i.e.,
the elapsed time between two consecutive jumps is uniformly bounded below by a positive constant.
\end{remark}}
}{
\begin{remark}
\label{rk:hybrid_basic_conds}
By  items $1)$ and $2)$ of Assumption~\ref{ass:basic_data}, the data of
$\MH|_{M}$ satisfies the hybrid basic conditions \cite[Assumption 6.5]{Goebel:book}.
Then,  using item $3)$ of Assumption~\ref{ass:basic_data},
\cite[Lemma 2.7]{Sanfelice:2007} implies that
for any bounded and complete solution $\phi$ to $\MH|_{M}$ there exists $r>0$ such that
$t_{j+1}-t_{j}\geq r$ for all $j\geq 1,$
$t_{j}=\min I^{j},$ $t_{j+1}=\max I^{j}$;
i.e.,
the elapsed time between two consecutive jumps is uniformly bounded below by a positive constant.
\end{remark}
}

It can be shown that a hybrid limit cycle generated by
periodic solutions as in Definition \ref{def:periodic_orbit} is closed and bounded,
as established in the following result.

\begin{lemma}
\label{lem:closeness_orbit}
Given a hybrid system $\MH=(C,f,D,g)$ on $\BRn$ and a closed set 
$M\subset\BRn$ satisfying Assumption~\ref{ass:basic_data},
suppose that $\MH$ has a hybrid limit cycle $\MO$.
Then, $\MO$ is compact and forward invariant\footnote{Every $\phi\in \MS_{\MH}(\MO)$ is complete and satisfies $\text{rge}\ \phi\subset \MO$; see \cite[Definition~3.3]{Chai:Sanfelice:TAC19}.}.
\end{lemma}
\IfTAC{
\noindent A proof can be found in \cite[Lemma 4.4]{lou:TAC:2022}.
}{
\begin{proof}
We show that any hybrid limit cycle $\MO$ is
closed and bounded.
First, to prove closeness of $\MO$, consider a flow periodic solution $\phi^{*}$ to $\MH$ from $\phi^{*}(0,0)$ associated with a hybrid limit cycle $\MO$ with period $T^{*}$.
Then, $\MO = \{x\in \mathbb{R}^n: x = \phi^{*}(t,j), t\in [0,T^{*}], j\in\{0,1\}, (t,j)\in \mbox{dom }\phi^{*} \}$.
We have the following cases for the first jump of
$\phi^{*}$. Let $t^{*}$ be such that
$(t^{*}, 0)\in \dom\phi^{*}$,
$(t^{*}, 1)\in \dom\phi^{*}$. Then:
\begin{itemize}
\item If $t^{*} = 0$, i.e., $(0,1)\in\dom \phi^{*}$, then $\MO = \{x\in \mathbb{R}^n: x = \phi^{*}(0,0) \} \bigcup \{x\in \mathbb{R}^n: x = \phi^{*}(t,1), t\in [0,T^{*}],(t,1)\in \mbox{dom }\phi^{*} \}$. Since $\phi^{*}(0,0)\in D$ and $D\subset \dom g$, $\phi^{*}(0,1)=g(\phi^{*}(0,0))$ is bounded.
    Since $t \mapsto \phi^{*}(t,1)$ is continuous on $[0,T^{*}]$ due to the continuity of $f$ from item $2)$ of Assumption~\ref{ass:basic_data}
    and $[0,T^{*}]$ is a closed interval, the set $\{x\in \mathbb{R}^n: x = \phi^{*}(t,1), t\in [0,T^{*}],(t,1)\in \mbox{dom }\phi^{*} \}$ is closed. Then, $\MO$ is closed as it is the union of two closed sets;
\item If $0<t^{*}<T^{*}$, then $\MO = \{x\in \mathbb{R}^n: x = \phi^{*}(t,0), t\in [0,t^{*}], (t,0)\in \dom \phi^{*} \} \bigcup \{x\in \mathbb{R}^n: x = \phi^{*}(t,1), t\in [t^{*},T^{*}],(t,1)\in \mbox{dom }\phi^{*} \}$. Using the continuity of $f$ from item $2)$ of Assumption~\ref{ass:basic_data}, since $t\mapsto \phi^{*}(t,0)$ is continuous on $[0,t^{*}]$ and the set $[0,t^{*}]$ is a closed interval, the set $\{x\in \mathbb{R}^n: x = \phi^{*}(t,0), t\in [0,t^{*}], (t,0)\in \dom \phi^{*} \}$ is closed. Similarly, since $t\mapsto \phi^{*}(t,1)$ is continuous on the closed interval $[t^{*}, T^{*}]$, the set $\{x\in \mathbb{R}^n: x = \phi^{*}(t,1), t\in [t^{*},T^{*}],(t,1)\in \mbox{dom }\phi^{*} \}$ is closed. Therefore, $\MO$ is closed;
\item If $t^{*} = T^{*}$, the argument follows similarly as in the first case above.
\end{itemize}

To prove its boundedness, we proceed by contradiction.
Suppose that $\MO$ is unbounded.
Then, due to $T^{*}$ being finite and $D\subset \dom g$, it follows that the flow periodic solution $\phi^{*}$
 can only escape to infinity in finite time during flows.
This further implies that $\phi^{*}$ is not complete and that its domain is not closed, which leads to a contradiction with the definition of flow periodic solution with period $T^{*}$. 

Next, we prove that $\MO$ is forward invariant.
First, each $\phi^*\in \MS_{\MH}(\MO)$ is complete,
which directly follows from the definition of flow periodic solution in Definition~\ref{def:periodic_sol}.
Then, it remains to prove that each $\phi^*\in \MS_{\MH}(\MO)$ stays in $\MO$, that is,
$\text{rge}\ \phi^*\subset \MO$.
By contradiction, suppose that $\text{rge}\ \phi^*\not\subset \MO$.
Then, by Definition~\ref{def:periodic_orbit}, the flow periodic solution $\phi^*$
defines another hybrid limit cycle $\MO'$,
which contradicts with the uniqueness of the hybrid limit cycle $\MO$.
Therefore, for each $\phi^*\in \MS_{\MH}(\MO)$, we have
$\text{rge}\ \phi^*\subset \MO$.
Hence, $\MO$ is forward invariant.
\end{proof}
}

\IfTAC{\CRed{\begin{remark}
Since a hybrid limit cycle $\MO$ to $\MH|_{M}$ is compact, for any solution $\phi$ to $\MH|_{M}$, the distance $|\phi(t,j)|_{\MO}$ is well-defined for all $(t,j)\in \dom \phi$.
\end{remark}}
}{\begin{remark}
Since a hybrid limit cycle $\MO$ to $\MH|_{M}$ is compact, for any solution $\phi$ to $\MH|_{M}$, the distance $|\phi(t,j)|_{\MO}$ is well-defined for all $(t,j)\in \dom \phi$.
\end{remark}
}

\IfTAC{
We revisit the previous example to illustrate the properties of a hybrid system $\MH$ satisfying
Assumption~\ref{ass:basic_data}.
}
{
We revisit the previous examples to illustrate the properties of a hybrid system $\MH$ satisfying
Assumption~\ref{ass:basic_data}.
}

\begin{example}\label{exam:TCP3}
Consider the congestion control system in \IfTAC{Section \ref{sec:motive_ex}}{Example~\ref{exam:TCP}}\IfTAC{.}{; see also Example~\ref{exam:TCP2}.}
By definition, the sets $C_{\textrm{\tiny TCP}}$ and $D_{\textrm{\tiny TCP}}$ of the model in \eqref{eq:TCP3}
are closed. Moreover, $f_{\textrm{\tiny TCP}}$ and $g_{\textrm{\tiny TCP}}$ are continuously differentiable.
Define the function $h:\mathbb{R}^2\to \mathbb{R}$ as $h(x) = q_{\max}-q$. Then, $ C_{\textrm{\tiny TCP}}$ and $ D_{\textrm{\tiny TCP}}$ can be written as
$ C_{\textrm{\tiny TCP}}=\{x\in\BR^{2}: h(x)\geq 0\}$ and
$ D_{\textrm{\tiny TCP}}=\{x\in\BR^{2}: h(x)=0, L_{f_{\textrm{\tiny TCP}}}h(x)\leq 0\}$, respectively.
Consider the compact set
$M_{\textrm{\tiny TCP}}:=(M_{\mathrm{T}}\cap  C_{\textrm{\tiny TCP}}) \setminus M_{1},$
where $M_{\mathrm{T}}$ is defined in \eqref{eqMT} 
and
$M_{1}=\{(q_{\max},B)\}+\varepsilon\BB^{\circ}$
%
with $\varepsilon>0$ small enough; see \Figure~\ref{tcp:region}.
We obtain that
$M_{\textrm{\tiny TCP}}\cap D_{\textrm{\tiny TCP}}=\{x\in\BR^{2}: q=q_{\max}, r\in [B+\varepsilon, B+\sqrt{2aq_{\max}}]\}$ and
for each $x\in M_{\textrm{\tiny TCP}}\cap D_{\textrm{\tiny TCP}}$, $L_{f_{\textrm{\tiny TCP}}}h(x)=B\!-\!r\!<\!0$.
Moreover,
due to the condition on parameters  $m(B\!+\!\sqrt{2aq_{\max}})\!<\!B$ (see \IfTAC{Section \ref{sec:motive_ex}}{Example~\ref{exam:TCP}}),
it can be verified that
$g_{\textrm{\tiny TCP}}(M_{\textrm{\tiny TCP}}\!\cap\! D_{\textrm{\tiny TCP}})\!\cap\! (M_{\textrm{\tiny TCP}}\!\cap\! D_{\textrm{\tiny TCP}})\!=\!\emptyset$
and
$g_{\textrm{\tiny TCP}}(M_{\textrm{\tiny TCP}}\!\cap\! D_{\textrm{\tiny TCP}}) \!\subset \! M_{\textrm{\tiny TCP}} \!\cap\! C_{\textrm{\tiny TCP}}$.
Furthermore, for any point
$x\!\in\! M_{\textrm{\tiny TCP}} \cap C_{\textrm{\tiny TCP}}$,
since the $r$ component of the flow map $f_{\textrm{\tiny TCP}}$, i.e., $\dot r\!=\!a$, is positive,
$\mathbf{T}_{M_{\textrm{\tiny TCP}}\cap C_{\textrm{\tiny TCP}}}(x) \!\cap\! \{f_{\textrm{\tiny TCP}}(x)\}\!=\! \{f_{\textrm{\tiny TCP}}(x)\} \!\neq\! \emptyset$
for each $x\in (M_{\textrm{\tiny TCP}}\cap C_{\textrm{\tiny TCP}})\setminus D_{\textrm{\tiny TCP}}$.\footnote{ $\mathbf{T}_{(M \cap C)}(x)$ denotes the tangent cone to the set $M\cap C$ at $x$; see \cite[Definition 5.12]{Goebel:book}. }
When $x \in  M_{\textrm{\tiny TCP}}\cap D_{\textrm{\tiny TCP}},$
we have $q=q_{\max}$ and $r\geq B+\varepsilon$ with $\varepsilon>0$ small enough,
which implies that $r-B>0$ and solutions from $x$ cannot be extended via flow.
By \cite[Proposition 6.10]{Goebel:book}, every maximal solution to $\MH_{\mathrm{TCP}}|_{M_{\mathrm{TCP}}}=
(M_{\textrm{\tiny TCP}}\cap C_{\textrm{\tiny TCP}},f_{\textrm{\tiny TCP}},M_{\textrm{\tiny TCP}}\cap D_{\textrm{\tiny TCP}},g_{\textrm{\tiny TCP}})$ is complete.
Therefore, Assumption~\ref{ass:basic_data} holds.
Moreover, a solution $\phi^{*}$ to $\MH_{\mathrm{TCP}}|_{M_{\mathrm{TCP}}}$
from
$\phi^{*}(0,0)=(q_{\max}, 2Bm/(1+m))\in M_{\textrm{\tiny TCP}}\cap C_{\textrm{\tiny TCP}}$
is a flow periodic solution with $T^{*}=2B(1-m)/(a+ma)$. 
\end{example}

\IfTAC{}{
\begin{example}\label{exam:Izhi:revisit2}
Consider the Izhikevich neuron system introduced in Example~\ref{exam:Izhi:revisit1}.
By design, the sets $\MC_{\mathrm{I}}$ and $\MD_{\mathrm{I}}$ \SMR of the model in \eqref{eq:izhi2} \STR
are closed. Moreover,
$f_{\mathrm{I}}$ and $g_{\mathrm{I}}$ are continuously differentiable.
Define the function $h:\mathbb{R}^2\to \mathbb{R}$ as $h(x) = 30-v$. Then, $\MC_{\mathrm{I}}$ and $\MD_{\mathrm{I}}$ can be written as
$\MC_{\mathrm{I}}=\{x\in\BR^{2}: h(x)\geq 0\}$ and
$\MD_{\mathrm{I}}=\{x\in\BR^{2}: h(x)=0, L_{f_{\mathrm{I}}}h(x)\leq 0\}$, respectively.
Consider the closed set $M_{\mathrm{I}}:=\{x\in\BR^2: w\leq 325+I_{\text{ext}}\}$.
Then,
for each $x\in M_{\mathrm{I}}\cap\MD_{\mathrm{I}}$ we have
$L_{f_{\mathrm{I}}}h(x)=-f_{1}(x)= -(0.04 v^{2}+5v+140-w+I_{\text{ext}})\leq
-(326-(325+I_{\text{ext}})+I_{\text{ext}})=-1<0$.

For $\MH_{\mathrm{I}}|_{M_{\mathrm{I}}}:=(M_{\mathrm{I}}\cap\MC_{\mathrm{I}},f_{\mathrm{I}},M_{\mathrm{I}}\cap\MD_{\mathrm{I}},g_{\mathrm{I}})$,
it can be verified that
$g_{\mathrm{I}}(M_{\mathrm{I}}\cap\MD_{\mathrm{I}})\cap (M_{\mathrm{I}}\cap\MD_{\mathrm{I}})=\emptyset$
for the parameters given in \eqref{eq:Izhi:param}.  
Note that $g_{\mathrm{I}}(M_{\mathrm{I}}\cap\MD_{\mathrm{I}}) \subset M_{\mathrm{I}} \cap  C_{\mathrm{I}}$. Furthermore, for any point
$x = (\bar v,\bar w)\in M_{\mathrm{I}} \cap  C_{\mathrm{I}}$, if $x$ belongs to the boundary of the set $M_{\mathrm{I}}\cap  C_{\mathrm{I}}$ and $\bar w = 325+I_{ext}$, then $b \bar v<\bar w$ and  we have $a(b\bar v- \bar w)<0$ (recall that $b = 0.2$). This implies that the $w$ component of the flow map $f_{\mathrm{I}}$ is negative.
Then, $\mathbf{T}_{M_{\mathrm{I}}\cap  C_{\mathrm{I}}}(x) \cap \{f_{\mathrm{I}}(x)\} = \{f_{\mathrm{I}}(x)\} \neq \emptyset$ for each $x\in (M_{\mathrm{I}}\cap  C_{\mathrm{I}})\setminus D_{\mathrm{I}}.$
If $x$ belongs to the interior of the set $(M_{\mathrm{I}}\cap  C_{\mathrm{I}})\setminus D_{\mathrm{I}}$, $\mathbf{T}_{M_{\mathrm{I}}\cap  C_{\mathrm{I}}}(x) \cap \{f_{\mathrm{I}}(x)\} = \{f_{\mathrm{I}}(x)\} \neq \emptyset$.
When $x \in  M_{\mathrm{I}}\cap \MD_{\mathrm{I}}$, $f_{1}(x)>0$, solutions cannot be extended via flow.
By \cite[Proposition 6.10]{Goebel:book}, every maximal solution to $\MH_{\mathrm{I}}|_{M_{\mathrm{I}}}$ is complete.
Therefore, the neuron system $\MH_{\mathrm{I}}$ on $\BR^{2}$ and $M_{\mathrm{I}}$ satisfy Assumption~\ref{ass:basic_data}
and, as pointed out in Example~\ref{exam:Izhi:revisit1},
$\MH_{\mathrm{I}}|_{M_{\mathrm{I}}}$ has a flow periodic solution $\phi^{*}$ with period $T^{*}$,
which defines a unique hybrid limit cycle
$\MO\subset M_{\mathrm{I}}\cap \MC_{\mathrm{I}}$.
\end{example}
}

\IfTAC{}{\SMR
\begin{example}\label{exam:Compass}
Consider a compass gait biped in \cite{Goswami:1997,Gritli:2014} which
consists of a double pendulum with point masses $m_h$ and $m$ concentrated
at the hip and legs (a stance leg and a swing leg), respectively, as shown in \Figure~\ref{compass:fig0}.
The two legs are modeled as rigid bars without knees and feet, and
with a frictionless hip. Given adequate initial conditions, the compass gait biped robot
is powered only by gravity and
performs a passive walk for a given constant slope $\phi$ 
without any external intervention.
The movement of the compass gait biped
is mainly composed of two phases: a swing phase and an impact phase.
In the first case,
the biped is modeled as a double pendulum.
The latter case occurs when the swing leg strikes the ground and the stance leg
leaves the ground.
Therefore, the compass gait biped can be modeled as a hybrid system to describe the continuous and discrete dynamics
of the system.
\begin{figure}[ht]
\vspace{3mm}
\centering
\psfrag{TT}[][][1.0]{$t$}
\includegraphics[width=\figwidth]{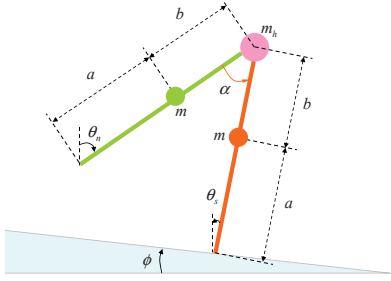}
\vspace{-3mm}
\caption{
A compass gait biped robot on a slop.}
\label{compass:fig0}
\end{figure}

The movement of the legs and hip during each step is described by the continuous dynamics of the
model, while the discrete dynamics describe the instantaneous change that occurs upon the impact at the end of each
step. At all times, one of the legs is the stance leg while the other is the swing leg, and they switch roles upon each step.
The state component vector $x$ is comprised of the angle vector $q$, which contains
the nonsupport (swing) leg angle $\theta_{n}$ and
the stance leg angle $\theta_{s}$; the velocity vector
$\dot{q}$, which contains the nonsupport (swing) leg angular velocity $\dot{\theta}_{n}$
and the stance leg angular velocity $\dot{\theta}_{s}$.
A complete hybrid system model of the compass gait biped, denoted $\MH_{\mathrm{B}}$ is defined as
\begin{equation*}\label{gait:eq1}
{\small
\MH_{\mathrm{B}}: \!
\left\{
\begin{array}{lcllllll}
\!
\dot{x}\!=\! f_{\mathrm{B}}(x):=\!\!
\begin{bmatrix}
\dot{q}\\
-\MM(q)^{-1}(\MN(q,\dot q)\dot q + \MG(q))
\end{bmatrix}  && \!\!\!\! x\in C_{\mathrm{B}}\\[1em]
\!\!
x^{+}\!=\! g_{\mathrm{B}}(x):=\!
\begin{bmatrix}
\Lambda q\\
Q_p(\alpha)^{-1}Q_m(\alpha)\dot q
\end{bmatrix}  && \!\!\!\! x\in D_{\mathrm{B}}
\end{array}
\right.
}
\end{equation*}
where $x=(q,\dot{q})=(\theta_{n},\theta_{s},\dot\theta_{n},\dot\theta_{s})\in\BR^{4}$ is the state, $\alpha  = {\theta_n} - {\theta_s}$,
$h(x):=\cos(\theta_{s}+\phi)-\cos(\theta_{n}+\phi)$,
$C_{\mathrm{B}}=\{x\in \BR^{4}: h(x)\geq 0\}$,
$D_{\mathrm{B}}=\{x\in \BR^{4}: h(x)=0, L_{f}h(x)\leq 0\}$,
{\small$\Lambda=\left[ \begin{array}{l}
0\;\;1\\
1\;\;0
\end{array} \right]$}, $\gamma$ is the gravitational acceleration,
\\
${\small
{Q_m}(\alpha ) = \left[ {\begin{array}{*{20}{c}}
{ - mab}&{ - mab + ({m_h}{l^2} + 2mal)\cos (\alpha )}\\
0&{ - mab}
\end{array}} \right]}$,\\
${\small {Q_p}(\alpha ) = \left[\!\! {\begin{array}{*{20}{c}}
{m{b^2} \!\!-\!\! mbl\cos (\alpha )}&{(m \!+\! {m_h}){l^2} \!+\! m{a^2} \!-\! mbl\cos (\alpha )}\\
{m{b^2}}&{ - mbl\cos (\alpha )}
\end{array}} \!\!\right]}$.\\ 
 $\MM(q)\in \BR^{2\times 2}$, $\MN(q)\in \BR^{2\times 2}$, and $\MG(q)\in \BR^{2}$ denotes
the inertia matrix, the matrix with centrifugal coefficients, and the vector of gravitational torques, which
are given as

$\MM(q) = \left[ {\begin{array}{*{20}{c}}
{m{b^2}}&{ - mlb\cos ({\theta_s} - {\theta_n})}\\
{ - mlb\cos ({\theta_s} - {\theta_n})}&{({m_h} + m){l^2} + m{a^2}}
\end{array}} \right],$

$\MN(q,\dot q) = \left[ {\begin{array}{*{20}{c}}
0&{ - mlb\sin (\alpha ){{\dot \theta }_s}}\\
{mlb\sin (\alpha ){{\dot \theta }_n}}&0
\end{array}} \right]$,

$\MG(q) = \left[ {\begin{array}{*{20}{c}}
{mb\gamma \sin ({\theta_n})}\\
{ - ({m_h}l + ma + ml)\gamma \sin ({\theta_s})}
\end{array}} \right]$.

We are interested in the system
\eqref{gait:eq1} restricted to
the region
\begin{equation}\label{eqCMT}
M_{\mathrm{c}}:=
\BR^{4}\backslash M_{0}
\end{equation}
for given parameters $\gamma$, $m_h$, $m$, $a$, $b$ and $\phi\in (-\frac{\pi}{4},\frac{\pi}{4})^2$,
where the open set $M_{0}:=\{(q,\dot{q})\in \BR^{4}: \dot{\theta}_{s}>-\varepsilon\}$ with $\varepsilon>0$ small enough,
is ruled out to ensure that the compass gait biped walks down the slope.
Under adequate initial conditions, the compass gait biped system has a flow periodic solution with $T^{*}$. 
However, due to the nonlinear form of the flow map, an analytic expression of the solution is not available.
For instance, when 
the initial condition is taken as
$(q(0,0),\dot{q}(0,0))=(0,0,2,-0.4)$
and
the parameters are chosen as
\begin{align}\label{eq:biped:param}
&\gamma=9.81 \mathrm{m/s^2}, m_h = 12\mathrm{kg}, m = 5\mathrm{kg},\\ \nonumber
&a =b= 0.5\mathrm{m}, \phi=0.0524\mathrm{rad},
\end{align}
a numerical approximation of this solution is shown
in~\Figure~\ref{fig:biped}, which is given by
the solution $\phi^{*}$ to $\MH_{\mathrm{B}}$ from
$\phi^{*}(0,0)=(-0.33, 0.22, -0.38,-1.09)$.
This solution is flow periodic
with $T^{*}\approx 0.385\mbox{s}$.  
%
%
\begin{figure}[ht]
\vspace{3mm}
\centering
\psfrag{A1}[][][0.8]{$P_1(q)$}
\psfrag{B1}[][][0.8]{$P_1(\dot q)$}
\psfrag{A2}[][][0.8]{$P_2(q)$}
\psfrag{B2}[][][0.8]{\quad$P_2(\dot q)$}
\psfrag{Theta}[][][0.8]{$q$}
\psfrag{dTheta}[][][0.8]{$\dot q$}
\includegraphics[width=\figwidth]{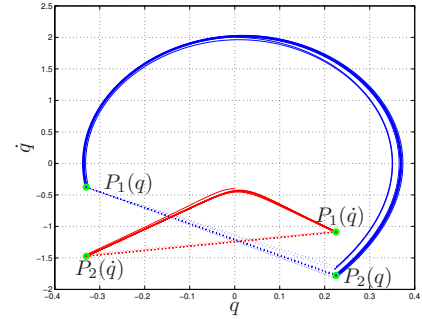}
\vspace{-3mm}
\caption{
Phase plot of a solution  
to the compass gait biped system
in \eqref{gait:eq1}
from $(0,0,2,-0.4)$. Blue: $\theta_{n}$ vs $\dot{\theta}_{n}$; Red: $\theta_{s}$ vs $\dot{\theta}_{s}$.
The point $P_{1}$ corresponds to $(q,\dot{q})=(-0.33, 0.22, -0.38,-1.09)$
and
the point $P_{2}$ corresponds to $(q,\dot{q})=(0.22, -0.33, -1.78, -1.47)$.
\Cred{A limit cycle of the compass gait biped system  is defined by the solution to \eqref{gait:eq1} from the point $P_1$ that jumps at the
point $P_2$.}}
\label{fig:biped}
\end{figure}

By design, the sets $\MC_{\mathrm{B}}$ and $\MD_{\mathrm{B}}$ are closed. Moreover,
$f_{\mathrm{B}}$ and $g_{\mathrm{B}}$ are continuously differentiable.
The open sets
$M_{1}:=
\{(q,\dot{q})\in M_{\mathrm{c}}:
\theta_{n}+\theta_{s}+2\phi \in (-\varepsilon, \varepsilon),
\sin(\theta_{n}+\phi)(\dot{\theta}_{n}+\dot{\theta}_{s})>-\varepsilon
\}$,
$M_{2}:=
\{(q,\dot{q})\in M_{\mathrm{c}}:
\theta_{n}+\theta_{s}+2\phi \in (-\varepsilon, \varepsilon),
\dot{\theta}_{n}>-\varepsilon
\}$,
and
$M_{3}:=\{(q,\dot{q})\in M_{\mathrm{c}}:
\theta_{n}-\theta_{s} \in (-\varepsilon, \varepsilon),
\sin(\theta_{n}+\phi)(\dot{\theta}_{n}-\dot{\theta}_{s})<\varepsilon
\}$
with $\varepsilon>0$ small enough,
are ruled out to ensure the transversality of the limit cycle,
ensure that the nonsupport lag walks down the slope at jumps (this is natural, since the nonsupport lag switches its role to the stance leg
when it is in contact with the slope),
and ensure that a jump happens only when two legs are in contact with the slope,
respectively.
Consider the closed set $M_{\mathrm{B}}:=M_{\mathrm{c}}\backslash (M_{1}\cup M_{2}\cup M_{3})$.
Then,
for each $x\in M_{\mathrm{B}}\cap\MD_{\mathrm{B}}$,
we have $\cos(\theta_{s}+\phi)=\cos(\theta_{n}+\phi)$
and $L_{f_{\mathrm{B}}}h(x)=\sin(\theta_{n}+\phi)\dot{\theta}_{n}-\sin(\theta_{s}+\phi)\dot{\theta}_{s}\leq 0$.
Hence, at this time,
we have two cases: $\theta_{s}=\theta_{n}$ 
or $\theta_{s}+\phi=-(\theta_{n}+\phi)$.
By the definition of $M_{3}$, since $L_{f_{\mathrm{B}}}h(x)=\sin(\theta_{n}+\phi)(\dot{\theta}_{n}-\dot{\theta}_{s})\geq \varepsilon>0$
if $\theta_{s}=\theta_{n}$, only the case $\theta_{s}+\phi=-(\theta_{n}+\phi)$ will happen
for each $x\in M_{\mathrm{B}}\cap\MD_{\mathrm{B}}$, which implies that
$\sin(\theta_{s}+\phi)=-\sin(\theta_{n}+\phi)$.
In addition, by the definition of $M_{1}$,
we have
\begin{equation}\label{my:eqLf}
L_{f_{\mathrm{B}}}h(x)=\sin(\theta_{n}+\phi)(\dot{\theta}_{n}+\dot{\theta}_{s})\leq -\varepsilon<0.
\end{equation}
Moreover,
by the jump map $q^{+}=\Lambda q$
for each $x\in M_{\mathrm{B}}\cap\MD_{\mathrm{B}}$, we have $\theta_{n}+\phi=-(\theta_{s}+\phi)$
and $\theta_{s}^{+}+\phi=\theta_{n}+\phi=-(\theta_{s}+\phi)=-(\theta_{n}^{+}+\phi)$.
In addition, by the definitions of $M_{0}$ and $M_{2}$,
we have
for each $x\in M_{\mathrm{B}}\cap\MD_{\mathrm{B}}$,
$\dot{\theta}_{n}+\dot{\theta}_{s}<0$
and
$\dot{\theta}_{n}^{+}+\dot{\theta}_{s}^{+}<0$, which
together with \eqref{my:eqLf} imply $\sin(\theta_{n}+\phi)=-\sin(\theta_{n}^{+}+\phi)>0$.
Then, for each $x\in M_{\mathrm{B}}\cap\MD_{\mathrm{B}}$,
we have that $\sin(\theta_{n}^{+}+\phi)<0$ and $\dot{\theta}_{n}^{+}+\dot{\theta}_{s}^{+}<0$,
which implies that
$L_{f_{\mathrm{B}}}h(x^{+})=\sin(\theta_{n}^{+}+\phi)(\dot{\theta}_{n}^{+}+\dot{\theta}_{s}^{+})>0$.
Therefore, we have
$g_{\textrm{B}}(M_{\textrm{B}}\cap D_{\textrm{B}})\cap (M_{\textrm{B}}\cap D_{\textrm{B}})=\emptyset$
and
$g_{\textrm{B}}(M_{\textrm{B}}\cap D_{\textrm{B}}) \subset M_{\textrm{B}} \cap  C_{\textrm{B}}$.
Furthermore, for any point
$x\in M_{\textrm{B}} \cap  C_{\textrm{B}}$,
since the $\theta_{s}$ component of the flow map $f_{\textrm{B}}$, i.e., $\dot{\theta}_{s}\leq -\varepsilon<0$, is negative,
$\mathbf{T}_{M_{\textrm{B}}\cap  C_{\textrm{B}}}(x) \cap \{f_{\textrm{B}}(x)\} = \{f_{\textrm{B}}(x)\} \neq \emptyset$
for each $x\in (M_{\textrm{B}}\cap  C_{\textrm{B}})\setminus D_{\textrm{B}}$.
When $x \in  M_{\textrm{B}}\cap D_{\textrm{B}},$
we have
$\cos(\theta_{s}+\phi)=\cos(\theta_{n}+\phi)$ and
$L_{f_{\mathrm{B}}}h(x)\leq -\varepsilon$
with $\varepsilon>0$ small enough,
which implies that $L_{f_{\mathrm{B}}}h(x)<0$ and solutions from $x$ cannot be extended via flow.
By \cite[Proposition 6.10]{Goebel:book}, every maximal solution to $\MH_{\mathrm{B}}|_{M_{\mathrm{B}}}$ is complete.
Therefore, the compass gait biped system $\MH_{\mathrm{B}}$ on $\BR^{4}$ and $M_{\mathrm{B}}$ satisfy Assumption~\ref{ass:basic_data}
and, 
$\MH_{\mathrm{B}}|_{M_{\mathrm{B}}}$ has a flow periodic solution $\phi^{*}$ with period $T^{*}$,
which defines a hybrid limit cycle
$\MO\subset M_{\mathrm{B}}\cap \MC_{\mathrm{B}}$.
\end{example}
\STR
}

The following result establishes a transversality
property of any hybrid limit cycle for $\MH$ restricted to $M$.\footnote{A hybrid limit cycle $\MO$ to a hybrid system $\MH$
satisfying Assumption~\ref{ass:basic_data} is transversal to $M\cap D$ if \Blue{$\MO$} 
 intersects
\SM $M\cap D$ \ST at exactly one point $\bar{x}:=\MO\cap (M\cap D)$ with the property
$L_{f}h(\bar{x})\neq 0$.
}

\begin{lemma}\label{lem:transversal}
Given a hybrid system $\MH=(C,f,D,g)$ on $\BRn$ and a closed set $M\subset\BRn$  satisfying Assumption~\ref{ass:basic_data},
suppose that $\MH|_{M} = (M\cap C, f, M\cap D, g)$ has a hybrid limit cycle $\MO\subset M\cap C$.
Then, $\MO$ is transversal to $M\cap D$.
\end{lemma}
\begin{proof}
We proceed by contradiction.
Consider the flow periodic solution $\phi^{*}$ with period $T^{*}$ that generates the hybrid limit cycle $\MO$ for $\MH|_{M}$. By definition, there exists $x^{*} \in \MO$ such that $x^{*} \in \MO\cap (M\cap\MD)$ and $\phi^{*}(t^{*},j^{*}) = x^{*}$ for some $(t^{*},j^{*}) \in \dom \phi^{*}$.
Suppose that
$\MO$ intersects $M\cap\MD$ at another point $x'\!\neq\! x^{*}$, i.e., $x'\!\in\! \MO\!\cap \! (M\!\cap\!\MD)$ and $\phi^{*}(t',j') \!=\! x'$ for some $(t',j') \in \dom \phi^{*}$.
Then, by items $1)$ and $3)$ of Assumption~\ref{ass:basic_data}, it follows that $h(x') = 0$ and $L_{f}h(x') < 0$.
Since $h$ is continuously differentiable and $f$ is continuous, $x\mapsto L_{f} h(x)$ is continuous. Then, there exists $\tilde \delta >0$ such that $L_{f}h(x) < 0$ for all $x \in x'+\tilde \delta \mathbb{B}$.
Therefore, the solution $\phi^{*}$ to $\MH|_{M}$ cannot be extended through
flow at $x'$. In fact, since $x' \in M\cap\MD$, $\phi^{*}$ will jump
immediately when it reaches $x'$. This contradicts the fact that $\phi^{*}$ has only one jump in its period $T^{*}$.
\end{proof}

\IfTAC{
}
{Asymptotic stability of $\MO$ implies a
finite-time convergence and recurrence property of the jump set from points in its basin of attraction, which is denoted as $\MB_{\MO}.$ The following result states this property.

\begin{lemma}\label{thm:recurrence}
Consider a hybrid system $\MH=(C,f,D,g)$ on $\BRn$ and a closed set $\X\subset\BRn$ satisfying Assumption~\ref{ass:basic_data}.
Suppose that
$\MH|_{M} = (M\cap C, f, M\cap D, g)$ has a flow periodic solution $\phi^{*}$ with
period $T^{*}>0$ 
that defines a
locally asymptotically stable
hybrid limit cycle $\MO\subset M\cap C$.
Then, the set $M\cap D$ is finite time attractive and recurrent
from the basin of attraction of $\MO$
-- in the sense that
for each solution $\phi$ to $\MH|_{M}$ with $\phi(0, 0)\in \MB_{\MO}$,
there exists $\{(\tilde t_i,\tilde j_i)\}_{i=0}^\infty$,  $\tilde{j}_0=0$,  $(\tilde t_i,\tilde j_i) \in \dom \phi$ with $\tilde t_i$ and $\tilde j_i$ strictly increasing and unbounded such that $|\phi(\tilde t_i,\tilde j_i)|_{M\cap D} = 0$ for all $i \in \mathbb{N}$.
\end{lemma}
\begin{proof}
First, we show that $M\cap D$ is finite time attractive with basin of attraction of $\MO$, namely, $\MB_{\MO}$.
We proceed by contradiction.
Suppose that there exists $\phi \in \MS_{\MH|_{M}}$ with
$\phi(0, 0)\in \MB_{\MO}\subset M\cap C$ such that
for any $(t,j)\in \dom \phi$, $\phi(t,j)\notin M\cap D$.
%
%
From the local attractivity of $\phi$,
there exists $\mu>0$ such that the solution $\phi$ to $\MH|_{M}$
starting from $|\phi(0,0)|_{\MO} \leq \mu$ is complete and satisfies
$\lim\limits_{t+j\rightarrow\infty}|\phi(t,j)|_{\MO}=0.$
Therefore, at least one jump will happen in finite flow time, namely, there exists $(t^{*},j^{*})\in\dom\phi$ such that $\phi(t^{*},j^{*})\in D$.
Recalling that $\phi(0, 0)\in M\cap C$ and $\phi \in \MS_{\MH|_{M}}$, by definition of solution
we have
$\phi(t^{*},j^{*})\in M\cap D$,
which leads to a contradiction.
Therefore, the set $M\cap D$ is finite time attractive from $\MB_{\MO}.$

Next, we prove that the set $M\cap D$ is recurrent
from $\MB_{\MO}$.
Suppose that there exists $\phi \in \MS_{\MH|_{M}}$ with
$\phi(0, 0)\in \MB_{\MO}\cap (M\cap D)$.
By item $3)$ of Assumption~\ref{ass:basic_data},
$g(M\cap D)\cap (M\cap D)\!=\!\emptyset$.
Then, $\phi(0,1)=g(\phi(0,0))\notin M\cap D$.
In addition, from the local asymptotic stability of $\MO$,
$\phi(0,1)\in \MB_{\MO}$. 
We claim that there exists $t'$ such that $\phi(t',1)\in M\cap D$.
Proceeding by contradiction, suppose that $\phi(t,1)\notin M\cap D$ for all $(t,1)\in\dom \phi$.
However, from the above argument, the set $M\cap D$ is finite time attractive from $\MB_{\MO}$,
which implies that from $\phi(0,1)\in \MB_{\MO}$,
$M\cap D$ is finite time attractive.
This leads to a contradiction. Therefore, there exists $t'$ such that $\phi(t',1)\in M\cap D$.
Recurrence follows by repeating the same argument and using completeness
of solutions from $\MB_{\MO}$.
\end{proof}

\begin{remark}
The finite-time convergence and recurrence property of the jump set in Lemma~\ref{thm:recurrence}
are necessary conditions
for the existence of asymptotically stable hybrid limit cycles.
This provides a way to verify the existence of an asymptotically stable hybrid limit cycle in a hybrid system.
\end{remark}
}

To state our next result, let us introduce the \emph{time-to-impact function} for hybrid \IfTAC{}{dynamical} systems as in $\MH$.
Alternative equivalent definitions can be found in \cite{Grizzle:2001} and \cite[Definition 2]{Maghenem:ACC:2020}.
In \cite{Maghenem:ACC:2020}, a \emph{minimal-time function} notion with respect to a closed
set is presented for a constrained continuous-time system, which
provides the first time that a solution starting from a given initial condition
reaches that set.
Following \cite{Grizzle:2001}, for a hybrid system $\MH = (C, f, D,g)$,
the \emph{time-to-impact function with respect to $D$} is defined by
$T_{I}: \overline{C}\cup D \RA \BR_{\geq 0}\cup \{\infty\}$,  where\footnote{In particular, when there does not exist $t\geq 0$
such that $\sol(t,x)\in D$,
we have $\{t\geq 0: \sol(t,x)\in D\}=\emptyset$, which gives
$T_{I}(x)=\infty$.}
\begin{equation}\label{eq:TI}
T_{I}(x):=\inf \{ t\geq 0: \phi(t,j)\in D,\ \phi \in \MS_{\MH}(x) \}
\end{equation}
for each $x\in \overline{C}\cup D$.

Inspired by \cite[Lemma 3]{Grizzle:2001}, we show that
the function \IfTAC{$T_{I}$}{$x \mapsto T_{I}(x)$} is continuous on a subset of $M\!\cap\! (\overline{C}\!\cup\! D)$,
as specified next. 

\begin{lemma}
\label{lem:continuity_impact_time}
Suppose a hybrid system $\MH=(C,f,D,g)$ on $\BRn$ and a closed set $M\subset\BRn$ satisfy
Assumption~\ref{ass:basic_data}.
Then, $T_I$ is continuous at points in ${\mathcal{X}}:=\{x \in M\cap C: 0 \!<\! T_{I}(x) \!<\! \infty\}$.
\end{lemma}
\IfTAC{
\noindent A proof can be found in \CBlue{\cite[Lemma 4.12]{lou:TAC:2022}}.
}{
\begin{proof} 
Let $\varepsilon>0$. Consider a flow solution $t \mapsto \phi^f(t, x^{*})$ to $\dot x = f(x)$ from $x^{*} \in M\cap C$ with $0< T_I(x^{*}) < \infty$.
From the definition of $T_I$ and the condition $0< T_I(x^{*}) < \infty$,
$\phi^f(t, x^{*})$ remains in $M\cap C$ over $[0, T_I(x^{*})]$.
Denote $\bar x := \phi^f(T_I(x^{*}),x^{*})$. Then, $\bar x\in M\cap D$. By the definition of $T_I$ and items $1)$
and $3)$ in Assumption~\ref{ass:basic_data}, we have $h(\phi^f(t,x^{*})) > 0$ for all $t\in[0,T_I(x^{*}))$.
Let $\bar \varepsilon = \frac{1}{2}\min\{T_I(x^{*}),\varepsilon\}$ and $\tilde{t}_1 := T_I(x^{*}) - \bar \varepsilon$.
Using item $2)$ of Assumption~\ref{ass:basic_data}, $\phi^f$ can be extended to the interval $[T_I(x^{*}),T_I(x^{*})+\bar \varepsilon]$,\footnote{{Note that $\phi^f$ is a solution to $\dot x = f(x)$ only, without constraints.}} if needed, by decreasing $\varepsilon$. Then, by using the property that $L_f(h(\bar x))<0$ and $h(\bar x) = 0$, we obtain\footnote{
This property follows from the fact that for each $t\in[0,T_I(x^{*}))$,
$h(\phi^f(t,x^{*})) > 0$ and for each $\bar x\in M \cap D$, $L_f(h(\bar x))<0$ and $h(\bar x) = 0$,
which imply that
there exists small enough $\bar \varepsilon>0$ such that
$h(\tilde x)<0$ with $\tilde x = \phi^f(\tilde{t}_2, x^{*})$ and $\tilde{t}_2 = T_I(x^{*})+ \bar \varepsilon$.}  $h(\tilde x)<0$,
where $\tilde x = \phi^f(\tilde{t}_2, x^{*})$ and $\tilde{t}_2 = T_I(x^{*})+ \bar \varepsilon$. From item $2)$ of Assumption~\ref{ass:basic_data},
$f$ is continuous on $M\cap \MC$ and differentiable on a neighborhood of $M\cap \MC$, which implies Lipschitz continuity of $f$.
Then, by \cite[Theorem 3.5]{Khalil:2002}, it follows that solutions to $\dot x = f(x)$ depend continuously on the initial conditions, that is,
for every $\epsilon>0$ there exists $\delta=\delta(\epsilon)>0$ such that if
$|x-x^*|<\delta$, then $|\phi^f(t,x)-\phi^f(t,x^{*})|<\epsilon$ for all $t\in\dom \phi^f$.
Therefore,
given $\epsilon=\min\{|\phi^f(\tilde{t}_1, x^{*})|_{M\cap D}, |\phi^f(\tilde{t}_2, x^{*})|_{M\cap D} \}$,
there exists $\delta>0$ such that for all $x\in M\cap C$ satisfying $|x-x^{*}|<\delta$, $|\phi^f(t,x)-\phi^f(t,x^{*})| <\epsilon$ for all $t\in[0,T_I(x^{*})+\bar \varepsilon]$. Then, it follows that $h(\phi^f(\tilde{t}_1,x))>0$, $h(\phi^f(\tilde{t}_2,x))<0$,
and by continuity of $h$, $\tilde{t}_1 < T_I(x) < \tilde{t}_2$.
Therefore, $|T_I(x) - T_I(x^{*})|< \varepsilon$, which implies $T_I$ is continuous at any $x^{*}$ such that $x^{*}\in M\cap C$ and $0<T_I(x^{*})<\infty$.
\end{proof}
}

Next, we show that the function $x \mapsto T_{I}(x)$ is also continuous on a subset of $\MO$.

\begin{lemma}\label{lem:continuityTI}
Given a hybrid system $\MH=(C,f,D,g)$ on $\BRn$ and a closed set $M\subset\BRn$ satisfying Assumption~\ref{ass:basic_data}, suppose that $\MH|_{M} = (M\cap C, f, M\cap D, g)$ has a unique hybrid limit cycle $\MO\subset M\cap C$ defined by the flow periodic solution $\phi^*$.
Then, $T_I$ is continuous on $\MO\setminus \{\phi^*(t^{*},0)\}$,
where $t^{*}$ is such that $(t^{*},0), (t^{*},1) \in \dom \phi^*$, namely, $(t^{*},0)$ is a jump time of $\phi^{*}$ and
$\phi^{*}(t^{*},0)$ is the point in $M\cap D$ at which $\phi^{*}$ jumps.
\end{lemma}
\begin{proof}
Consider a hybrid limit cycle $\MO\subset M\cap C$ defined by the flow periodic solution $\phi^*$.
For $(t^*,0), (t^*,1) \in \dom \phi^*$, we have $\phi^*(t^*,0) \in M\cap D$.
By Lemma \ref{lem:closeness_orbit}, since $\MO$ is forward invariant, for all $x\in \MO\setminus \{\phi^*(t^*,0)\}$,
there exists
$t>t^*$ such that $\phi^*(t,1)$ has a jump, which implies that $0<T_I(x)<\infty$.
By Lemma~\ref{lem:continuity_impact_time},
$T_I$ is continuous at points in $\mathcal{X}\!:=\!\{x \!\in\! M\!\cap\! C: 0 \!<\! T_{I}(x) \!<\! \infty\}$.
Then, $T_I$ is continuous on $\MO\setminus \{\phi^*(t,0)\}$.
\end{proof}

\IfTAC{\vspace{-2mm}}{}

\subsection{A Necessary Condition via Forward Invariance}
\label{sec4:B}

Following the spirit of \Cblue{the} necessary condition for existence of limit cycles in nonlinear continuous-time
systems in \cite{Ghaffari:2009}, \SM we have the following necessary condition for general hybrid systems with a hybrid limit cycle given by the zero-level set of a smooth enough function. \ST

\begin{proposition}\label{thm:necessary}
Consider a hybrid system $\MH=(C,f,$ $D,g)$ on $\BRn$
satisfying the hybrid basic conditions
with $f$ continuously differentiable.
Suppose every solution $\phi \in \MS_{\MH}$ is unique
and
there exists a hybrid limit cycle $\MO$
for $\MH$ with period $T^{*}>0$ 
\Blue{satisfying}
$$\Blue{\MO\subseteq \{x\in \BRn: p(x) = 0\},}$$
where $p: \BRn\RA \BR$ is
twice continuously differentiable on an open neighborhood $\MU$ of $\MO$.
Then, there exists \IfTAC{}{a function} $W: \BRn\RA \SM \BR_{\geq 0} \ST $ 
that is twice continuously differentiable
on $\MU$ and
\begin{equation}\label{eqn:necess00}
W(x)\geq 0 \qquad \forall x\in \MO,
\end{equation}
\begin{equation}
\langle \grad W(x), f(x)\rangle = 0 \qquad \forall x\in \MO \cap C,
\label{eqn:necess1}
\end{equation}
\begin{equation}
\langle\grad\langle \grad W(x), f (x) \rangle,f (x) \rangle= 0 \qquad \forall x\in \MO \cap C,
\label{eqn:necessx}
\end{equation}
\begin{equation}
W(g(x))-W(x) = 0 \qquad \forall x\in \MO \cap D.
\label{eqn:necess2}
\end{equation}
\Blue{Furthermore, if $p$ is such that
$p(\bar{x})\neq 0$ for some $\bar{x}\in C\cup D$, then $W$ is such that \eqref{eqn:necess00} holds with strict inequality.}
\end{proposition}
\IfTAC{
\noindent A proof can be found in \CBlue{\cite[Proposition 4.14]{lou:TAC:2022}}, where the forward invariance property of $\mathcal{O}$ in Lemma~\ref{lem:closeness_orbit} is used.
}{
\begin{proof}
Using forward invariance of $\MO$ (see Lemma~\ref{lem:closeness_orbit}),
under the assumption that
$p(x_0)=0$ for all $x_0\in \MO$ and $p$ is continuously differentiable on $\MU$,
we have that
the solution $t \mapsto \phi^f(t,x_0)$ to $\dot x = f(x)$ with $x_0 \in g(D)\cap\MO$ satisfies
\begin{equation}\label{eq:pf-Aux}
\frac{\partial p\circ\phi^f}{\partial t}(t,x_0) = 0 \qquad \forall t \in (0,T^{*}).
\end{equation}
Since forward invariance of $\MO$ implies that $\text{rge}\ \phi^f\subset \MO$,
this property leads to
\begin{equation}\label{eq:pf}
\langle \grad p(x), f (x) \rangle = 0 \qquad \forall x\in\MO \cap C.
\end{equation}

Similarly, we have that
the solution $t \mapsto \phi^f(t,x_0)$ to $\dot x = f(x)$ with $x_0 \in g(D)\cap\MO$ satisfies
\begin{equation}\label{eq:pf2-Aux}
\frac{\partial^2 p\circ\phi^f}{\partial t}(t,x_0) = 0 \qquad \forall t \in (0,T^{*}),
\end{equation}
which, since
\begin{equation}
  \begin{array}{llllcccc}
\frac{\partial}{\partial t}\frac{\partial p\circ\phi^f}{\partial t}(t,x_0)
&=&
\frac{\partial}{\partial t}
\left.\langle \grad p(x), f (x) \rangle\right|_{x = \phi^f(t,x_0)}\nonumber\\
&=&
\left.\langle\grad\langle \grad p(x), f (x) \rangle,f (x) \rangle\right|_{x = \phi^f(t,x_0)}\nonumber
  \end{array}
\end{equation}
leads to
\begin{equation}\label{eq:pf2}
\langle\grad\langle \grad p(x), f (x) \rangle,f (x) \rangle = 0 \qquad \forall x \in \MO \cap C.
\end{equation}
In fact,
if there exists $\overline{x} \in \MO \cap C$ such that
$$
\langle\grad\langle \grad p(\overline{x}), f (\overline{x}) \rangle,f (\overline{x}) \rangle \not = 0,
$$
then $\langle \grad p(\overline{x}), f (\overline{x}) \rangle \not = 0$,
which contradicts \eqref{eq:pf}.

Since $\MO$ is forward invariant and
every maximal solution to $\MH$ is unique,
for each $x\in \MO\cap D$ we have $g(x) \in \MO$. In addition, since $p$ vanishes on
$\MO$, we obtain
\begin{equation}\label{eq:pg}
p(g(x))-p(x)=0 \qquad \forall x\in \MO \cap D.
\end{equation}
%
%

Now, we use the properties above to construct a function $W$ such that
it satisfies the properties stated in the result.
\Blue{By assumption, since there exists $\bar{x}\in C\cup D$ such that $p(\bar{x})\neq 0$,}
following \cite{Ghaffari:2009},
define the function $W$ as
$$W(x) = \left( p(x)-p(\bar{x}) \right)^{\bar{n}},$$
where $\bar{n} \in \mathbb{N}\setminus\{0\}$ is an arbitrary positive even integer.
Then, $W(x)\geq 0$ and
$W$ is continuously differentiable on $\MU$.
In particular,
\Blue{since $p(x)=0$ for all $x\in \MO$ and $p(\bar{x})\neq 0$,}
we have $W(x)\Blue{> 0}$ for all $x\in \MO.$
%
%
%
Using \eqref{eq:pf} and \eqref{eq:pf2}, we have, for all $x\in\MO \cap C$,
\begin{align}
&
\langle \grad W(x), f(x) \rangle
=
\bar{n} \left( p(x) - p(\bar{x}) \right)^{\bar{n}-1}
\langle \grad p(x), f (x) \rangle =0,\nonumber
\end{align}
and
\IfConf{
\begin{align}
&\langle\grad\langle \grad W(x), f (x) \rangle,f (x) \rangle\\
= &\bar{n}(\bar{n}-1) \left( p(x) - p(\bar{x}) \right)^{\bar{n}-2} \langle \grad p(x), f (x) \rangle^{2}\nonumber\\
&  + \bar{n} \left( p(x) - p(\bar{x}) \right)^{\bar{n}-1}
\langle\grad\langle \grad p(x), f (x) \rangle,f (x) \rangle =0. \nonumber
\end{align}
}{
\begin{align}
\langle\grad\langle \grad W(x), f (x) \rangle,f (x) \rangle
&= \bar{n}(\bar{n}-1) \left( p(x) - p(\bar{x}) \right)^{\bar{n}-2} \langle \grad p(x), f (x) \rangle^{2}\nonumber\\
& \quad + \bar{n} \left( p(x) - p(\bar{x}) \right)^{\bar{n}-1}
\langle\grad\langle \grad p(x), f (x) \rangle,f (x) \rangle =0. \nonumber
\end{align}
}
Finally, using \eqref{eq:pg}, we have, for all $x \in \MO \cap D$,
\begin{align}
W(g(x))-W(x)
&= \left( p(g(x)) - p(\bar{x}) \right)^{\bar{n}} - \left( p(x) - p(\bar{x}) \right)^{\bar{n}} \nonumber\\
&= \left( p(x) - p(\bar{x}) \right)^{\bar{n}} - \left( p(x) - p(\bar{x}) \right)^{\bar{n}}
=0.
\nonumber
\end{align}
\end{proof}
}

Proposition~\ref{thm:necessary} provides a necessary condition, that by seeking for a function $W$ with the properties therein,
can be used to identify the existence of a hybrid limit cycle with period $T^*$.
In addition, as exploited in \cite[Theorem 1]{Ghaffari:2009}, it can be used to determine the stability of limit cycles
for continuous-time systems.


The following example illustrates the result in Proposition~\ref{thm:necessary}.
\begin{example}\label{exam:TCPWx}
Consider the hybrid congestion control system in Example~\ref{exam:TCP3}.
The set defined by points $(q, r)$ such that $q-\frac{(r-B)^2}{2a}=R$
with $R=q_{\max}-\frac{B^2(m-1)^2}{2a(m+1)^2}$
represents a hybrid limit cycle for $\MH_{\textrm{\tiny TCP}}$, namely,
$$\MO:=\left\{(q,r)\in M_{\textrm{\tiny TCP}}: q-\frac{(r-B)^2}{2a}=R\right\},$$
is a hybrid limit cycle. In particular, the state vector $x=(q,r)$ moves clockwise
within $\MO$ as depicted in \IfTAC{\Figure~\ref{tcp:region}}{\Figure~\ref{tcp:fig1}}.
Using the flow and jump maps,
it is verified that $\MO$ is forward invariant.
Note that when $\MO\cap D_{\textrm{\tiny TCP}}$, $q=q_{\max}$ and $r=2B/(m+1)$.
To validate \Blue{Proposition}~\ref{thm:necessary},
\IfTAC{define the continuously differentiable functions $p(x):=q-\frac{(r-B)^2}{2a}-R$,
\Blue{which satisfies $p(0)=-q_{\max}-\frac{4m}{(m+1)^2}\frac{B^2}{2a}\neq 0$,}
 and $W: \BR^2\RA \BR_{\geq 0}$ as}
{
define the continuously differentiable functions $p(x):=q-\frac{(r-B)^2}{2a}-R$,
pick the point $\bar{x}:=(0,0)\notin \MO$, \Blue{which satisfies
$p(\bar{x})=-\frac{B^2}{2a}-R=-q_{\max}-\frac{4m}{(m+1)^2}\frac{B^2}{2a}<0$,}
 and define a continuously differentiable function
$W: \BR^2\RA \BR_{\geq 0}$ as 
}\IfTAC{\begin{equation}
{\small
   W(x)
   =\Big(q-\frac{(r-B)^2}{2a}+\frac{B^2}{2a} \Big)^2 \Blue{>0 \quad \forall x\in \MO}.}
\end{equation}}{\begin{equation}
  \begin{array}{llllcccc}
   W(x)& = & \Big( \big(q-\frac{(r-B)^2}{2a}-R \big)-\big(-\frac{B^2}{2a}-R \big) \Big)^2\nonumber\\
    & = & \Big(q-\frac{(r-B)^2}{2a}+\frac{B^2}{2a} \Big)^2 \Blue{>0 \quad \forall x\in \MO}.
  \end{array}  
\end{equation}}
This function satisfies \eqref{eqn:necess1}-\eqref{eqn:necess2} since
$\langle \grad p(x), f_{\textrm{\tiny TCP}}(x) \rangle=
[1\; \frac{B-r}{a}]f_{\textrm{\tiny TCP}}(x)=r-B-(r-B)=0$
for all $x\in  C_{\textrm{\tiny TCP}}.$ Then, for all
$x\in \MO\cap M_{\textrm{\tiny TCP}}\cap C_{\textrm{\tiny TCP}},$
\IfConf{
\begin{equation}
{\small
   \langle \grad W(x), f_{\textrm{\tiny TCP}}\rangle
\!=\!
   2\Big( q-\frac{(r-B)^2}{2a}+\frac{B^2}{2a} \Big)(r-B-r+B)
\!=\!0\nonumber
}\end{equation}}
{\begin{equation}
   \langle \grad W(x), f_{\textrm{\tiny TCP}}\rangle
=
   2\Big( q-\frac{(r-B)^2}{2a}+\frac{B^2}{2a} \Big)(r-B-r+B)
=0\nonumber
\end{equation}}and
\IfTAC{
$
\langle\grad\langle \grad W(x), f_{\textrm{\tiny TCP}}(x) \rangle, f_{\textrm{\tiny TCP}}(x) \rangle =0.
$
}
{
\begin{equation}
  \begin{array}{llllcccc}
  &&
\langle\grad\langle \grad W(x), f_{\textrm{\tiny TCP}}(x) \rangle, f_{\textrm{\tiny TCP}}(x) \rangle\nonumber\\
   & = &
2 \langle \grad p(x), f_{\textrm{\tiny TCP}}(x) \rangle^{2}+
2\Big( q-\frac{(r-B)^2}{2a}+\frac{B^2}{2a} \Big) \nonumber\\
&& \times \langle\grad\langle \grad p(x), f_{\textrm{\tiny TCP}}(x) \rangle, f_{\textrm{\tiny TCP}}(x) \rangle
\nonumber\\
 & = & 0. \nonumber
  \end{array}
\end{equation}
}
Moreover,  
for all $x\in \MO\cap  M_{\textrm{\tiny TCP}}\cap D_{\textrm{\tiny TCP}}$,
using the fact that $q=q_{\max}$ and $r=2B/(m+1)$, we have
\IfTAC{
$
   W(g_{\textrm{\tiny TCP}}(x))-W(x)=0.
$
}
{
\begin{equation}
  \begin{array}{llllcccc}
   W(g_{\textrm{\tiny TCP}}(x))-W(x) & = &  \Big( q-\frac{(mr-B)^2}{2a}+\frac{B^2}{2a} \Big)^2\nonumber\\
   &&-\Big( q-\frac{(r-B)^2}{2a}+\frac{B^2}{2a} \Big)^2\nonumber\\
    &=& 0.\nonumber
  \end{array}
\end{equation}}
\end{example}

\IfTAC{}{
The following example illustrates the results in
Lemma~\ref{lem:closeness_orbit}, Lemma~\ref{lem:transversal}, Lemma \ref{lem:continuity_impact_time} as well as Proposition~\ref{thm:necessary}.
\begin{example}\label{exam:AS}
Consider a hybrid system $\MH_{\mathrm{S}} = (C_{\textrm{\tiny S}},f_{\textrm{\tiny S}},D_{\textrm{\tiny S}},g_{\textrm{\tiny S}})$ with state $x = (x_1,x_2)$ and data
\begin{equation}\label{eq:simple_ex}
  \MH_{\mathrm{S}} \left\{
  \begin{aligned}
    \dot{x} & = f_{\textrm{\tiny S}}(x):=
    \rb\begin{bmatrix}
    x_2\\
    -x_1
    \end{bmatrix}
     && x \in C_{\textrm{\tiny S}}, \\
    x^{+} & = g_{\textrm{\tiny S}}(x):=
    \begin{bmatrix}
    \rc\\
    0
    \end{bmatrix}
     && x \in D_{\textrm{\tiny S}}, \\
  \end{aligned}
  \right.
\end{equation}
where $C_{\textrm{\tiny S}}:= \{x\in \mathbb{R}^2: x_1\geq 0 \}$ and $D_{\textrm{\tiny S}}:=\{x\in \mathbb{R}^2: x_1 = 0, x_2\leq 0\}$.
The two parameters $\rb$ and $\rc$ satisfy $\rb>0$ and $\rc>0$.
Since $C_{\textrm{\tiny S}}$ and $D_{\textrm{\tiny S}}$ are closed, and the flow and jump maps are continuous with $f_{\textrm{\tiny S}}$ continuously differentiable,
the hybrid system $\MH_{\mathrm{S}}$ satisfies the hybrid basic conditions.
Note that every solution $\phi\in \MS_{\MH_{\mathrm{S}}}$ is unique.
The flow dynamics characterizes an oscillatory behavior.
In fact, a \Cblue{maximal} solution $\phi^{*}$ to $\MH_{\mathrm{S}}$ from $\phi^{*}(0,0) = (\rc, 0)$ is a unique flow periodic solution with period $T^{*} = \frac{\pi}{2\rb}$.
As depicted in~\Figure~\ref{SecondOrder:fig1},
a solution $\phi$ starting from the point $P_{0}=\{(5, 0)\}$ flows to the point $P_{1}$,
jumps to the point $P_{2}$,
flows to the point $P_{3}$,
and jumps back to the point $P_{2}$.
\begin{figure}[!ht]
\psfrag{X1}[][][1.0]{$x_1$}
\psfrag{X2}[][][1.0]{$x_2$}
\psfrag{P0}[][][1.0]{$P_0$}
\psfrag{P1}[][][1.0]{$P_1$}
\psfrag{P2}[][][1.0]{$P_2$}
\psfrag{P3}[][][1.0]{$P_3$}
\centering{
\includegraphics[width=\figwidth]{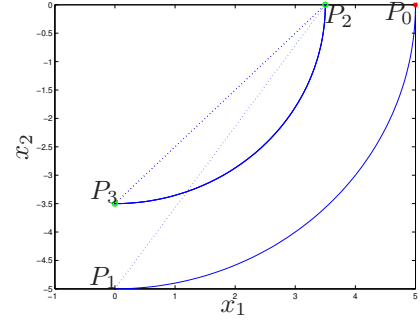}
}
\caption{
Phase plot of a solution to $\MH_{\mathrm{S}}$ from $\phi(0,0) = (5, 0)$
where $P_{0}=\{(5, 0)\}$, $P_{1} = \{(0,-5)\}$, $P_{2} = \{(c, 0)\}$ and $P_{3} = \{(0,-c)\}$.
Here, we set $\rb=0.8$ and $\rc=3.5$.
}
\label{SecondOrder:fig1}
\end{figure}

Define the function $h:\mathbb{R}^2\to \mathbb{R}$ as $h(x) = x_1$
for all $x\in \mathbb{R}^2$. Then, the sets $C_{\textrm{\tiny S}}$ and $D_{\textrm{\tiny S}}$
can be written as $C_{\textrm{\tiny S}} = \{x\in \mathbb{R}^2: h(x) \geq 0\}$ and
$D_{\textrm{\tiny S}} = \{x\in \mathbb{R}^2: h(x) =0, L_{f_{\textrm{\tiny S}}}h(x) \leq 0\}$.
Define a compact set $M_{\textrm{\tiny S}}: = \{x\in \mathbb{R}^2: |x|\geq c-\epsilon, x_2\leq 0 \}$
with $\epsilon\in (0,c)$.
Then, we obtain that for all $x\in M_{\textrm{\tiny S}}\cap D_{\textrm{\tiny S}}$, $L_{f_{\textrm{\tiny S}}}h(x)=\rb x_{2}<0$ and
$\MO \subset M_{\textrm{\tiny S}}\cap C_{\textrm{\tiny S}}$.
In addition, it is found that the invariant set defined by points $(x_1, x_2)$
such that $x_1^2+x_2^2=c^2$
represents a hybrid limit cycle $\MO$, i.e.,
$$\MO:=\left\{(x_1,x_2)\in M_{\textrm{\tiny S}}\cap C_{\textrm{\tiny S}}: x_1^2+x_2^2=c^2 \right\},$$
along which the state vector $x=(x_1,x_2)$ moves clockwise.
Using the flow and jump maps,
it is verified that $\MO$ is forward invariant.
%
Define the continuously differentiable function $p(x):=x_1^2+x_2^2-c^2$.
To validate Proposition~\ref{thm:necessary},
pick the point $\bar{x}:=(0,0)\notin \MO$ satisfying $p(\bar{x})
=-c^2\neq 0$
and define a continuously differentiable function 
$W: \BR^2\RA \BR_{\geq 0}$ 
as $\forall x\in\MO$,
$W(x)=(x_1^2+x_2^2)^2\geq 0$.
This function
satisfies \eqref{eqn:necess1}-\eqref{eqn:necess2} using the fact
$\langle \grad p(x), f_{\textrm{\tiny S}}(x) \rangle=
[2x_1\; 2x_2]f_{\textrm{\tiny S}}(x)=0$
for all $x\in  C_{\textrm{\tiny S}}$. In particular,
$\forall x\in \MO\cap M_{\textrm{\tiny S}}\cap C_{\textrm{\tiny S}},$
\begin{equation}
\langle \grad W(x), f_{\textrm{\tiny S}}(x) \rangle=2(x_1^2+x_2^2)[2x_1\; 2x_2]f_{\textrm{\tiny S}}(x)=0, \nonumber
\end{equation}
\begin{equation}
\langle\grad\langle \grad W(x), f_{\textrm{\tiny S}}(x) \rangle, f_{\textrm{\tiny S}}(x) \rangle=0, \nonumber
\end{equation}
and $\forall x\in \MO\cap  M_{\textrm{\tiny S}}\cap D_{\textrm{\tiny S}}$,
$W(g_{\textrm{\tiny S}}(x))-W(x)=0,$
where the condition $x=(0,-c)$
for $\MO\cap D_{\textrm{\tiny S}}$
is applied. Therefore, Proposition~\ref{thm:necessary} is verified.

Note that the hybrid limit cycle
$\MO:=\left\{(x_1,x_2)\in M_{\textrm{\tiny S}}\cap C_{\textrm{\tiny S}}: x_1^2+x_2^2=c^2 \right\}$
is bounded, otherwise a flow periodic solution $\phi^{*}$
will escape to infinity in finite time which leads to a contradiction with the definition of a flow periodic solution.
Moreover, due to the closeness of $M_{\textrm{\tiny S}}\cap C_{\textrm{\tiny S}}$, $\MO$ is closed; hence, $\MO$ is compact.
By the data of $\MH_{\mathrm{S}}$ in \eqref{eq:simple_ex},
one can verify that each $\phi\in \MS_{\MH}(\MO)$ is complete and satisfies $\text{rge}\ \phi\subset \MO$.
Then, $\MO$ is compact and forward invariant,
which illustrates Lemma~\ref{lem:closeness_orbit}.
In addition, since for all $x\in M_{\textrm{\tiny S}}\cap D_{\textrm{\tiny S}}$, $L_{f_{\textrm{\tiny S}}}h(x)=\rb x_{2}<0$ and
$\bar{x}=\MO\cap (M_{\textrm{\tiny S}}\cap D_{\textrm{\tiny S}})=(-c,0)$,
$\MO$ is transversal to $M_{\textrm{\tiny S}}\cap D_{\textrm{\tiny S}}$, which illustrates Lemma~\ref{lem:transversal}.
Finally, for the hybrid limit cycle $\MO$ defined by a flow periodic solution $\phi^*$
and for any $x\in \MO\setminus \{\phi^*(t,0)\}$,
$T_I(x)\in [0, \frac{\pi}{2b}]$ is continuous which illustrates Lemma~\ref{lem:continuity_impact_time}.
\end{example}
}

\section{Existence of Hybrid Limit Cycles}

In this section, we introduce a
stability notion that relates a solution to nearby solutions, which enables us to
provide sufficient conditions for the existence of
hybrid limit cycles
 for the class of hybrid systems in \eqref{sec2:eq1}. 

\subsection{Zhukovskii Stability for Hybrid Systems}\label{sec:notions}

Zhukovskii stability for a continuous-time system
consists of
the property
that, with a suitable reparametrization of perturbed trajectories, Lyapunov stability
implies Zhukovskii stability; see, e.g., \cite{Yang:2000,Leonov:2006}. We extend this notion to hybrid systems and establish links to the existence of hybrid limit cycles. To this end, inspired by \cite{Ding:2004,Yang:2000,Leonov:2006}, 
we employ
the family of maps
$\MT$ defined by
$$\MT  =\{\tau(\cdot)\! :\!\; \tau: \BR_{\geq 0} \!\RA\! \BR_{\geq 0}
\Blue{\text{ is a homeomorphism}}, \tau(0)\!=0\}.$$

A map $\tau$ in the family $\MT$
is employed to reparametrize ordinary time for the trajectories of the hybrid system \eqref{sec2:eq1} and formulate stability and attractivity notions involving
the reparametrized trajectories, as formulated next.

\begin{definition}\label{def:ZAS}
Consider a hybrid system $\MH$ on $\BRn$ as in \eqref{sec2:eq1}.
A maximal solution $\phi_1$ to $\MH$ is
said to be
\begin{itemize}
\item [1)]
\emph{Zhukovskii stable} (ZS) if for each
$\varepsilon>0$ there exists $\delta>0$ such that
for each $\phi_2 \in \MS_{\MH}(\phi_1(0,0) + \delta\BB)$
there exists $\tau \in \MT$ such that for each $(t,j)\in \dom \phi_1$ we have
$(\mtau, j) \in \dom \phi_2$ and $|\phi_1(t,j)-\phi_2(\mtau, j )| \leq \varepsilon;$
%
\item [2)]
\emph{Zhukovskii locally attractive} (ZLA) if there exists
$\mu>0$ such that
for each $\phi_2 \in \MS_{\MH}(\phi_1(0,0) + \mu\BB)$
there exists $\tau \in \MT$ such that
for each $\varepsilon>0$
there exists $T>0$ for which
    we have that
$(t,j)\in \dom \phi_1$ and $t+j\geq T$
imply
$(\mtau, j) \in \dom \phi_2$ and $|\phi_1(t,j) - \phi_2(\mtau, j )| \leq \varepsilon;$
\item [3)]
\emph{Zhukovskii locally asymptotically stable} (ZLAS)
if it is both ZS and ZLA.
\end{itemize}
\end{definition}

\begin{remark}
The map $\tau$ in Definition~\ref{def:ZAS}
reparameterizes the flow time of the solution
$\phi_2$. In particular, the ZS notion only
requires that the solution $\phi_2$ stays close to the solution $\phi_1$
for the same value of the jump counter $j$ but potentially at different flow times $t$.
Note that $\tau$ in the ZS and ZLA notions may depend on the initial conditions of $\phi_1$ and $\phi_2$.
For simplicity and for the purposes of this work, the ZLA notion is written as a uniform property, in the sense of hybrid time and over the
compact set of initial conditions defined by $\mu$.
When $\phi_1$ and each $\phi_2$ are complete, the nonuniform
version of that property would require
$$
\lim_{(t,j) \in \dom \phi_1, t + j \to \infty} |\phi_1(t,j) - \phi_2(\mtau,j)| = 0,
$$
which resembles the notion defined in the literature of continuous-time systems; see
\cite[Definition 4.1]{Yang:2000} and \cite[Definition 2]{Leonov:2006}.
\end{remark}

The ZLAS notion
will be related to existence of hybrid limit cycles
by analyzing the properties of
a Poincar\'{e} map in Section \ref{section:exist} (within the proof of Theorem~\ref{thm:exist2}) 
and the $\omega$-limit set of a hybrid arc.
\IfTAC{Next, the ZLAS notion in Definition~\ref{def:ZAS} is illustrated in an example with a hybrid limit cycle.
}
{Next, the ZLAS notion in Definition~\ref{def:ZAS} is illustrated in two examples with a hybrid limit cycle.
}

\IfTAC{}{
\begin{example}\label{ex:timer_system2}
Consider the timer system in \eqref{eq:timer}. 
Note that every maximal solution to the timer system is
unique and complete.
To verify the ZS notion, consider $\phi_1 \in \MS_{\MH_{\mathrm{T}}}$.
Given $\varepsilon>0$, let $0<\delta<\varepsilon$.
Then, for each $\phi_2 \in \MS_{\MH_{\mathrm{T}}}(\phi_1(0,0) + \delta\BB)$,
$T_I(\phi_1(0,0))=1-\phi_1(0,0)$ and $T_I(\phi_2(0,0))=1-\phi_2(0,0)$.
Without loss of generality, we further suppose
$\phi_1(0,0)>\phi_2(0,0)$. Then, the solution $\phi_1$ jumps before $\phi_2$ \Cblack{since jumps occur when the timer reaches one.}
Denote $t_{\Delta}=T_I(\phi_2(0,0))-T_I(\phi_1(0,0))=\phi_1(0,0)-\phi_2(0,0)>0$.
Now construct the map $\tau$ as
\begin{equation}\label{xxxx1}
  \tau\Ztau=\left\{
\begin{array}{lll}
    \frac{T_I(\phi_2(0,0))}{T_I(\phi_1(0,0))}t && t \in [0, T_I(\phi_1(0,0))], \\[1em]
    t+t_{\Delta} && t> T_I(\phi_1(0,0)). \\
\end{array}
  \right.
\end{equation}
Note that $\tau$
\Blue{is  a homeomorphism and} satisfies $\tau(0)=0$, hence it belongs to $\MT$, and, in addition, is continuous.
Then, for $j=0$, for each $t\in [0, T_I(\phi_1(0,0))]$,
we have $\tau\Ztau=\frac{T_I(\phi_2(0,0))}{T_I(\phi_1(0,0))}t=\frac{1-\phi_2(0,0)}{1-\phi_1(0,0)}t$,
which satisfies $(\tau\Ztau,0)\in \dom \phi_2$
and
\begin{equation}
\begin{array}{r cll}
\IfConf{&&}{}
|\phi_1(t,0)-\phi_2(\tau\Ztau, 0)|\IfConf{\\}{}
&=& |\phi_1(0,0)+t-\phi_2(0,0)-\tau\Ztau|\\
&=& |\phi_1(0,0)-\phi_2(0,0)+t-\frac{1-\phi_2(0,0)}{1-\phi_1(0,0)}t |\\
&=& |(\phi_1(0,0)-\phi_2(0,0))\big(1-\frac{t}{T_I(\phi_1(0,0))}\big)|\\
&\leq& |\phi_1(0,0)-\phi_2(0,0)|\leq \delta<\varepsilon.
\end{array}
\end{equation}
Note that $\phi_1(T_I(\phi_1(0,0)),0)=\phi_2(\tau(T_I(\phi_1(0,0))), 0)=1$.
Then,
for each $j\in \BN\backslash\{0\}$ and each $t\geq T_I(\phi_1(0,0))$ 
such that $(t,j)\in \dom \phi_1$,
we have $\mtau=t+t_{\Delta}$,
which satisfies
 $(\mtau,j)\in \dom \phi_2$
and
$|\phi_1(t,j)-\phi_2(\mtau, j)|
=0<\varepsilon.$
Therefore, the solution $\phi_{1}$ is ZS.
In fact, any solution $\phi_1 \in \MS_{\MH_{\mathrm{T}}}$ is ZS.

To verify the ZLA notion, let $\mu>0$. \Cblack{Let $\phi_1$ be
a maximal solution to $\MH_{\mathrm{T}}$.} Then, for each $\varepsilon>0$ and for each
$\phi_2 \in \MS_{\MH_{\mathrm{T}}}(\phi_1(0,0) + \mu\BB)$,
we have
$T_I(\phi_2(0,0))=1-\phi_2(0,0)$.
Similar to the above \Cblack{steps showing ZS of $\phi_1$,}
without loss of generality, assume
$\phi_1(0,0)>\phi_2(0,0)$.
Then, the solution $\phi_1$ jumps before $\phi_2$.
Note that $\phi_1(T_I(\phi_1(0,0)), 1)=\phi_2(\tau(T_I(\phi_1(0,0))), 1)=0$.
Then, for $j=1$ and for each $t\geq T_I(\phi_1(0,0))$, 
we have $\Ytau=t+t_{\Delta}$,
which satisfies $(\Ytau,1)\in \dom \phi_2$
and
$|\phi_1(t,1)-\phi_2(\Ytau, 1)|
=0<\varepsilon.$
In fact, \Cblack{for each $j\in \BN\backslash\{0\}$ and each $t\geq T_I(\phi_1(0,0))$}  
we have $\Ytau=t+t_{\Delta}$
and
\begin{align}
  (t,j)\in \dom \phi_1, \quad t+j\geq T=T_I(\phi_1(0,0))+1 \nonumber 
\end{align}
imply that
$(\mtau,j)\in \dom \phi_2$
and
$|\phi_1(t,j)-\phi_2(\mtau, j)|
=0<\varepsilon.$
Therefore, \Cblack{$\phi_{1}$ is ZLA. In fact, any solution
$\phi_1 \in \MS_{\MH_{\mathrm{T}}}$ is ZLA.
Hence, every maximal solution to $\MH_{\mathrm{T}}$ is ZLAS.}
\begin{figure}[!htbp]
\centering
\subfigure[
The projections of two solutions
from $\phi_1(0,0)=0.8$ and $\phi_2(0,0)=0$
on the $t$ direction.
]{
\label{figTimer:sub1}
\psfrag{P1}[][][1.0]{$\phi_1$}
\psfrag{P2}[][][1.0]{$\phi_2$}
\psfrag{t}[][][1.0]{$t$}
\includegraphics[width=0.27\textwidth]{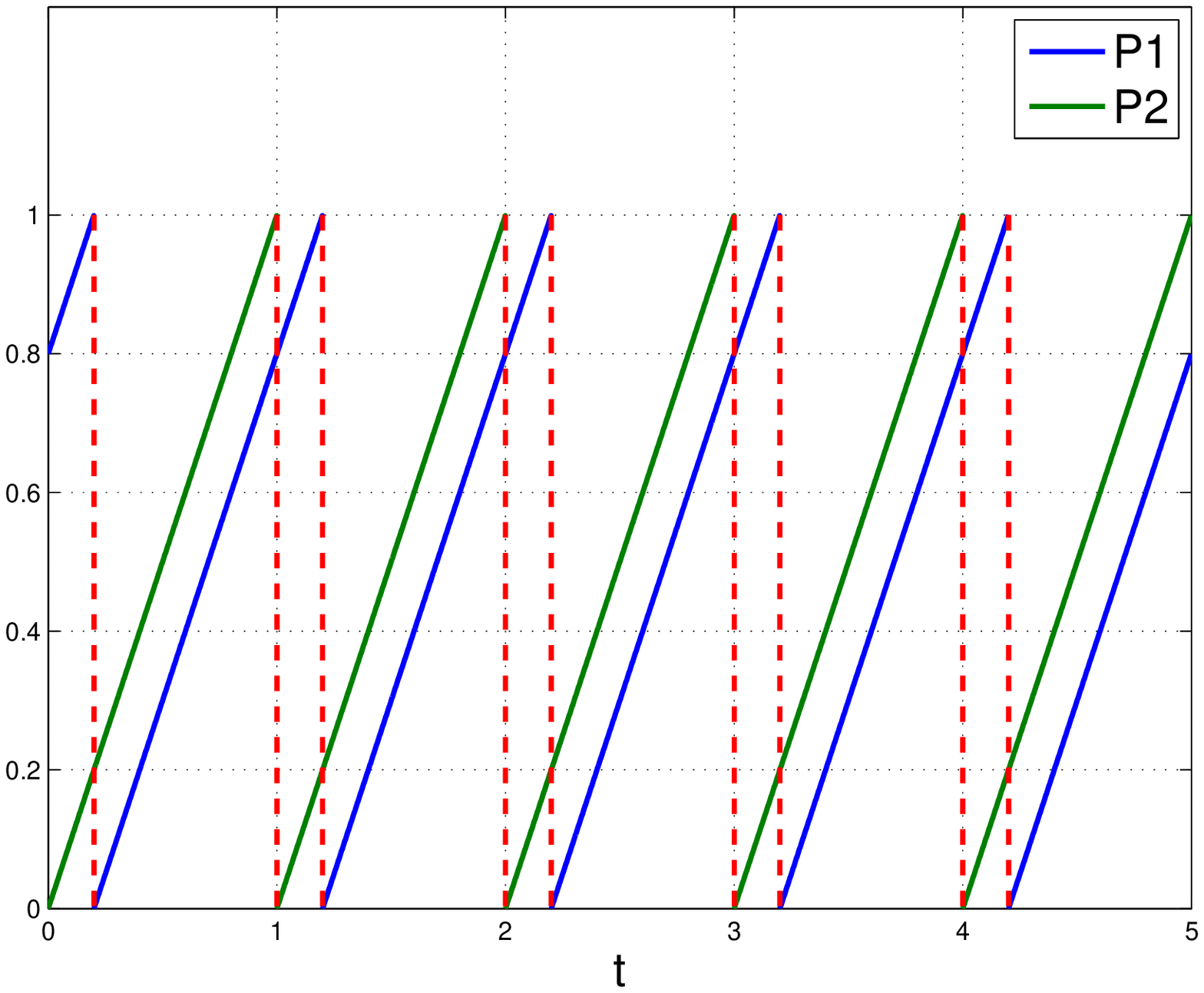}
}
\
\subfigure[
Comparison between Euclidean
distance $d(\phi_1,\phi_2)$ (top) and the distance $d_z(\phi_1,\phi_2)$
(bottom).
]{
\label{figTimer:sub2}
\psfrag{D1}[][][0.8]{$d(\phi_1,\phi_2)$}
\psfrag{D2}[][][0.8]{$d_z(\phi_1,\phi_2)$}
\psfrag{t}[][][1.0]{$t$}
\includegraphics[width=0.27\textwidth]{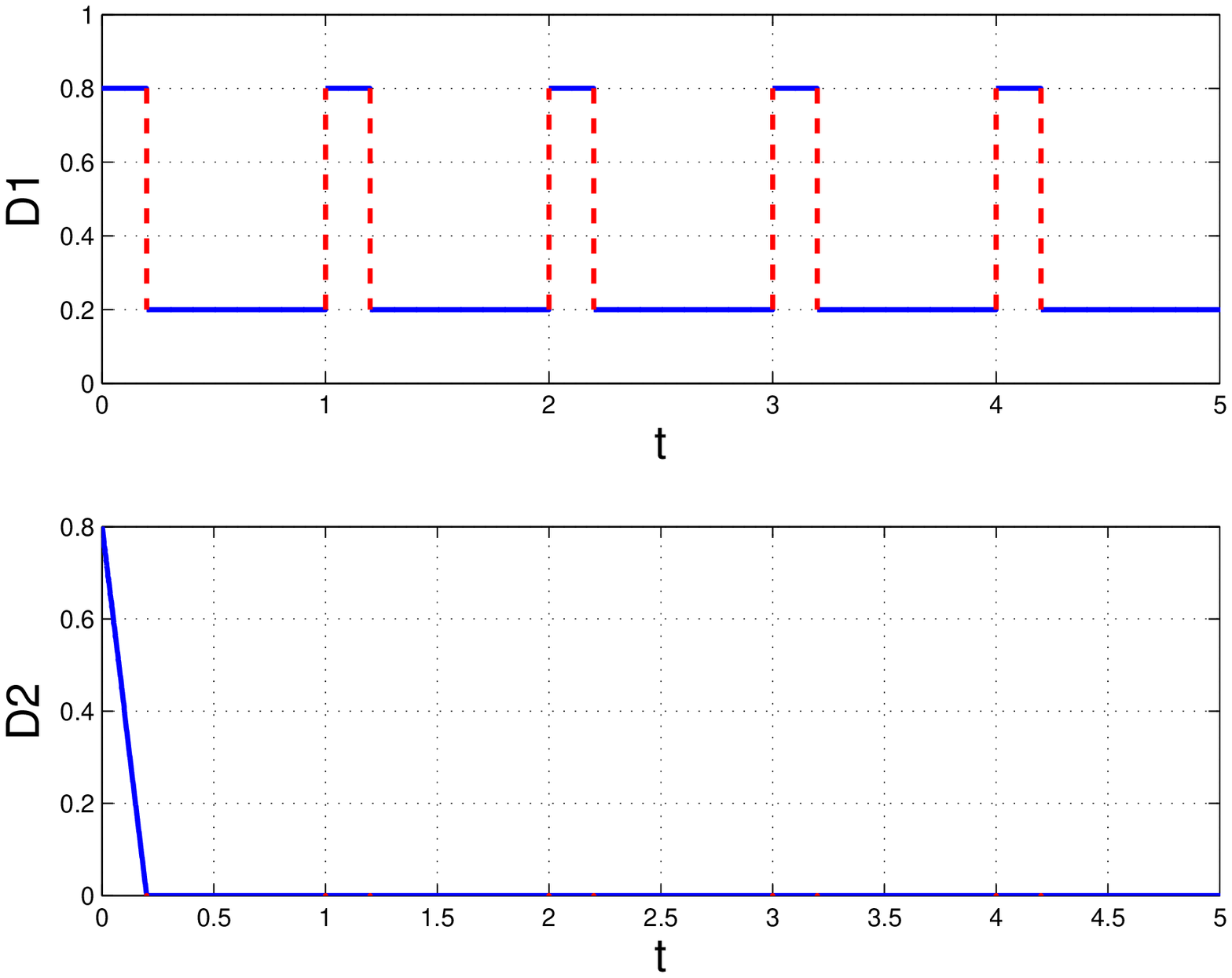}
}
\caption{
Two solutions $\phi_1$ and $\phi_2$
to the timer system in Example~\ref{ex:timer_system2}.
Unlike the Euclidean distance, which is $d(\phi_1(t,j_1),
\phi_2(t,j_2)):= |\phi_1(t,j_1)-\phi_2(t,j_2)|$ for all $(t, j_1)\in\dom\phi_1$
and $(t, j_2)\in\dom\phi_2$, which assumes
the value $0.2$ for $0.8$ seconds after every $0.2$ seconds, the distance
satisfies $d_z(\phi_1(t,j),\phi_2(\mtau,j)):=
|\phi_1(t,j)-\phi_2(\mtau,j)|$ for all
$(t,j)\in\dom\phi_1$ and $(\mtau,j)\in\dom\phi_2$, which decreases to zero after $0.2$ seconds.
Here, $\mtau=5t$ for $t\in [0, 0.2]$ and $\mtau=t+t_{\Delta}$
for $t>0.2$
with $t_{\Delta}=\phi_1(0,0)-\phi_2(0,0)=0.8$.
}
\label{fig:timer2}
\end{figure}

\Figure~\ref{fig:timer2} shows two solutions
to the timer system in \eqref{eq:timer} and the distance between them.
\Figure~\ref{figTimer:sub1} shows that
no matter how close the two maximal solutions are initialized,
the peak always exists for the Euclidean distance between
them.
However, the distance between solution $\phi_1$ and solution $\phi_2$ with parameterization $\mtau$, \Cblack{denoted $d_{z}$ and}
shown in \Figure~\ref{figTimer:sub2}, is zero a short time.

\end{example}
}

\begin{example}\label{exam:AS1}
Consider the academic system $\MH_{\mathrm{A}} = (C_{\mathrm{A}},f_{\mathrm{A}},D_{\mathrm{A}},g_{\mathrm{A}})$ with \Cblack{scalar} state $x$ and data
\begin{equation}\label{eq:TS}
  \MH_{\mathrm{A}} \left\{
  \begin{aligned}
    \dot{x} & = f_{\mathrm{A}}(x):= -ax+b
     && x \in C_{\mathrm{A}}, \\
    x^{+} & = g_{\mathrm{A}}(x):= b_{2}
     && x \in D_{\mathrm{A}}, \\
  \end{aligned}
  \right.
\end{equation}
%
where \Cblack{$C_{\mathrm{A}}:= [0, b_{1}]$} and $D_{\mathrm{A}}:=\{x\in [0, b_{1}]: x=b_{1} \}$.
The parameters $a$, $b$, $b_{1}$, and $b_{2}$ satisfy $a>0$ and $b>ab_{1}>ab_{2}>0$.
Define the compact set $M_{\mathrm{A}}:=[0, b_{1}]$
 and define a continuously differentiable function $h: M_{\mathrm{A}}\RA \BR$ as
$h(x):=b_1-x$.
Then, $C_{\mathrm{A}}$ and $D_{\mathrm{A}}$ can be rewritten as
$C_{\mathrm{A}} = \{x\in M_{\mathrm{A}}: h(x)\geq 0\}$
and $D_{\mathrm{A}}=\{ x\in C_{\mathrm{A}}: h(x)=0, L_{f_{\mathrm{A}}}h(x) \leq 0\}$, \Cblack{respectively,} where we used the property
$L_{f_{\mathrm{A}}}h(x) = -(-ax+b)= ab_{1}-b<0$ for all $x \!\in\! M_{\mathrm{A}}\cap D_{\mathrm{A}}$.
By design, the sets $C_{\mathrm{A}}$ and $D_{\mathrm{A}}$ are closed.
Moreover,
the function $f_{\mathrm{A}}$ is continuously differentiable
and the function $g_{\mathrm{A}}$ is continuous.
%
Furthermore, it can be verified that
\scalebox{0.98}{$g_{\mathrm{A}}(M_{\mathrm{A}} \!\cap\! D_{\mathrm{A}}) \!\cap\! (M_{\mathrm{A}} \!\cap\! D_{\mathrm{A}})
 \!=\! \emptyset$}.
Therefore, Assumption~\ref{ass:basic_data} holds.
Note that every maximal solution $\phi$ to $\MH_{\mathrm{A}}|_{M_{\mathrm{A}}}=
(M_{\mathrm{A}}\cap C_{\mathrm{A}}, f_{\mathrm{A}}, M_{\mathrm{A}}\cap D_{\mathrm{A}}, g_{\mathrm{A}})$ is unique via \cite[Proposition 2.11]{Goebel:book}.

To verify the ZS notion,
let us consider a maximal solution $\phi_1$ to $\MH_{\mathrm{A}}|_{M_{\mathrm{A}}}$.
For a given $\varepsilon$, let $0<\delta<\RRed{\varepsilon}$.  
Then, for each $\phi_2 \in
\MS_{\MH_{\mathrm{A}}|_{M_{\mathrm{A}}}}(\phi_1(0,0) + \delta\BB)$,
we have
$T_I(\phi_1(0,0))=
\frac{1}{a}\ln\frac{a\phi_{1}(0,0)-b}{ab_{1}-b}$
and
$T_I(\phi_2(0,0))=
\frac{1}{a}\ln\frac{a\phi_{2}(0,0)-b}{ab_{1}-b}$.
Without loss of generality, assume
$\phi_1(0,0)>\phi_2(0,0)$.
Then, the solution $\phi_1$ jumps before $\phi_2$ since jumps occur when $x$ reaches $b_1$.
%
Denote
$t_{\Delta}=T_I(\phi_2(0,0))-T_I(\phi_1(0,0))= 
\frac{1}{a}\ln\frac{a\phi_{2}(0,0)-b}{ab_{1}-b}
-
\frac{1}{a}\ln\frac{a\phi_{1}(0,0)-b}{ab_{1}-b}
=\frac{1}{a}\ln\frac{a\phi_{2}(0,0)-b}{a\phi_{1}(0,0)-b}>0$.
Let us construct $\tau$ as
\begin{equation}\label{xxxx2}
  \tau\Ztau=\left\{
\begin{array}{lll}
    \frac{T_I(\phi_2(0,0))}{T_I(\phi_1(0,0))}t && t \in [0, T_I(\phi_1(0,0))], \\[1em]
    t+t_{\Delta} && t>T_I(\phi_1(0,0)). \\
\end{array}
  \right.
\end{equation}
Note that $\tau$ \Blue{is  a homeomorphism and} satisfies $\tau(0)=0$, hence it belongs to $\MT$, and, in addition, is continuous.
Then, for $j=0$, for each $t\in [0, T_I(\phi_1(0,0))]$, 
we have $\tau\Ztau=\frac{T_I(\phi_2(0,0))}{T_I(\phi_1(0,0))}t$,
which satisfies $(\tau\Ztau,0)\in \dom \phi_2$
and
\IfConf{
\begin{equation}\label{eq:phi1phi2}
{\footnotesize
\!\!\begin{array}{r cll}
&&\!\!\!\!\!\!|\phi_1(t,0)-\phi_2(\tau\Ztau, 0)|\\
&\!\!\!\!=\!\!\!\!\!& \!\!\!\Big| \big((\phi_1(0,0)\!-\!\frac{b}{a})e^{-at}+\frac{b}{a}\big) \!-\! \big( (\phi_2(0,0)-\frac{b}{a})e^{-a \tau\Ztau}\!+\!\frac{b}{a}\big) \Big|\\
&\!\!\!\!=\!\!\!\!\!& \!\!\!\Big| \big(\phi_1(0,0)-\frac{b}{a}\big)e^{-at} - \big(\phi_2(0,0)-\frac{b}{a}\big)e^{-a \tau\Ztau} \Big|,
\end{array}
}
\end{equation}}
{
\begin{equation}\label{eq:phi1phi2}
{\begin{array}{r cll}
|\phi_1(t,0)-\phi_2(\tau\Ztau, 0)|
&\!\!\!\!=\!\!\!\!& \Big| \big((\phi_1(0,0)-\frac{b}{a})e^{-at}+\frac{b}{a}\big) - \big( (\phi_2(0,0)-\frac{b}{a})e^{-a \tau\Ztau}+\frac{b}{a}\big) \Big|\\
&\!\!\!\!=\!\!\!\!& \Big| \big(\phi_1(0,0)-\frac{b}{a}\big)e^{-at} - \big(\phi_2(0,0)-\frac{b}{a}\big)e^{-a \tau\Ztau} \Big|,
\end{array}}
\end{equation}
}
where
\IfTAC{$e^{-a \tau\Ztau}=
e^{-at}
\Big(\frac{a\phi_1(0,0)-b}{a\phi_2(0,0)-b}\Big)^{\frac{t}{T_I(\phi_1(0,0))}}.$
}{\begin{eqnarray}
e^{-a \tau\Ztau}&=&e^{-a \frac{T_I(\phi_2(0,0))}{T_I(\phi_1(0,0))}t}=
e^{-at}e^{at\big(1-\frac{T_I(\phi_2(0,0))}{T_I(\phi_1(0,0))}\big)}
\nonumber\\
&=&
e^{-at}
e^{-\frac{at}{T_I(\phi_1(0,0))}\big(T_I(\phi_2(0,0))-T_I(\phi_1(0,0))\big)}
\IfConf{\nonumber\\
&=&}{=}
e^{-at}
e^{-\frac{at}{T_I(\phi_1(0,0))}\Big(
\frac{1}{a}\ln\frac{a\phi_{2}(0,0)-b}{a\phi_{1}(0,0)-b}
\Big)}
\nonumber\\
&=&
e^{-at}
\Big(e^{-\ln\frac{a\phi_2(0,0)-b}{a\phi_1(0,0)-b}}\Big)^{\frac{t}{T_I(\phi_1(0,0))}}
\IfConf{\nonumber\\
&=&}{=}
e^{-at}
\Big(\frac{a\phi_1(0,0)-b}{a\phi_2(0,0)-b}\Big)^{\frac{t}{T_I(\phi_1(0,0))}}.\nonumber
\end{eqnarray}}
Since
$b/a>\phi_1(0,0)>\phi_2(0,0)$ and
$T_I(\phi_2(0,0))>T_I(\phi_1(0,0)),$
we have that
for each $t\in [0, T_I(\phi_1(0,0))]$,
$0<\big(\frac{a\phi_1(0,0)-b}{a\phi_2(0,0)-b}\big)\leq \big(\frac{a\phi_1(0,0)-b}{a\phi_2(0,0)-b}\big)^{t/T_I(\phi_1(0,0))}\leq 1.$
Therefore, \eqref{eq:phi1phi2} is equivalent to
\IfTAC{{ $$
|\phi_1(t,0)-\phi_2(\tau\Ztau, 0)|
\leq |\phi_1(0,0)-\phi_2(0,0)| e^{-at}\leq \delta<\varepsilon.\nonumber
$$}}{\begin{eqnarray}
\IfConf{&&}{}
|\phi_1(t,0)-\phi_2(\tau\Ztau, 0)|
\IfConf{\nonumber\\}{}
&=& \Big| \big(\phi_1(0,0)-\frac{b}{a}\big)e^{-at} - \big(\phi_2(0,0)-\frac{b}{a}\big)e^{-a \tau\Ztau} \Big|\nonumber\\
&\leq & \Big| \big(\phi_1(0,0)-\frac{b}{a}\big)e^{-at} - \big(\phi_2(0,0)-\frac{b}{a}\big)e^{-at} \Big|\nonumber\\
&\leq & |\phi_1(0,0)-\phi_2(0,0)| e^{-at}\leq \delta<\varepsilon.
\end{eqnarray}}
Note that $\phi_1(T_I(\phi_1(0,0)),0)=\phi_2(\tau(T_I(\phi_1(0,0))), 0)=b_{1}$.
%
In fact, for \Cblack{each $j\in \BN\backslash\{0\}$ and each} $t\geq T_I(\phi_1(0,0))$ 
such that $(t,j)\in \dom \phi_1$,
we have $\mtau=t+t_{\Delta}$,
which satisfies $(\mtau,j)\in \dom \phi_2$
and
$|\phi_1(t,j)-\phi_2(\mtau, j)|
=0<\varepsilon.$
Therefore, the solution $\phi_{1}$
is ZS.
In fact, any solution
$\phi_1 \in  \MS_{\MH_{\mathrm{A}}|_{M_{\mathrm{A}}}}$ is ZS.

To verify the ZLA notion, let $\mu>0$. \Cblack{Let $\phi_1$ be
a maximal solution to $\MH_{\mathrm{A}}|_{M_{\mathrm{A}}}$.}
Then, for each $\varepsilon>0$ and for each $\phi_2 \in \MS_{\MH_{\mathrm{A}}|_{M_{\mathrm{A}}}}(\phi_1(0,0) + \mu\BB)$,
we have
$T_I(\phi_1(0,0))=
\frac{1}{a}\ln\frac{a\phi_{1}(0,0)-b}{ab_{1}-b}$
and
$T_I(\phi_2(0,0))=
\frac{1}{a}\ln\frac{a\phi_{2}(0,0)-b}{ab_{1}-b}$.
Similar to the above proof of the ZS notion,
without loss of generality, assume
$\phi_1(0,0)>\phi_2(0,0)$.
Then, \Cblack{the} solution $\phi_1$ jumps before $\phi_2$.
Note that $\phi_1(T_I(\phi_1(0,0)), 1)=\phi_2(\tau(T_I(\phi_1(0,0))), 1)=b_{2}$.
Then, for $j=1$ and for each $t\geq T_I(\phi_1(0,0))$, 
we have $\Ytau=t+t_{\Delta}$,
which satisfies $(\Ytau,1)\in \dom \phi_2$
and
$|\phi_1(t,1)-\phi_2(\Ytau, 1)|
=0<\varepsilon.$
%
%
In fact,
\Cblack{for each $j\in \BN\backslash\{0\}$ and} each $t\geq T_I(\phi_1(0,0))$, 
we have that $\Ytau=t+t_{\Delta}$
and
\begin{align}
  (t,j)\in \dom \phi_1, \quad t+j\geq T=T_I(\phi_1(0,0))+1 \nonumber
\end{align}
imply that
$(\mtau,j)\in \dom \phi_2$
and
$|\phi_1(t,j)-\phi_2(\mtau, j)|
=0<\varepsilon.$
Therefore, $\phi_{1} \in \MS_{\MH_{\mathrm{A}}|_{M_{\mathrm{A}}}}$
is ZLA. \Cblack{In fact, any solution
$\phi_1 \in  \MS_{\MH_{\mathrm{A}}|_{M_{\mathrm{A}}}}$ is ZLA.
Hence, every maximal solution to $\MH_{\mathrm{A}}|_{M_{\mathrm{A}}}$ is ZLAS.}
\end{example}

Next, we establish a link between the Zhukovskii stability notion in Definition~\ref{def:ZAS}
and incremental graphical stability as introduced in \cite{Li:Sanfelice:CDC15}.
The later notion is presented next for self-containedness.
\begin{definition}\label{def:incrementalStable}
\cite[Definition 3.2]{Li:Sanfelice:CDC15}
Consider a hybrid system $\MH$ on $\BRn$ as in \eqref{sec2:eq1}.
The hybrid system $\MH$ is said to be
\begin{itemize}
\item [1)]
\emph{incrementally graphically stable} ($\delta$S) if for every
$\varepsilon>0$ there exists $\delta>0$ such that for any two maximal
solutions $\phi_1, \phi_2$ to $\MH$, $|\phi_1(0,0)-\phi_2(0,0)|\leq \delta$ implies that,
for each $(t,j)\in \dom \phi_1$, there exists
    $(s,j)\in\dom \phi_2$ satisfying $|t-s|\leq\varepsilon$ and
    \begin{equation}\label{increSta}
    |\phi_1(t,j)-\phi_2(s,j)|\leq \varepsilon;
    \end{equation}

\item [2)]
\emph{incrementally graphically locally attractive} ($\delta$LA) if there exists
$\mu>0$ such that for every $\varepsilon > 0$ and for any two maximal
solutions $\phi_1, \phi_2$ to $\MH$, $|\phi_1(0,0)-\phi_2(0,0)|\leq \mu$ implies
that there exists $T>0$ such that
for each $(t,j)\in \dom \phi_1$ such that $t+j\geq T$, there exists
    $(s,j)\in\dom \phi_2$ satisfying $|t-s|\leq\varepsilon$ and
    \begin{equation}\label{increSta2}
    |\phi_1(t,j)-\phi_2(s,j)|\leq \varepsilon;
    \end{equation}
%


\item [3)]
\emph{incrementally graphically locally asymptotically stable} ($\delta$LAS)
if it is both $\delta$S and $\delta$LA.
\end{itemize}
\end{definition}

\IfTAC{}
{
The following example is provided to illustrate
\CBlue{the $\delta$S notion in Definition~\ref{def:incrementalStable}.}

\begin{example}\label{exam:SecondOrder}
Consider the hybrid system $\MH_{\mathrm{S}}|_{M_{\textrm{\tiny S}}} = (M_{\textrm{\tiny S}}\cap C_{\textrm{\tiny S}},f_{\textrm{\tiny S}}, M_{\textrm{\tiny S}}\cap D_{\textrm{\tiny S}},g_{\textrm{\tiny S}})$ in Example~\ref{exam:AS}, with $M_{\textrm{\tiny S}}$ given therein.
We verify $\delta S$ for $\MH_{\mathrm{S}}|_{M_{\textrm{\tiny S}}}$ by definition.
\begin{figure}[!ht]
\psfrag{X1}[][][1.0]{$x_1$}
\psfrag{X2}[][][1.0]{$x_2$}
\psfrag{THE}[][][0.9]{$\theta$}
\psfrag{O}[][][0.8]{$O$}
\psfrag{P0}[][][0.8]{$Q_0$}
\psfrag{P1}[][][0.8]{$Q_1$}
\psfrag{Q0}[][][0.8]{$P_0$}
\psfrag{Q1}[][][0.8]{$P_1$}
\centering{
\IfTAC{\includegraphics[width=\figwidth]{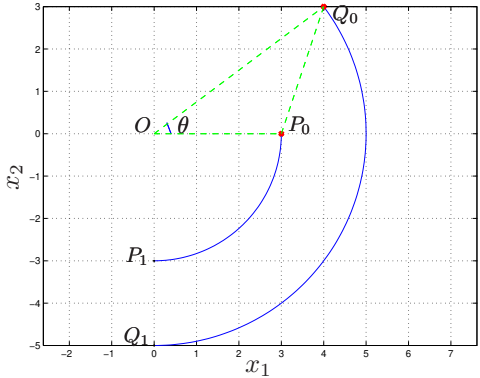}}
{\includegraphics[width=0.35\textwidth]{SecondOrderSystem_distance}}}
\vspace{-2mm}
\caption{
Phase plot of two different solutions $\phi_1$ and $\phi_2$ to $\MH_{\mathrm{S}}$
before the jump,
where $P_{0} = \{(3,0)\}$,
$P_{1} = \{(0, -3)\}$,
$Q_{0}=\{(4, 3)\}$, and $Q_{1} = \{(0, -5)\}$.
The point $P_{0}$ corresponds to $\phi_1(0,0)$
and
$Q_{0}$ corresponds to $\phi_2(0,0)$.
In this example, we have
$T_I(\phi_1(0,0))=\frac{1}{b}(\frac{\pi}{2}+\arctan\frac{0}{3})=\frac{\pi}{2b}$
and
$T_I(\phi_2(0,0))=\frac{1}{b}(\frac{\pi}{2}+\arctan\frac{3}{4})$.
Due to the form of the system, we have
$\theta=b(T_I(\phi_2(0,0))-T_I(\phi_1(0,0)))>0$.
Therefore,
$\phi_1$ jumps before $\phi_2$.
}
\label{SecondOrder:disfig}
\end{figure}
%
%
Consider $\phi_1=(\phi_{1,x_1},\phi_{1,x_2})\in  \MS_{\MH_{\mathrm{S}}|_{M_{\textrm{\tiny S}}}}$ and $\phi_2=(\phi_{2,x_1},\phi_{2,x_2}) \in  \MS_{\MH_{\mathrm{S}}|_{M_{\textrm{\tiny S}}}}$.
Given $\varepsilon>0$, let
$0<\delta<\min\{\varepsilon,\delta_{\phi},\delta_{t}\}$ be
such that $|\phi_1(0,0) - \phi_2(0,0)|\leq \delta$,
where\footnote{
\begin{itemize}
  \item To explain why we impose the constraint $\delta\leq \delta_{\phi}$,
we show a phase plot of two different solutions
in \Figure~\ref{SecondOrder:disfig}.
As we show in the example, given $\varepsilon>0$
there exists $\delta>0$ 
such that $\bar{t}'_{1}-\bar{t}_{1}\leq \varepsilon$
as illustrated in \eqref{eq:56}.
%
From \Figure~\ref{SecondOrder:disfig},
using the cosine theorem for the triangle in green, we have
$|P_0Q_0|^2=|OP_0|^2+|OQ_0|^2-2|OP_0||OQ_0|\cos\theta$.
Then, $\theta=\arccos d_{0}$, where
$d_{0}=(|OP_0|^2+|OQ_0|^2-|P_0Q_0|^2)/(2|OP_0||OQ_0|).$
Note that
$|OP_0|=|\phi_1(0,0)|$, $|OQ_0|=|\phi_2(0,0)|,$
and
$|P_0Q_0|=|\phi_1(0,0)-\phi_2(0,0)|.$
Then,
given $\varepsilon>0$,
to ensure $\bar{t}'_{1}-\bar{t}_{1}=\frac{\theta}{b}
\leq \varepsilon$, we require that
$d_{0}\geq \cos(b\varepsilon)$,
which is equivalent to
$|\phi_1(0,0)-\phi_2(0,0)|^2\leq |\phi_1(0,0)|^2+|\phi_2(0,0)|^2- 2|\phi_1(0,0)| |\phi_2(0,0)|\cos(b\varepsilon)$.
Hence, one can take $0<\delta\leq \delta_{\phi}$.
%
\item We impose the constraint $\delta\leq \delta_{t}$ since we need to show
$|\phi_1(t,1)-\phi_2(\bar{t}'_{1},1)|\leq \varepsilon$ when proving the system is $\delta$S,
as illustrated in \eqref{eq5.7:sec}.
\end{itemize}
}
{\small $$\delta_{\phi}=\sqrt{|\phi_1(0,0)|^2+|\phi_2(0,0)|^2- 2|\phi_1(0,0)| |\phi_2(0,0)|\cos(b\varepsilon)}$$}
and
{\small $$\delta_{t}=\sqrt{|\phi_1(0,0)|^2+|\phi_2(0,0)|^2- 2|\phi_1(0,0)| |\phi_2(0,0)|
\big(1-\frac{\varepsilon^2}{2c^2}\big)}.$$}
Without loss of generality, assume
$\phi_1$ jumps first. 
For each $j\in \BN\backslash\{0\}$, let
$\bar{t}_{j} = \max_{(t,j-1)\in \dom\phi_1\cap \dom\phi_2}t$ and
$\bar{t}'_{j} = \min_{(t,j)\in \dom\phi_1\cap \dom\phi_2}t$.
Then, we have that for each $t\in [0, \bar{t}_{1}]$,
$(s,0)=(t,0)\in \dom \phi_2$ and
\IfConf{
$|\phi_1(t,0)-\phi_2(s, 0)|
=|e^{At}\phi_1(0,0)-e^{At}\phi_2(0, 0)|\leq \delta<\varepsilon,$
}
{\begin{equation}
|\phi_1(t,0)-\phi_2(s, 0)|
=
|e^{At}\phi_1(0,0)-e^{At}\phi_2(0, 0)|\leq \delta<\varepsilon,
\end{equation}
}
where $A=\begin{bmatrix} 0 & b\\
-b & 0\end{bmatrix}$ and
$e^{At}=\begin{bmatrix} \cos bt & \sin bt\\
-\sin bt & \cos bt\end{bmatrix}$.

For each $t\in [\bar{t}_{1}, \bar{t}'_{1}],$
we have that
$(s,0)=(t,0)\in \dom \phi_2$ and
{\begin{equation}
{\small
\begin{array}{r cll}
|\phi_1(\bar{t}_{1},0)-\phi_2(s, 0)|
&=& |e^{A\bar{t}_{1}}\phi_1(0,0)-e^{At}\phi_2(0, 0)|\\
&\leq & |e^{A\bar{t}_{1}}\phi_1(0,0)-e^{A\bar{t}_{1}}\phi_2(0, 0) |\\
&\leq& \delta<\varepsilon.
\end{array}
}
\end{equation}}
%
Note that by the form of system \eqref{eq:simple_ex},
{\begin{equation}\label{eq:56}
{\small
\begin{array}{r cll}
\bar{t}'_{1}-\bar{t}_{1}&=&
T_I(\phi_2(0,0))-T_I(\phi_1(0,0))\\
&=&
\frac{1}{b}\arccos\frac{|\phi_1(0,0)|^2+|\phi_2(0,0)|^2-|\phi_1(0,0)-\phi_2(0, 0)|^2}{2|\phi_1(0,0)| |\phi_2(0,0)|}\\
&\leq&
\frac{1}{b}\arccos\frac{|\phi_1(0,0)|^2+|\phi_2(0,0)|^2-\delta^2}{2|\phi_1(0,0)| |\phi_2(0,0)|}\\
&\leq&
\frac{1}{b}\arccos
\frac{|\phi_1(0,0)|^2+|\phi_2(0,0)|^2-\delta_{\phi}^2}{2|\phi_1(0,0)| |\phi_2(0,0)|}
\leq \varepsilon.
\end{array}
}\end{equation}}
Then, for each $t\in [\bar{t}_{1}, \bar{t}'_{1}]$,
$|t-\bar{t}_{1}|\leq |\bar{t}'_{1}-\bar{t}_{1}|
\leq \varepsilon$.
By definition of $g_{\textrm{\tiny S}}$, we have $\phi_2(\bar{t}'_{1},1)=(c,0)$.
For each $t\in [\bar{t}_{1}, \bar{t}'_{1}]$,
$\phi_1(t,1)=e^{A(t-\bar{t}_{1})}(c,0)$.
Therefore, using \eqref{eq:56}, we have
\begin{equation}\label{eq5.7:sec}
\begin{array}{r cll}
\IfConf{&&}{}
|\phi_1(t,1)-\phi_2(\bar{t}'_{1},1)|
\IfConf{\\}{}
&=&
c\Big|
\big( 1-\cos(b(t-\bar{t}_{1})),
\sin (b(t-\bar{t}_{1})) \big)\Big| \\
&=&
\sqrt{2}c\sqrt{1-\cos(b(t-\bar{t}_{1}))}\\
&\leq&
\sqrt{2}c\sqrt{1-\cos(b(\bar{t}'_{1}-\bar{t}_{1}))}\\
&\leq&
\sqrt{2}c\sqrt{1-\frac{|\phi_1(0,0)|^2+|\phi_2(0,0)|^2-\delta^2}{2|\phi_1(0,0)| |\phi_2(0,0)|}}\\
&\leq&
\sqrt{2}c\sqrt{1-\frac{|\phi_1(0,0)|^2+|\phi_2(0,0)|^2-\delta_{t}^2}{2|\phi_1(0,0)| |\phi_2(0,0)|}}
\leq \varepsilon.
\end{array}
\end{equation}

In fact, for each $t\in [\bar{t}_{j}, \bar{t}'_{j}]$, where $j\in\BN\backslash\{0\}$,
$|t-\bar{t}_{j}|\leq |\bar{t}'_{j}-\bar{t}_{j}|
\leq \varepsilon$.  
Moreover, for each $t\in [\bar{t}'_{j}, \bar{t}_{j+1}]$, where $j\in\BN\backslash\{0\}$,
we have that $(s,j)=(t,j)\in \dom \phi_2$ and
\begin{equation}
\begin{array}{r cll}
\IfConf{&&}{}
|\phi_1(t,j)-\phi_2(s, j)|
\IfConf{\\}{}
&=&
|e^{A(t-\bar{t}'_{j})}\phi_1(\bar{t}'_{j},j)-e^{A(t-\bar{t}'_{j})}\phi_2(\bar{t}'_{j}, j)|\\
&\leq&
|\phi_1(\bar{t}'_{j},j)-\phi_2(\bar{t}'_{j}, j)|\\    
&\leq&
\sqrt{2}c\sqrt{1-\cos(b(\bar{t}'_{j}-\bar{t}_{j}))}
\leq \varepsilon,
\end{array}
\end{equation}
where we used the facts that
$\phi_1(\bar{t}'_{j},j)=e^{A(\bar{t}'_{j}-\bar{t}_{j})}\phi_1(\bar{t}_{j},j)$ and
$\phi_1(\bar{t}_{j}, j)=\phi_2(\bar{t}'_{j}, j)=(c,0)$.

For each $t\in [\bar{t}_{j+1}, \bar{t}'_{j+1}]$, where $j\in\BN\backslash\{0\}$,
we have that $(s,j)=(t,j)\in \dom \phi_2$ and
\begin{equation}
\begin{array}{r cll}
\IfConf{&&}{}
|\phi_1(\bar{t}_{j+1},j)-\phi_2(s, j)|
\IfConf{\\}{}
&\leq &
|\phi_1(\bar{t}_{j+1},j)-\phi_2(\bar{t}_{j+1}, j)|\\
&=&
|e^{A(\bar{t}_{j+1}-\bar{t}'_{j})}
(\phi_1(\bar{t}'_{j},j)-\phi_2(\bar{t}'_{j}, j))|\\
&\leq&
\sqrt{2}c\sqrt{1-\cos(b(\bar{t}'_{j}-\bar{t}_{j}))}
\leq \varepsilon.
\end{array}
\end{equation}
Therefore, the system is $\delta$S.
%
\end{example}
}

\IfTAC{}{
The following theorem establishes \CBlue{sufficient conditions for $\delta$LAS via ZLAS.}

\begin{theorem}\label{thm:LAS:ZAS}
Consider a hybrid system $\MH=(C,f,D,g)$ on $\BRn$ and a closed set $M\subset\BRn$ satisfying Assumption~\ref{ass:basic_data}.
Suppose every maximal solution $\phi$ to $\MH|_{M} = (M\cap C, f, M\cap D, g)$ is complete.
The following hold:
\vspace{-0.1em}
\begin{itemize}
\item[a)] \CRed{If each $\phi \in \MS_{\MH|_{M}}$ is ZS,
and the condition $|t-\mtau|\leq \varepsilon$
\RRed{for each $(t,j)\in\dom\phi$ and each $\phi \in \MS_{\MH|_{M}}$},
is satisfied by $\tau$ in the definition of ZS,
then the hybrid system $\MH|_{M}$ is $\delta$S;}
\item[b)] \CRed{If each $\phi \in \MS_{\MH|_{M}}$ is ZLA,
and the condition $|t-\mtau|\leq \varepsilon$
\RRed{for each $(t,j)\in\dom\phi$ and each $\phi \in \MS_{\MH|_{M}}$},
is satisfied by $\tau$ in the definition of ZLA,
then
the hybrid system $\MH|_{M}$ is $\delta$LA.}
\end{itemize}
\vspace{-0.4em}
\end{theorem}
%
\begin{proof}
Consider any solution $\phi_1 \in \MS_{\MH|_{M}}$.

a) To show \CBlue{$\delta$S}, given $\varepsilon>0$, let $0<\delta<\varepsilon$ and
consider any solution $\phi_2 \in \MS_{\MH|_{M}}(\phi_1(0,0) + \delta\BB)$.
\CBlue{Since $\phi_1 \in \MS_{\MH|_{M}}$ is ZS,
then, there exists $\tau \in \MT$ such that
for each $(t,j)\in \dom \phi_1$ we have
$(\mtau,j)\in\dom \phi_2$, $|\phi_1(t,j)-\phi_2(\mtau,j)|\leq \varepsilon.$
In addition, by assumption, we have $|t-\mtau|\leq \varepsilon$.
%
Hence, we simply take
$s:=\mtau$
and then obtain \eqref{increSta}.
Therefore, $\MH|_{M}$ is $\delta$S.}

b) \CBlue{We next show $\delta$LA.
%
Since $\phi_1 \in \MS_{\MH|_{M}}$ is ZLA,
then, there exists
$\mu>0$ such that
for each $\phi_2 \in \MS_{\MH}(\phi_1(0,0) + \mu\BB)$
there exists $\tau \in \MT$ such that
for each $\varepsilon>0$
there exists $T>0$ for which
    we have that
$(t,j)\in \dom \phi_1$ and $t+j\geq T$
imply
$(\mtau, j) \in \dom \phi_2$ and $|\phi_1(t,j) - \phi_2(\mtau, j )| \leq \varepsilon.$
In addition, by assumption, we have $|t-\mtau|\leq \varepsilon$.
Hence, we simply take
$s:=\mtau$
and then obtain \eqref{increSta2}.
Therefore, $\MH|_{M}$ is $\delta$LA.}
\end{proof}

\CRed{The sufficient conditions for $\delta$LAS in
Theorem~\ref{thm:LAS:ZAS} can be illustrated in
the hybrid systems $\MH_{\mathrm{T}}$ in Example~\ref{ex:timer_system2}
and $\MH_{\mathrm{A}}$ in Example \ref{exam:AS1},
since the ZLAS properties have been shown in
Example~\ref{ex:timer_system2} and Example \ref{exam:AS1}, respectively.}

} 

\IfTAC{}
{
\SMB
The following two examples are provided to illustrate the sufficient conditions for $\delta$LAS in
Theorem~\ref{thm:LAS:ZAS}.

\begin{example}\label{ex:timer_system3b}
Consider the timer system in \eqref{eq:timer}. 
It has been shown in Example~\ref{ex:timer_system2} that
each maximal solution is ZLAS.
Hence, by Theorem \ref{thm:LAS:ZAS}, to show the system is $\delta$LAS,
it remains to show that 
the constructed map $\tau\Ztau$ in
Example~\ref{ex:timer_system2} satisfy that $|t-\mtau|\leq \varepsilon$.

By \eqref{xxxx1}, for $j=0$, for each $t\in [0, T_I(\phi_1(0,0))]$,
we have
\begin{equation*}
\begin{array}{r cll}
|t-\mtau| &=& \Big| t-\tfrac{T_I(\phi_2(0,0))}{T_I(\phi_1(0,0))}t\Big| \\
&=& \Big| \tfrac{T_I(\phi_2(0,0))-T_I(\phi_1(0,0))}{T_I(\phi_1(0,0))} \Big| t\\
&\leq & T_I(\phi_2(0,0))-T_I(\phi_1(0,0))\\
&=&t_{\Delta}=\phi_1(0,0)-\phi_2(0,0)\leq \delta<\varepsilon.
\end{array}
\end{equation*}

On the other hand, for each $j\in \BN\backslash\{0\}$ and each $t\geq T_I(\phi_1(0,0))$ 
such that $(t,j)\in \dom \phi_1$,
we have $\mtau=t+t_{\Delta}$.
Then, it follows that
\begin{equation*}
|t-\mtau|\leq t_{\Delta}=\phi_1(0,0)-\phi_2(0,0)\leq \delta<\varepsilon.
\end{equation*}
\end{example}

\begin{example}\label{exam:AS2b}
Consider the hybrid system $\MH_{\mathrm{A}}$ in Example \ref{exam:AS1}.
It has been shown in Example~\ref{exam:AS1} that each maximal solution to $\MH_{\mathrm{A}}$ is ZLAS.
Hence, by Theorem \ref{thm:LAS:ZAS}, to show the system is $\delta$LAS,
it remains to show that 
the constructed map $\tau\Ztau$ in
Example~\ref{exam:AS1} satisfies that $|t-\mtau|\leq \varepsilon$.

From Example~\ref{exam:AS1}, given $\varepsilon>0$, let $0<\delta<\min\{\varepsilon, b\varepsilon\}$.
By \eqref{xxxx2}, for $j=0$, for each $t\in [0, T_I(\phi_1(0,0))]$,
we have
\begin{equation*}
\begin{array}{r cll}
|t-\mtau| &=& \Big| t-\tfrac{T_I(\phi_2(0,0))}{T_I(\phi_1(0,0))}t\Big| \\
&=& \Big| \tfrac{T_I(\phi_2(0,0))-T_I(\phi_1(0,0))}{T_I(\phi_1(0,0))} \Big| t\\
&\leq & T_I(\phi_2(0,0))-T_I(\phi_1(0,0))\\
&=& t_{\Delta}=\tfrac{1}{a}\ln\tfrac{a\phi_{2}(0,0)-b}{a\phi_{1}(0,0)-b}.
\end{array}
\end{equation*}

Since $a>0$ and $b/a>\phi_1(0,0)>\phi_2(0,0)\geq 0$, by using the inequality $m-1\geq \ln m$ for any $m>0$, 
it follows that
\begin{equation*}
\begin{array}{r cll}
t_{\Delta}&=&\tfrac{1}{a}\ln\tfrac{a\phi_{2}(0,0)-b}{a\phi_{1}(0,0)-b}\\
&\leq&
\tfrac{1}{a}\Big( \tfrac{a\phi_{2}(0,0)-b}{a\phi_{1}(0,0)-b}-1 \Big)\\
&=&
\tfrac{\phi_{1}(0,0)-\phi_{2}(0,0)}{b-a\phi_{1}(0,0)}\\
&\leq&
|\phi_{1}(0,0)-\phi_{2}(0,0)|/b\\
&\leq&
\delta/b<\varepsilon.
\end{array}
\end{equation*}
Therefore, we have $|t-\mtau|<\varepsilon$.

For \Cblack{each $j\in \BN\backslash\{0\}$ and each} $t\geq T_I(\phi_1(0,0))$ 
such that $(t,j)\in \dom \phi_1$,
we have $\mtau=t+t_{\Delta}$,
which satisfies $(\mtau,j)\in \dom \phi_2$.
Then, one can directly obtain $|t-\mtau|\leq t_{\Delta}\leq \delta/b<\varepsilon$.
\end{example} \STB
}

\IfTAC{\vspace{-3mm}}{}
\subsection{Existence of Hybrid Limit Cycles via Zhukovskii and Incremental Graphical Stability}\label{sec:existence}
\label{section:exist}

In this section, we present conditions for the existence of a hybrid limit cycle for hybrid systems that
are ZLAS.
The existence of such a hybrid limit cycle is related to nonemptyness of an $\omega$-limit set
and
continuity of a Poincar\'{e} map $\Gamma$ on
a closed set $\Sigma$ near an $\omega$-limit point.

\IfTAC{ \Blue{Inspired by \cite[Chapter V, Definition 2.13]{Nemytskii:book}, the following notion is introduced in a sufficiently ``short" tube
$\Phi_{\bar{t}}(U):=\{\phi_{x}(t,0):
t\mapsto \phi_{x}(t,0) \text{ is a solution to } \dot{x}=f(x)\quad x\in\BRn
\text{ from }\phi_{x}(0,0)\in U,
 t\in [0, \bar{t}],  (t,0)\in\dom\phi_{x}\}$, where $U\subset \BRn$ and $\bar{t}\geq 0$.}
}{
The importance of such a Poincar\'{e} map
is that it allows one
to determine the existence of a hybrid limit cycle, as we show in this section.
Before presenting that result,
\Blue{inspired by \cite[Chapter V, Definition 2.13]{Nemytskii:book},}
the following notion is introduced in a sufficiently ``short" tube
$\Phi_{\bar{t}}(U):=\{\phi_{x}(t,0):
t\mapsto \phi_{x}(t,0) \text{ is a solution to } \dot{x}=f(x)\quad x\in\BRn
\text{ from }\phi_{x}(0,0)\in U,
 t\in [0, \bar{t}],  (t,0)\in\dom\phi_{x}\}$, where $U\subset \BRn$ and $\bar{t}\geq 0$. 
}

\begin{definition}\label{def:sectionlocal}
\Blue{(\emph{forward local section})}
Consider a dynamical system $\dot{x}=f(x)\quad x\in\BRn.$
Given $U\subset \BRn$ and $\bar{t}\geq 0$, a closed set $\Sigma \subset \Phi_{\bar{t}}(U)$ is called a \emph{forward local section}
if for each solution $\phi_x$ to
$\dot{x}=f(x)\quad x\in\BRn$
starting from
$\phi_{x}(0) \in \Blue{U}$,
there exists a unique $t_{v}\in [0, \bar{t}]$ 
such that $t_{v}\in\dom\phi_x$ and $\phi_x(t_{v})\in \Sigma$.\footnote{Here, for the system $\dot{x}=f(x)\quad x\in\BRn$,
since it is a continuous-time system, we have $\dom\phi_{x}\subset \BR_{\geq 0}$.
Then, we write $\phi_{x}(t)$ instead of $\phi_{x}(t,0)$.}
\end{definition}

\SMR
To guarantee the existence of a \emph{forward local section},
inspired by \cite[Chapter V, Theorem 2.14]{Nemytskii:book},
we present the following result, which is different \Blue{from}
\cite[Chapter V, Theorem 2.14]{Nemytskii:book} as it only allows for forward times.  
\begin{lemma}\label{lemm:Bebutov:my}
Consider the dynamical system $\dot{x}=f(x)\quad x\in\BRn.$
\Blue{If $f$ is continuously differentiable
and $v$ is not an equilibrium point of the dynamical system,}
then, for any sufficiently small $\bar{t}>0$,
there exists $\sigma>0$ such that there exists a forward local section 
$\Sigma \subset \Phi_{\bar{t}}(v+\sigma\BB).$
\end{lemma}
\IfTAC{
\noindent A proof can be found in \CBlue{\cite[Lemma 5.11]{lou:TAC:2022}}.
}{
\begin{proof}
Consider a maximal flow solution $t \mapsto \phi^f(t, x_0)$ to $\dot x = f(x)$ from $x_0 \in \BRn$.
Obviously, $\sol(0,x_0)=x_0$.

\Red{By assumption,
$f$ is continuously differentiable on $\BRn$,
which implies Lipschitz continuity of $f$ \RRed{in $x$ on $\BRn$}
Then, by \cite[Theorem 3.2]{Khalil:2002},
for each initial condition, there exists a unique solution to the dynamical system $\dot{x}=f(x)\quad x\in\BRn$, 
and, by \cite[Theorem 3.5]{Khalil:2002},}
\Blue{it follows that solutions to $\dot x = f(x)$ depend continuously on the initial conditions.
In addition, since $v$ is not an equilibrium point of the dynamical system,
there exists $t_0>0$ such that $|v-\sol(t_{0},v)|>0$.}
Define a function $\varphi$ as
$$
\varphi(x, t)=\int_t^{t+t_0} |v-\sol(\tau,x)| \dif\tau,
$$
which is continuous with respect to $t$,
and
has the partial derivative
$$
\frac{\partial\varphi(x, t)}{\partial t}=|v-\sol(t+t_{0},x)|-|v-\sol(t,x)|.
$$
\CBlue{Note that $\frac{\partial\varphi(x, t)}{\partial t}$ is also continuous with respect to $t$.}
\Blue{Moreover,} since each solution to $\dot{x}=f(x)\quad x\in\BRn$
depends continuously on the initial condition,
\CBlue{the function $\frac{\partial\varphi(x, t)}{\partial t}$ is also continuous with respect to $x$.}

Using the fact that $\sol(0,v)=v$, we have
$$
\frac{\partial\varphi(v, t)}{\partial t}\Big|_{t=0}=|v-\sol(t_{0},v)|>0.
$$
\CBlue{In addition, since $\frac{\partial\varphi(x, t)}{\partial t}$ is continuous with respect to $x$,} 
there exists $\varepsilon>0$ such that for every $p\in v+\varepsilon\BB^{\circ}$
\begin{equation}\label{varphi:incre:eq}
\CBlue{\frac{\partial\varphi(p, t)}{\partial t}\Big|_{t=0}>0.}  
\end{equation}
\CBlue{Then, using the fact that $\frac{\partial\varphi(x, t)}{\partial t}$ is continuous with respect to $t$,
there exists $t_{c}>0$ such that $\sol(t,v)\in v+\varepsilon\BB^{\circ}$
and
\begin{equation}\label{varphi:derivative:eq1}
\frac{\partial\varphi(\sol(t,v), t)}{\partial t}>0,
\end{equation}
for all $t\in [0, t_{c}]$.}

We next prove that i) there exists $\bar{t}>0$ and $\sigma>0$ such that
for all $t\in [0, \bar{t}]$,
$\sol(t, v+\sigma\BB^{\circ}) \subset v+\varepsilon\BB$,
\Mag{($\sol(t, v+\sigma\BB^{\circ}) \subset v+\varepsilon\BB$,)}
and that ii) for each $q\in v+\sigma\BB$, there exists a unique $t_v\in [0,\bar{t}]$ such that
$\varphi(q, t_v)=\varphi(v, \bar{t}/2).$

We prove the first assertion.
First, \Blue{due to the
continuous dependence on initial conditions}, we can choose \CBlue{$\bar{t}\in [0,t_c]$} such that for all $t\in [0, \bar{t}]$
$$
\sol(t,v) \in v+\varepsilon\BB^{\circ}.
$$
Denote $\bar{t}_{1}:=\frac{\bar{t}}{2}.$ Then,
using the fact that 
\CBlue{$\varphi$ is a 
continuous function of $t$,
and using \eqref{varphi:derivative:eq1},} we have
$$
\varphi(v, \bar{t})>\varphi(v, \bar{t}_{1})>\varphi(v, 0).
$$
Next, one can choose $\eta>0$ such that
\addtocounter{equation}{0}
\IfTAC{
\begin{subequations}\label{varphi:eq1}
\begin{align}
&\sol(\bar{t}, v)+\eta\BB \subset v+\varepsilon\BB^{\circ},\tag{17a}\\
&\sol(0, v)+\eta\BB \subset v+\varepsilon\BB^{\circ},\tag{17b}
\end{align}
\end{subequations}
}{\begin{subequations}\label{varphi:eq1}
\begin{align}
&\sol(\bar{t}, v)+\eta\BB \subset v+\varepsilon\BB^{\circ},\tag{40a}\\
&\sol(0, v)+\eta\BB \subset v+\varepsilon\BB^{\circ},\tag{40b}
\end{align}
\end{subequations}
}
and such that for $q \in \sol(\bar{t}, v)+\eta\BB^{\circ}$
there holds $\varphi(q, 0)>\varphi(v, \bar{t}_{1})$ and for $q \in \sol(0, v)+\eta\BB^{\circ}$
there holds $\varphi(q, 0)<\varphi(v, \bar{t}_{1})$.
\\
Finally, we choose a number $\sigma>0$ such that 
\addtocounter{equation}{0}
\IfTAC{
\begin{subequations}\label{varphi:eq2}
\begin{align}
&\sol(\bar{t}, v+\sigma\BB) \subset \sol(\bar{t},v)+\eta\BB^{\circ},\tag{18a}\\
&\sol(0, v+\sigma\BB) \subset \sol(0,v)+\eta\BB^{\circ},\tag{18b}
\end{align}
\end{subequations}
}{
\begin{subequations}\label{varphi:eq2}
\begin{align}
&\sol(\bar{t}, v+\sigma\BB) \subset \sol(\bar{t},v)+\eta\BB^{\circ},\tag{41a}\\
&\sol(0, v+\sigma\BB) \subset \sol(0,v)+\eta\BB^{\circ},\tag{41b}
\end{align}
\end{subequations}
}
and such that for all $t\in [0, \bar{t}]$
\begin{equation}\label{varphi:eqn:1a}
\sol(t, v+\sigma\BB^{\circ}) \subset v+\varepsilon\BB.
\end{equation}

We next prove the second assertion.
For each $q\in v+\sigma\BB$,
using \eqref{varphi:eq1} and \eqref{varphi:eq2},
we have
$$\sol(\bar{t},q)\in \sol(\bar{t},v)+\eta\BB^{\circ} \subset v+\varepsilon\BB^{\circ},$$
$$\sol(0,q)\in \sol(0,v)+\eta\BB^{\circ} \subset v+\varepsilon\BB^{\circ}.$$
Then, 
we obtain
$$\varphi(q, \bar{t})>\varphi(v, \bar{t}_{1})>\varphi(q, 0).$$
%
Hence,
we have that for each $q\in v+\sigma\BB$, there exists 
$t_v\in [0,\bar{t}]$ such that
$\varphi(q, t_v)=\varphi(v, \bar{t}_{1}).$ \Blue{Therefore, for the chosen $\sigma$,
the 
time $t_v\in T_{v}:=\{ t\in [0,\bar{t}]: q\in v+\sigma\BB, \varphi(q,t_{v})=\varphi(v,\bar{t}_{1})\}$ can be used to construct the desired forward local section $\Sigma$.}
The visualized idea of the above proof can be seen in \Figure~\ref{fig:tube}. 
\begin{figure}[!ht]
\psfrag{W}[][][0.5]{\hspace{19mm}$\sol(t,x)\!\rightarrow\! \varphi(x,t)$}
\psfrag{E}[][][0.7]{$\varepsilon$}
\psfrag{Y}[][][0.4]{$\sol(0,v)$}
\psfrag{A}[][][0.4]{\hspace{-3mm}$\sol(\bar{t}_{1},v)$}
\psfrag{C}[][][0.4]{\hspace{-4mm}$\sol(\bar{t},v)$}
\psfrag{Z}[][][0.4]{\hspace{5mm}$\sol(0,q)$}
\psfrag{B}[][][0.4]{\hspace{5mm}$\sol(t_{v},q)$}
\psfrag{D}[][][0.4]{\hspace{3mm}$\sol(\bar{t},q)$}
\psfrag{P}[][][0.6]{\!\!\!$\varphi$}
\psfrag{H}[][][0.5]{$\varphi(q,0)$}
\psfrag{J}[][][0.5]{$\varphi(q,t_v)$}
\psfrag{M}[][][0.5]{\hspace{1mm}$\varphi(q,\bar{t})$}
\psfrag{I}[][][0.5]{$\varphi(v,0)$}
\psfrag{K}[][][0.5]{$\varphi(v,\bar{t}_{1})$}
\psfrag{N}[][][0.5]{\hspace{1mm}$\varphi(v,\bar{t})$}
\psfrag{R}[][][0.6]{$t_{v}$}
\psfrag{S}[][][0.6]{$\bar{t}_{1}$}
\psfrag{U}[][][0.6]{$\bar{t}$}
\psfrag{T}[][][0.6]{$t$}
\centering{\includegraphics[width=8cm]{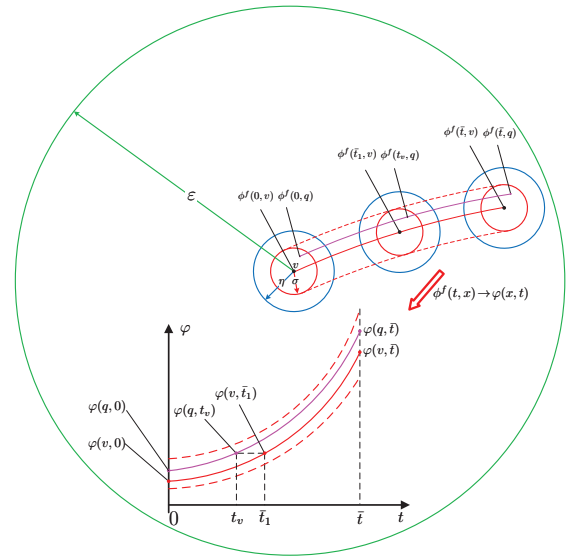}}
\caption{Visualized idea of the proof in Lemma~\ref{lemm:Bebutov:my}} 
\label{fig:tube}
\end{figure}

By the definition of $\Phi_{\bar{t}}$, since $\sol(\bar{t}_{1},v)\in \Phi_{\bar{t}}(v+\sigma\BB)$,
we have $\CBlue{\sol(t_{v},q)}\in \Phi_{\bar{t}}(v+\sigma\BB)$.
Then, the desired forward local section $\Sigma$ can be constructed as
$$\Sigma:=\{\sol(t_{v},q)\in \Phi_{\bar{t}}: q\in v+\sigma\BB\}$$ 
where $t_v\in T_{v}$. 

\CBlue{Finally, it remains to prove the uniqueness of the time $t_v\in T_{v}$.
Suppose that there exist two numbers, $t_{v}^{\prime}\in T_{v}$ and $t_{v}^{\prime \prime}\in T_{v}$,
\CRed{$t_{v}^{\prime}\neq t_{v}^{\prime \prime}$,} such that $\sol(t_{v}^{\prime},q)\in \Sigma$ and
$\sol(t_{v}^{\prime\prime},q)\in \Sigma$.
Without loss of generality, we further suppose $t_{v}^{\prime}<t_{v}^{\prime \prime}$.
Then,
$\varphi(q, t_{v}^{\prime})=\varphi(q, t_{v}^{\prime\prime})=\varphi(v,\bar{t}_{1})$.
\CRed{Note that by \eqref{varphi:derivative:eq1} and \eqref{varphi:eqn:1a},
for each $q\in v+\sigma\BB$, $\varphi(q, \cdot)$ is monotonically increasing on $T_{v}$,
which implies that 
$\frac{\partial\varphi(q, t)}{\partial t}>0$ for all $t\in [t_{v}^{\prime}, t_{v}^{\prime \prime}]$
and thus $\varphi(q, t_{v}^{\prime})<\varphi(q, t_{v}^{\prime\prime}).$
However, since $\varphi(q, t_{v}^{\prime})=\varphi(q, t_{v}^{\prime\prime})=\varphi(v,\bar{t}_{1})$, this leads to
a contradiction.}
%
Therefore, we have that $t_{v}^{\prime}=t_{v}^{\prime\prime}.$}
This completes the proof.
%
\end{proof}
}
\STR

The following result reveals the behavior of the solutions to
the flow dynamics of the hybrid system $\MH|_{M}\!=\! (M\!\cap\! C, f, M\!\cap\! D, g)$
in some neighborhood of any point in $\Blue{M\cap C}$ 
and ensures the existence of a forward local section $\Sigma$ in the tube $\Phi_{\bar{t}}$.

\begin{lemma}\label{lemm:tubular1}
Consider a hybrid system $\MH=(C,f,D,g)$ on $\BRn$ and a compact set $M\subset\BRn$ satisfying Assumption~\ref{ass:basic_data}.
Suppose that for the hybrid system $\MH|_{M} = (M\cap C, f, M\cap D, g)$,
$M\cap C$ has a nonempty interior and 
contains no equilibrium 
set for the flow dynamics
\begin{equation}\label{eq:flowdyna}
\dot{x}=f(x)\qquad x\in \Blue{M\cap C}. 
\end{equation}
For each $v\in (M\cap C)^{\circ}$ 
and a sufficiently small $\bar{t}>0$,
there exists $\sigma>0$
such that
each solution $\phi_x$ to \eqref{eq:flowdyna}
starting from
$\phi_{x}(0,0) \in \Phi_{\bar{t}}(v+\sigma\BB)$
has the following properties:
i) $\Phi_{\bar{t}}(v+\sigma\BB)\!\subset\! (M\cap C)^{\circ}$; ii)
there exists
a forward local section
$\Sigma \!\subset\! \Phi_{\bar{t}}(v+\sigma\BB)$.
\end{lemma}
\IfTAC{
\noindent
A proof can be found in \cite[Lemma 5.12]{lou:TAC:2022}.
}{
\begin{proof}
The first property, namely,
$\Phi_{\bar{t}}(v+\sigma\BB)\subset (M\cap C)^{\circ},$ follows directly from the definition
of $\Phi_{\bar{t}}(v+\sigma\BB)$, since $\Phi_{\bar{t}}(v+\sigma\BB)$ are truncated solutions
to \eqref{eq:flowdyna} starting from  $v+\sigma\BB$.
Next, we apply \Blue{Lemma~\ref{lemm:Bebutov:my}} 
to prove the second property.
Since
$M\cap C$ contains no equilibrium set for the flow dynamics \eqref{eq:flowdyna},
each $x\in M\cap C$ is not an equilibrium point.
Moreover,
from item $2)$ of Assumption~\ref{ass:basic_data},
$f$ is continuous on $M\cap \MC$ and differentiable on a neighborhood of $M\cap \MC$.
Then, by \Blue{Lemma~\ref{lemm:Bebutov:my}}, 
for each point $v\in (M\cap C)^{\circ}$ and
sufficiently small $\bar{t}>0$,
there exists $\sigma>0$
such that
each solution $\phi_x$ to \eqref{eq:flowdyna} starting from
$\phi_{x}(0,0) \in \Blue{v+\sigma\BB}$
satisfies the second property in the theorem.
The proof is complete.
\end{proof}
}

The following result
is derived via an application of
the tubular flow theorem
\IfTAC{\cite[Chapter 2, Theorem 1.1]{Palis:1982}}{(Theorem~\ref{appx:thm1} in Appendix)}
to the flow dynamics \eqref{eq:flowdyna}. 

\begin{lemma}\label{lemm:tubular}
Consider a hybrid system $\MH=(C,f,D,g)$ on $\BRn$ and a compact set $M\subset\BRn$ satisfying Assumption~\ref{ass:basic_data}.
Suppose that for the hybrid system $\MH|_{M} = (M\cap C, f, M\cap D, g)$,
\Blue{$M\cap C$ has a nonempty interior}
and contains \Red{no critical points}\footnote{\Red{
For a differential map $f: \BR^{m}\rightarrow \BRn$,
a point $x$ is a critical point of $f$
if $\frac{\partial f}{\partial x}(x)$ is not full rank and
is a regular point if $\frac{\partial f}{\partial x}(x)$ is full rank.}}
 \Red{of the map $f$}.
For any open set $U\subset M\cap C$
and
for each point $v\in U$, 
there exists an open neighborhood $\mathcal{N}_{v}\subset U$ of $v$
such that
solutions to \eqref{eq:flowdyna} from $\mathcal{N}_{v}$
are diffeomorphic to the solutions to the system
\begin{equation}\label{eq:xii}
\dot \xi_1 = 1, \ \dot \xi_i = 0 \quad \forall i \in \{2,3,\cdots,n\}
\end{equation}
on $(-1, 1)^n.$
\end{lemma}
\begin{proof}
We use the tubular flow theorem \IfTAC{\cite[Chapter 2, Theorem 1.1]{Palis:1982}}{(see Theorem~\ref{appx:thm1} in Appendix)}
to prove the result.
First, we verify the conditions of the tubular flow theorem.
Since
\Red{$M\cap C$ contains no critical points for the map $f$ in \eqref{eq:flowdyna},}
each
$x\in M\cap C$ 
is a regular point of $f$.
Moreover,
since
\Blue{$M\cap C$ has a nonempty interior and}
 $f$ is continuously differentiable by item 2) of Assumption~\ref{ass:basic_data},
$f$ is a vector field of class $\mathcal{C}^{r}$,\footnote{$\mathcal{C}^{r}$ denotes
the differentiability class of mappings having $r$ continuous derivatives.} $r\geq 1$,
on any open set $U\subset M\cap C.$
Therefore, all conditions in the tubular flow theorem
are verified.

Now, by the tubular flow theorem,
letting $v\in U$ be a regular point of $f$,
there exists an open neighborhood $\mathcal{N}_{v}\subset U$ of $v$
such that solutions to \eqref{eq:flowdyna} from $\mathcal{N}_{v}$
are diffeomorphic to the solutions to the system \eqref{eq:xii}
on $(-1, 1)^n.$
\end{proof}

The following result provides sufficient conditions for the existence of
a hybrid limit cycle of a hybrid system.\footnote{
Here, we establish sufficient conditions for the existence of
a hybrid limit cycle \emph{with multiple jumps} in each period.
A hybrid limit cycle notion allowing for multiple jumps in a period can
be defined similarly; see \IfTAC{\cite{lou.li.sanfelice16:TAC}}{\cite{lou.li.sanfelice16:TAC,Lou:ADHS15}}.
For specific systems with one jump as will be illustrated in next examples, the result is also applicable.}
In addition to technical conditions, ZLAS would serve as a sufficient condition for the existence of a
hybrid limit cycle,
which 
is motivated by the use of ZLAS for continuous-time systems in \cite{Ding:2004}.

\begin{theorem}\label{thm:exist2}
Consider a hybrid system $\MH=(C,f,D,g)$ on $\BRn$ and a compact set $M\subset\BRn$ satisfying Assumption~\ref{ass:basic_data}.
Suppose that for the hybrid system $\MH|_{M} = (M\cap C, f, M\cap D, g)$,
$M\cap C$
\Red{has a nonempty interior and contains no critical points for the map $f$},
and contains no equilibrium set for the flow dynamics \eqref{eq:flowdyna},
and
for each $x\in M\cap C$,
each maximal solution to $\MH|_{M}$ 
is 
complete with 
its hybrid time domain unbounded in the $t$ direction,
and
each solution to \eqref{eq:flowdyna} 
is not complete and ends at a point in $M\cap C$.
Then, for each solution $\phi \in \MS_{\MH|_{M}}(M\cap C)$,
$\MH|_{M}$ has a nonempty $\omega$-limit set $\Omega(\phi)$.
In addition, if 
the solution $\phi$ is ZLAS and
$\Omega(\phi)\cap (M\cap C)^{\circ}$ is nonempty,
then
$\Omega(\phi)$
is a hybrid limit cycle for $\MH|_{M}$ with period given by some
$T^{*} > 0$ and multiple jumps per period.
\end{theorem}
\begin{proof}
First, we prove nonemptyness and forward invariance of $\Omega(\phi)$.
Since the hybrid system $\MH=(C,f,D,g)$ on $\BRn$ and a compact set $M\subset\BRn$ satisfy Assumption~\ref{ass:basic_data},
every maximal solution to $\MH|_{M}$ is unique
via \cite[Proposition 2.11]{Goebel:book}
and
$\MH|_{M}$ satisfies the hybrid basic conditions.
By \cite[Theorem 6.8]{Goebel:book}, $\MH|_{M}$ is nominally well-posed.
Since each solution $\phi \in \MS_{\MH|_{M}}(M\cap C)$ is unique and
complete,
the set $M\cap C$ is forward invariant for $\MH|_{M}$.
Note that completeness of each solution $\phi$ and the compactness of $M$ 
imply that each $\phi$ is bounded.
Then, it follows from \cite[Lemma 3.3]{Sanfelice:2007} that the $\omega$-limit set $\Omega(\phi)$ is a nonempty, compact, and weakly invariant subset of $M$.

Next, we prove the existence of a \Blue{forward local section} $\Sigma$.
By assumption, let the solution $\phi \in \MS_{\MH|_{M}}(M\cap C)$ be ZLAS.
Since $\Omega(\phi)\cap (M\cap C)^{\circ}$ is nonempty, 
we can choose a point $p\in \Omega(\phi)\cap (M\cap C)^{\circ}$.
By Definition~\ref{def:LS1},
one can choose a sequence $\{ (t_{i}, j_{i}) \}_{i=1}^{\infty}$ 
such that $\lim_{i\RA\infty} t_i + j_i=\infty$
and $\lim_{i\RA\infty} \phi(t_i, j_i) = p\in \Omega(\phi)\cap (M\cap C)^{\circ}$.
Therefore,
there exist positive constants \Blue{$0<\sigma<\delta$} such that
$p+\sigma\BB\subset \phi(t_l, j_l)+\delta\BB$
and for sufficiently small $\bar{t}>0,$
$\Phi_{\bar{t}}(p+\sigma\BB)\subset \phi(t_l, j_l)+\delta\BB$
for some $(t_{l}, j_{l})\in \{ (t_{i}, j_{i}) \}_{i=1}^{\infty}$ with $(t_{l}, j_{l})\in \dom\phi$.
Then,
with the picked constants $\sigma$ and $\bar{t}$ (which can be chosen smaller if necessary),
by Lemma~\ref{lemm:tubular1},
we have
the following properties:
i) $\Phi_{\bar{t}}(p+\sigma\BB)\subset (M\cap C)^{\circ}$; ii)
there exists
a \Blue{forward local section}
$\Sigma \subset \Phi_{\bar{t}}(p+\sigma\BB)$.

Now,
to show the existence of a hybrid limit cycle,
let us introduce a Poincar\'{e} map for local structure of
hybrid systems.
Given the \Blue{forward local section}
$\Sigma \subset \Phi_{\bar{t}}(p+\sigma\BB)$, we denote
the Poincar\'{e} map as $\Gamma:\Sigma\RA \Sigma$  and define it as
\IfConf{
\begin{equation}\label{eq:PXaa}
\begin{array}{llll}
\Gamma(x):= \big\{
\psi(t,j)\in \Sigma: & \!\!\!\! \psi \in \mathcal{S}_{\MH|_{M}}(x), t>0,\\
& \!\!\!\!\!\!\!\!\!\! (t,j)\in \dom \psi \; \big\}\quad \forall x\in \Sigma.
\end{array}
\end{equation}
}{
\begin{equation}\label{eq:PXaa}
\Gamma(x):= \big\{
\psi(t,j)\in \Sigma: \psi \in \MS_{\MH|_{M}}(x), t>0, (t,j)\in \dom \psi \; \big\}\quad \forall x\in \Sigma.
\end{equation}
}

We next prove that
i) for each solution $\psi$ to $\MH|_{M}$ starting from $\psi(0,0)\in\Sigma$,
there exists $(t,j)\in\dom\psi$ with $t>0$ such that
$\psi(t,j)\in\Sigma$,
 and that ii)
the Poincar\'{e} map $\Gamma$  
has a fixed point $q\in \Sigma$.\footnote{A point $q$ is a fixed point of a Poincar\'{e} map $\Gamma: \Sigma \to \Sigma$ if
$q=\Gamma(q)$.}

We prove the first assertion. By the definition of $\Omega(\phi)$ and from the analysis above, we have the following claim.

\textbf{Claim 1:} With the solution $\phi$ and $\Blue{\sigma}$ above, there exists $(t_k, j_k) \in \{ (t_{i}, j_{i}) \}_{i=1}^{\infty}$
such that $|\phi(t_k, j_k)-p|\leq \sigma/2$
and $|\phi(t_m, j_m)-p|\leq \sigma/2$ for each $t_m\geq t_{k}$ and each $j_m\geq j_{k}.$

Let $\phi_1(0,0):=\phi(t_k, j_k)$ as above
and define $\phi_1$ as the translation of $\phi$ by $(t_k, j_k)$, which leads to
a complete solution $\phi_1$ due to completeness of $\phi$.
By assumption, the solution $\phi_1$ to $\MH|_{M}$ is ZLAS.
Then, we have the following claim by Definition \ref{def:ZAS}.

\textbf{Claim 2:} 
With $\sigma$ above,
for each $\phi_2 \in \MS_{\MH|_{M}}(\phi_1(0,0) + \delta\BB)$
there exists a function $\tau \in \MT$ such that
for $\varepsilon=\sigma/2>0$
there exists $T>0$ for which we have
$(t, j)\in \dom \phi_1,$ $t+j\geq T$ implies
that $(\mtau, j) \in \dom \phi_2$ and
$|\phi_1(t,j) - \phi_2(\mtau, j)| \leq \varepsilon=\sigma/2.$

Since
$\Sigma\subset \Phi_{\bar{t}}(p+\sigma\BB)\subset \phi_1(0,0)+\delta\BB$,
each solution $\phi_{3}$ to $\MH|_{M}$ from \textcolor{black}{$\phi_{3}(0,0)\in \Sigma$} also satisfies
\textbf{Claim 2}.
\Blue{In addition,
by assumption, since each solution to \eqref{eq:flowdyna} 
is not complete and ends at a point in $M\cap C$,
we have recurrent jumps.}
Therefore,
from \textbf{Claim 1} and \textbf{Claim 2},
for each $\phi_{3} \in \MS_{\MH|_{M}}(\Sigma)$,
there exist $\tau \in \MT$ and $(t_{m}, j_{m})\in\dom\phi_{1}$
satisfying $t_m\geq 0$, $j_m\geq 1,$ and
$t_{m}+j_{m}\geq T$
such that $(\tau(t_m), j_m) \in \dom \phi_{3},$
$|\phi_{1}(t_m, j_m)-p|\leq \sigma/2,$
and
$$|\phi_{1}(t_m, j_m) - \phi_{3}(\tau(t_m), j_m)| \leq \sigma/2,$$
which leads to
$\phi_{3}(\tau(t_m), j_m)\in \Phi_{\bar{t}}(p+\sigma\BB).$\footnote{
If this conclusion holds for $j_m=1$,
the remaining proofs will show that
$\MH|_{M}$ has a hybrid limit cycle with one jump in each period.}
In addition, there exists $(t_{m}, j_{m})\in\dom\phi_{1}$ as above such that\footnote{
Since $p\in \Omega(\phi)\cap (M\cap C)^{\circ}$
and $\Phi_{\bar{t}}(p+\sigma\BB)\subset (M\cap C)^{\circ}$,
there always exists $(t_{m}, j_{m})$ such that
$\phi_{3}(\tau(t_m), j_m)\in \Blue{(p+\sigma\BB)}\backslash D.$
In fact, if that were not the case,
for each
$(t_{m}, j_{m})\in\dom\phi_{1}$
satisfying $t_m\geq 0$, $j_m\geq 1$ and
$t_{m}+j_{m}\geq T$,
we would have $\phi_{3}(\tau(t_m), j_m)\in D.$
Since
$|\phi_{1}(t_m, j_m)-p|\leq \sigma/2$
and
$|\phi_{1}(t_m, j_m) - \phi_{3}(\tau(t_m), j_m)| \leq \sigma/2$,
and since $\sigma$ is arbitrary,
$\lim_{m\RA\infty} \phi_{3}(\tau(t_m), j_m) = p\in D$,
which contradicts with the fact $p\in \Omega(\phi)\cap (M\cap C)^{\circ}$.
}
$\phi_{3}(\tau(t_m), j_m)\!\!\in\!\! \Blue{(p+\sigma\BB)}\backslash D.$
Let $\bar{\phi}(0,0)\!\!=\!\!\phi_{3}(\tau(t_m), j_m)$
and
define $\bar{\phi}$ as the translation of $\phi_{3}$ by $(\tau(t_m), j_m)$, which leads to
a complete solution $\bar{\phi}$ due to completeness of $\phi_3$.
%
Then, $\bar{\phi}(0,0)\in \Blue{(p+\sigma\BB)}\backslash D.$
By the second property
of Lemma~\ref{lemm:tubular1}
and the definition of \Blue{forward local section} in Definition \ref{def:sectionlocal},
we have that the solution $\bar{\phi}$ to \eqref{eq:flowdyna}
reaches
the \Blue{forward local section}
$\Sigma \subset \Phi_{\bar{t}}(p+\sigma\BB)$
at a unique time $t_{p}\in [0, \bar{t}]$, 
that is, $(t_{p},0)\in\dom\bar{\phi}$ and
$\bar{\phi}(t_{p},0)\in\Sigma$,
which implies $\phi_{3}(\tau(t_m)\!+\!t_{p}, j_m)\!\!\in\!\!\Sigma$.
Therefore, the first assertion is proved.

To prove the second assertion, that is, that the Poincar\'{e} map $\Gamma$ has a fixed point $q\in \Sigma$,
first we show continuity of $\Gamma$ on $\Sigma$.
Since
$\phi_{3}(0,0)\in\Sigma$ and
$\phi_{3}(\tau(t_m)+t_{p}, j_m)\in\Sigma$,
we have that
$\Gamma(\phi_{3}(0,0))=\phi_{3}(\tau(t_m)+t_{p}, j_m)\in\Sigma$.
Since
$\Sigma\subset\Phi_{\bar{t}}(p+\sigma\BB)$ and
$\bar{t}$ can be sufficiently small, we have
$\phi_{3}(0,0)\in p+\sigma\BB$
and
$\Gamma(\phi_{3}(0,0))\in p+\sigma\BB$.
Then,
it follows that 
\begin{equation*}
|\Gamma(\phi_{3}(0,0))-\phi_{3}(0,0)|\leq |\Gamma(\phi_{3}(0,0))-p|
+|p-\phi_{3}(0,0)|\leq 2\sigma.
\end{equation*}
Therefore, since \Blue{the chosen} $\sigma$ can be \Blue{small enough},
we have that the map $\Gamma$ in \eqref{eq:PXaa} is continuous on $\Sigma$.

Next, by applying the Brouwer's fixed point theorem\IfTAC{,}{~(see Theorem~\ref{appx:thm2} in Appendix),} 
we show
that
the map $\Gamma$ has a fixed point $q\in \Sigma$.
\Red{Note that by assumption,
each $x\in M\cap C$
is a regular point of $f$ and
all conditions in Lemma~\ref{lemm:tubular} are satisfied.
Therefore, $p$ is a regular point of $f$, and} by using Lemma \ref{lemm:tubular},
\CBlue{for any open set $U\subset M\cap C$ containing $p$,
there exists an open neighborhood $\mathcal{N}_{p}\subset U$ of $p$
such that}
solutions to \eqref{eq:flowdyna} from $\mathcal{N}_{p}$
are diffeomorphic to the solutions to the system \eqref{eq:xii} on $(-1, 1)^n.$
In other words,
there exists
a $C^{r}$ diffeomorphism
$H: (-1, 1)^n \rightarrow \mathcal{N}_{p}$
such that
for any solution $\phi_{\xi}$ to \eqref{eq:xii},
$\phi_x=H(\phi_{\xi})$ is a solution to $\dot{x}=f(x)$\ $x\in \mathcal{N}_{p}$.
Note that given an initial condition $(s_{1},s_{2},\cdots,s_{n})\in (-1, 1)^n$,
a solution to \eqref{eq:xii} is given by
$\phi_{\xi}(t,0)=(s_1+t,s_{2},\cdots,s_{n})$
for all $t\in [0, 1-s_1)$, and that
the function
$t\mapsto H(\phi_{\xi}(t,0))$
is a solution to $\dot{x}=f(x)$\ $x\in \mathcal{N}_{p}$.
Denote $\Sigma_{\xi}:=H^{-1}(\Sigma)=\{\phi_{\xi}(t_{p},0)\in (-1, 1)^n:
t\mapsto \phi_{\xi}(t,0) \text{ is a solution to } \eqref{eq:xii}
\text{ from }\phi_{\xi}(0,0)\in (-1, 1)^n,
t\in [0, t_{p}], (t,0)\in\dom\phi_{\xi}\}$.
Note that $\Sigma_{\xi}$ is convex and bounded.
In addition, since $\Sigma$ is closed and
$H$ is a diffeomorphism,
$\Sigma_{\xi}$ is also closed and thus
$\Sigma_{\xi}$ is a convex compact set.
Define a map\footnote{The operator $\circ$ defines a function composition, i.e., $H^{-1}\circ \Gamma\circ H(x) = H^{-1}(\Gamma(H(\xi)))$ for all $\xi \in (-1, 1)^n$.} $\Gamma_{\xi}=H^{-1}\circ \Gamma\circ H$
as $\Gamma_{\xi}: \Sigma_{\xi}\rightarrow \Sigma_{\xi}$.
Due to $H$ being a diffeomorphism and continuity of $\Gamma$,
$\Gamma_{\xi}$ is continuous.
Therefore, by
Brouwer's fixed point theorem\IfTAC{,}{~(see Theorem~\ref{appx:thm2} in Appendix),}
~$\Gamma_{\xi}$
has a fixed point $q'\in \Sigma_{\xi}$,
i.e., $\Gamma_{\xi}(q')=q'$.
Since $\Gamma_{\xi}=H^{-1}\circ \Gamma\circ H$,
we have $H^{-1}\circ \Gamma\circ H(q')=q'$, which implies that
$\Gamma\circ H(q')=H(q')\in\Sigma$. Let $q=H(q')$.
Therefore, we have that $\Gamma$
has a fixed point $q\in \Sigma$,
i.e., $\Gamma(q)=q$.

From the existence of a fixed point for $\Gamma$
and the fact that $\phi_{3}(0,0)\in\Sigma$ and
$\phi_{3}(\tau(t_m)+t_{p}, j_m)\in\Sigma$,
there is a flow periodic solution $\phi^{*}$ to $\MH|_{M}$ with period $T^{*}=\tau(t_m)+t_{p}$ and $j_m$ jumps in each
period.
Therefore, $\MH|_{M}$ has a hybrid limit cycle $\MO$ with $j_m$ jumps in each period.
Note that
for the solution
$\phi\in \MS_{\MH|_{M}}(q)$,
every point $\zeta^*$ in the hybrid limit cycle $\MO$ is in $\Omega(\phi)$
since
there exists a sequence $\{ (t_{i}, j_{i}) \}_{i=1}^{\infty}$ of points $(t_i, j_i) \in \dom\phi$ such that
$\lim_{i\RA\infty} \phi(t_i, j_i) = \zeta^{*}$ with $\lim_{i\RA\infty} t_i + j_i=\infty$.
To prove that every point in $\Omega(\phi)$ is also in the hybrid limit cycle $\MO$,
we proceed by contradiction.
Suppose that $q\in \Omega(\phi)$ and $q\notin \MO$.
Since from the analysis above,
there is a fixed point $q\in\Omega(\phi)\cap \Sigma$
for \Blue{any chosen} $\sigma\in (0,\delta)$
we have
$\phi_{3}(0,0)\in\Sigma$ and
$\phi_{3}(\tau(t_m)+t_{p}, j_m)\in\Sigma$.
Then, $\Gamma(q)=q$, which leads to a contradiction with $q\notin \MO$.
Thus, every point in $\Omega(\phi)$ is also in the hybrid limit cycle $\MO$.
Therefore, we have that $\Omega(\phi)$ is a hybrid limit cycle.
\end{proof}

\begin{remark}\label{rem:ff}
In Theorem~\ref{thm:exist2}, there are several ways to guarantee that
$M\cap C$ contains no equilibrium set for the flow dynamics \eqref{eq:flowdyna}.
One way to assure that is to check if
for each $x\in M\cap C$,
$f^{\T}(x)f(x)>0$.
\end{remark}

The following example illustrates Theorem~\ref{thm:exist2}.

\begin{example}\label{exam:AS2}
Consider the academic system $\MH_{\mathrm{A}}|_{M_{\mathrm{A}}}$ in Example~\ref{exam:AS1}.
We will verify the existence of a hybrid limit cycle via Theorem~\ref{thm:exist2}.
First, items 1)-3) of Assumption~\ref{ass:basic_data} have been illustrated in Example~\ref{exam:AS1}.
\Red{Since the Jacobian of the map $f_{\textrm{A}}$ given by $\mathbb{J}_{f_{\textrm{A}}}(x)=-a$
with $a>0$ has the maximal rank $1$,
$M_{\mathrm{A}}\cap C_{\mathrm{A}}$
contains no critical points for the map $f_{\textrm{A}}$.
By} the definition of $M_{\mathrm{A}}$, for all $x\in M_{\mathrm{A}}\cap C_{\mathrm{A}}$,
$x\leq b_{1}<b/a$,
which implies that $f_{\mathrm{A}}^{\T}(x)f_{\mathrm{A}}(x)=(-ax+b)^2\geq
(b-ab_1)^2>0$.
Then, by Remark \ref{rem:ff},
$M_{\textrm{A}}\cap C_{\textrm{A}}$
contains no equilibrium set for the flow dynamics
$\dot{x}=f_{\textrm{A}}(x)$ \ $x\in M_{\textrm{A}}\cap C_{\textrm{A}}$,
where $(M_{\textrm{A}}\cap C_{\textrm{A}})^{\circ}=(0, b_{1})$ is nonempty.
By the definitions of $f_{\mathrm{A}}$ and $g_{\mathrm{A}}$,
the set $M_{\mathrm{A}}$ is forward invariant and
each $\phi \in \MS_{\MH_{\mathrm{A}}|_{M_{\mathrm{A}}}}(M_{\mathrm{A}}\cap  C_{\mathrm{A}})$
is unique and complete with $\dom\phi$ unbounded in the $t$ direction. 
Next, from the data of $\MH_{\mathrm{A}}|_{M_{\mathrm{A}}}$,
each solution $\phi \in \MS_{\MH_{\mathrm{A}}|_{M_{\mathrm{A}}}}(M_{\mathrm{A}}\cap  C_{\mathrm{A}})$
to $\dot x = f_{\mathrm{A}}(x)$ \ $x\in M_{\mathrm{A}}\cap C_{\mathrm{A}}$ is not complete and ends at a point in $M_{\mathrm{A}}\cap  C_{\mathrm{A}}$.
Therefore, for each maximal solution $\phi$ from $\xi\in [b_2, b_1]$ given by,
for each $(t,j)\in \BR_{\geq 0}\times \BN$,
\begin{equation*}
\phi(t,j)\!=\!\left\{\!\!\!
  \begin{array}{llllcccc}
    (\xi\!-\!\frac{b}{a}) e^{-a t}\!+\!\frac{b}{a} & \!\!t\in [0, t_{1}], j=0\\
    (b_{2}\!-\!\frac{b}{a}) e^{-a (t-t_{1}' )}\!+\!\frac{b}{a} &  \!\!t\in [t_{1}', t_{1}]\!+\!jT^{*}, j\in \BN\backslash\{0\}
  \end{array}
  \right.
\end{equation*}
where $t_{1}'=(j-1)T^{*}+t_{1}$,
$t_{1}=\frac{1}{a}\ln\frac{a\xi-b}{ab_{1}-b}$,
and
$T^{*}=\frac{1}{a}\ln\frac{ab_{2}-b}{ab_{1}-b}$,
by Theorem~\ref{thm:exist2},
$\MH_{\mathrm{A}}|_{M_{\mathrm{A}}}$ has a nonempty $\omega$-limit set
$\Omega(\phi):=\{x\in [0, b_1]: x=\phi(t,1), t\in [t_{1}, t_{1}+T^*]\}$.
Finally, the ZLAS property of
each $\phi \in \MS_{\MH_{\mathrm{A}}|_{M_{\mathrm{A}}}}$ has been verified in Example~\ref{exam:AS1}.
In addition, from the construction of $\Omega(\phi)$ and the condition $b_1>b_2>0$,
$\Omega(\phi)\cap (M_{\mathrm{A}}\cap C_{\mathrm{A}})^{\circ}=[b_2, b_1)$ is nonempty.
Therefore, by Theorem \ref{thm:exist2},
$\Omega(\phi)$
is the hybrid limit cycle for $\MH_{\mathrm{A}}|_{M_{\mathrm{A}}}$ with period
\IfTAC{
$T^{*}=\frac{1}{a}\ln\frac{ab_{2}-b}{ab_{1}-b}$ and one jump per period.}
{
$T^{*}=\frac{1}{a}\ln\frac{ab_{2}-b}{ab_{1}-b}$ and one jump per period as shown in \Figure~\ref{Academic:fig1}.
\begin{figure}[!ht]
\psfrag{x}[][][1.0]{$x$}
\psfrag{t}[][][1.0]{$t$}
\centering{\includegraphics[width=0.3\textwidth]{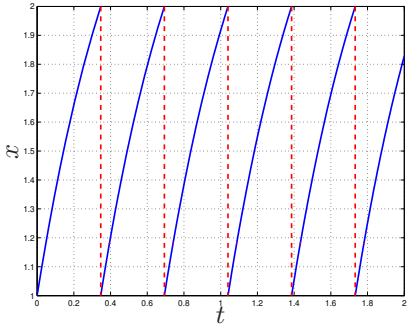}}
\caption{State trajectories of a solution to $\MH_{\mathrm{A}}|_{M_{\mathrm{A}}}$ from $b_2.$
Parameters used in the plot are $a=2,$ $b=6,$ $b_1=2$ and $b_2=1$. The solution is flow periodic with $T^{*}=\ln2/2$.
}
\label{Academic:fig1}
\end{figure}
}
\end{example}

In light of Example~\ref{exam:AS2},
~one may wonder if incremental graphical asymptotic stability
would serve as a necessary condition for the existence of a hybrid limit cycle.
Unfortunately, the fact that incremental graphical asymptotic stability
is a property for all solutions starting in a neighborhood makes it difficult
to allow for the existence of a hybrid limit cycle.
The following result establishes a sufficient condition for the nonexistence of hybrid limit cycles for systems that are $\delta$LAS.

\begin{theorem}\label{thm:nonexist}
For a hybrid system $\MH=(C,f,D,g)$ on $\BRn$ and a closed set $M\subset\BRn$ satisfying Assumption~\ref{ass:basic_data},
consider the hybrid system $\MH|_{M} = (M\cap C, f, M\cap D, g)$ and assume that
each solution $\phi \in \MS_{\MH|_{M}}(M\cap C)$
is 
complete with $\dom\phi$ unbounded in the $t$ direction.
If the hybrid system $\MH|_{M}$ is $\delta$LAS,
then $\MH|_{M}$ has no hybrid limit cycles for $\MH|_{M}$ with period given by some
$T^{*} > 0.$ 
\end{theorem}
\IfTAC{
\noindent A proof can be found in \CBlue{\cite[Theorem 5.17]{lou:TAC:2022}}.
}{
\begin{proof}
We proceed by contradiction.
Suppose $\MH|_{M}$ has a 
 hybrid limit cycle $\MO$
defined by a flow periodic solution $\phi^{*}$ with period $T^{*}>0.$ 
%
\SMF By Assumption~\ref{ass:basic_data},
every maximal solution to $\MH|_{M}$ is unique
via \cite[Proposition 2.11]{Goebel:book}. \STF
Consider two maximal solutions $\phi_1, \phi_2 \in \MS_{\MH|_{M}}$ where $\phi_1, \phi_2$ are two flow periodic solutions with
$\phi_i(t, j) = \phi_i(t + T^{*}, j+1)$ for each $(t,j)\in \dom\phi_i$, $i=1,2$.
Without loss of generality, for any $\mu>0$,
we can pick $\phi_1(0,0) \in M\cap D$ and $\phi_2(0,0) \in M\cap C$ satisfying
$|\phi_1(0,0)-\phi_2(0,0)|\leq \mu$.
Then, $\phi_1$ reaches the jump set before $\phi_2$, as $\phi_1(0,0)$ already belongs to the jump set.
For each $j\in \BN$, let
$t_{j} = \max_{(t,j)\in \dom\phi_1\cap \dom\phi_2}t$ and
$t'_{j} = \min_{(t,j+1)\in \dom\phi_1\cap \dom\phi_2}t$.
Then,
for each $j+1\in \BN\backslash\{0\}$,
the hybrid time domain of $\phi_1$
\Cred{is equal to the union of intervals}
$([t_{j},t_{j+1}], j+1)$
and
the hybrid time domain of $\phi_2$
\Cred{is equal to the union of intervals}
$([t_{j}',t_{j+1}'], j+1)$.
In addition, since for each $j\in \BN,$
$t'_{j}>t_{j}$,
it follows that $([t_{j}, t_{j}'], j+1) \subset \dom\phi_1$ and
$([t_{j}, t_{j}'), j+1) \not\subset \dom\phi_2$.
Moreover,
since $\phi_1$ and $\phi_2$ are two flow periodic solutions
that share the same hybrid limit cycle with period $T^{*}$, 
we have $t_j' - t_j = t_1' - t_1$ for each $j \in \BN$.
Now let $\varepsilon = \frac{1}{4} (t_1' - t_1) > 0$, and for any $T>0$, pick
$t^{*} = \frac{1}{2} ( t_{j^{*}} + t_{j^{*}}')$ at some $j^{*}\in \BN$
such that $(t^{*}, j^{*} + 1) \in \dom\phi_1$ and
$t^{*} + j^{*} + 1 > T$. Then, it is impossible to find $(s, j^{*} + 1) \in \dom\phi_2$ such that
$|\phi_1(t^{*}, j^{*}+1) - \phi_2(s, j^{*} + 1)| \leq \varepsilon$ with $|t^{*}-s| \leq \varepsilon$ since
$([t^{*} - \varepsilon, t^{*} + \varepsilon], j^{*} + 1) \subset
([t_{j^{*}}, t_{j^{*}}'), j^{*} + 1) \not\subset \dom\phi_2$.
This contradicts the assumption that the hybrid system $\MH|_{M}$ is $\delta$LAS.
\end{proof}
}

\section{Sufficient Conditions for Asymptotic Stability of Hybrid Limit Cycles}
\label{sec:stability}

\IfTAC{}{
In this section, we present stability properties of hybrid limit cycles for $\MH$.
}

\subsection{Notions}\label{sec:PM}

Following the stability notion introduced in \cite[Definition 3.6]{Goebel:book}, we employ the following notion for
stability of hybrid limit cycles.

\begin{definition}
\label{def:stability_hybrid_orbit}
Consider a hybrid system $\MH=(C,f,D,g)$ on $\BRn$ and a compact hybrid limit cycle $\MO$.
Then, the hybrid limit cycle $\MO$ is said to be
\begin{itemize}

\item \emph{stable} for $\MH$
if for every $\varepsilon>0$ there exists
$\delta>0$ such that every solution
$\phi$ to $\MH$ with $|\phi(0,0)|_{\MO}\leq \delta$ satisfies $|\phi(t,j)|_{\MO}\leq \varepsilon$
for each $(t,j)\in\dom\phi$;

\item \emph{globally attractive} for $\MH$
if every {maximal} solution $\phi$ to $\MH$ from $\bar{\MC}\cup \MD$ is complete and satisfies
$\lim\limits_{t+j\rightarrow\infty}\!|\phi(t,j)|_{\MO}\!=\!0;$

\item \emph{globally asymptotically stable} for $\MH$ if it is
both stable and globally attractive;

\item \emph{locally attractive} for $\MH$ if there exists $\mu>0$ such that every maximal solution $\phi$ to $\MH$
     starting from $|\phi(0,0)|_{\MO}\!\leq\! \mu$ is complete and satisfies
$\lim\limits_{t+j\rightarrow\infty}\!|\phi(t,j)|_{\MO}\!=\!0;$

\item \emph{locally asymptotically stable} for $\MH$ if it is
both stable and locally attractive.
\end{itemize}
\end{definition}

Given $M\subset\BRn$ and $\MH=(C,f,D,g)$,
for $x\in M\cap (\MC\cup \MD)$, define the ``distance" function $d:M\cap (\MC\cup \MD)\to \mathbb{R}_{\geq 0}$ as
\begin{align}
d(x):=\sup_{t\in[0,T_{I}(x)],\
(t,j)\in \dom \phi,\ \phi \in \MS_{\MH|_{M}}(x)} \; |\phi(t,j)|_{\MO}
\nonumber 
\end{align}
 when $0 \leq T_{I}(x) < \infty,$ and
$$
d(x):=\sup_{(t,j)\in \dom \phi,\ \phi \in \MS_{\MH|_{M}}(x)} \; |\phi(t,j)|_{\MO}
$$
if $T_{I}(x) = \infty,$ \Cred{where $T_{I}$ is the time-to-impact function defined in \eqref{eq:TI}}.
Note that $d$ vanishes on $\MO$.
%
Then, following \IfTAC{}{the ideas in} \cite[Lemma 4]{Grizzle:2001},
 the following property of the function $d$ can be established.
\begin{lemma}
\label{lem:continuity_distance}
Consider a hybrid system $\MH=(C,f,D,g)$ on $\BRn$ and a
closed set $M\subset\BRn$ satisfying Assumption~\ref{ass:basic_data}.
Suppose that every maximal solution to $\MH|_{M} = (M\cap C, f, M\cap D, g)$ is complete
and
$\MH|_{M}$ has a flow periodic solution $\phi^{*}$ with
period $T^{*}>0$
that defines a hybrid limit cycle
$\MO\subset M\cap C$.
%
Then, the function
$d:
M \cap C
\!\RA\! \BR_{\geq 0}$
is well-defined and continuous on $\MO$.
\end{lemma}
\IfTAC{
\noindent A proof can be found in \cite[Lemma 6.2]{lou:TAC:2022}.
}{
\begin{proof}
Given $x_0$ such that $x_0\in M\cap C$ and $0<T_I(x_0)<\infty$, the solution $t\mapsto \phi^f(t,x_0)$ to $\dot x = f(x)$ is well-defined on $t\in [0,T_I(x_0)]$. Therefore, $d(x_0)$ is also well-defined at $x_0$.

To show continuity of $d$ on $\MO$, consider a point $x\in \MO$.
As argued in the proof of Lemma~\ref{lem:continuity_impact_time}, the solution $t\mapsto \phi^f(t,x_0)$ to $\dot x = f(x)$ depends continuously on the initial condition. Then, by Lemma~\ref{lem:continuity_impact_time}, given $\varepsilon>0$, there exists $\delta>0$ such that for any $x\in \MO$ and $x'\in x+\delta\BB$,
via the triangle inequality, we have that
\begin{align}
|\phi^f(t,x')|_{\MO} \leq  |\phi^f(t,x) - \phi^f(t,x')| + |\phi^f(t,x)|_{\MO}
\end{align}
for each $t\in [0,\min\{T_I(x), T_I(x')\}]$, where $|T_I(x) - T_I(x')| < \varepsilon$ and $|\phi^f(t,x) - \phi^f(t,x')|<\varepsilon$ for each $t\in [0,\min\{T_I(x), T_I(x')\}]$.
\Blue{Then, for the solution $\phi \in \MS_{\MH}(x')$,
we have $\phi(t,0)=\sol(t,x')$
for each $0\leq t\leq T_{I}(x')$.}
 Note that by properties of $f$, $|\phi^f(t,x)|_{\MO} = 0$.
Then,
it follows that $|\phi^f(t,x')|_{\MO} \leq \varepsilon$ and
\begin{align*}d(x')\!=\!\!\!\sup_{
t\in[0,T_{I}(x')], \phi \in \MS_{\MH}(x')}\!\!\! |\phi(t,0)|_{\MO}
=\! \sup_{t\in[0,T_{I}(x')]}\; |\sol(t,x')|_{\MO}.
\end{align*}
\end{proof}
}


\subsection{Asymptotic Stability Properties of $\MO$} 
\label{subsec:stability}

To establish conditions for asymptotic stability of a hybrid limit cycle,
let us introduce a Poincar\'{e} map
for hybrid systems.
 Referred to as the {\em hybrid Poincar\'{e} map}, given a maximal solution
$\phi$ to $\MH|_{M}$, we denote it as $P:M\cap \MD\RA M \cap \MD$  and define
it as\footnote{The hybrid Poincar\'{e} map $P$ in \eqref{eq:PX} is different from
the Poincar\'{e} map $\Gamma:\Sigma\RA \Sigma$ in \eqref{eq:PXaa}.
The map $P$ in \eqref{eq:PX} maps $M\cap \MD$ to $M\cap \MD$ within one jump,
while the map $\Gamma$ in \eqref{eq:PXaa}
maps a closed set $\Sigma\subset (M\cap C)^{\circ}$ 
to $\Sigma$
and allows for multiple jumps.}
\begin{equation}\label{eq:PX}
\begin{array}{lllll}
P(x)\!\!\!\!&:=&\!\!\! \big\{
\phi(T_{I}(g(x)),j): \phi \in \MS_{\MH|_{M}}(g(x)),\\
&& \;\; (T_{I}(g(x)),j)\in \dom \phi \; \big\}
\;\; \forall x\!\in\! M\!\cap\! \MD,
\end{array}
\end{equation}
where $T_{I}$ is the time-to-impact function defined in \eqref{eq:TI}.

The importance of the hybrid Poincar\'{e} map in \eqref{eq:PX}
is that it allows one
to determine the stability of hybrid limit cycles.
Before revealing the stability properties of a hybrid limit cycle, we introduce the
following stability notions for the hybrid Poincar\'{e} map $P$ in \eqref{eq:PX}.
Let $P^k$ denote $k$ compositions of the hybrid Poincar\'{e} map $P$ with itself;
namely, $P^k(x)=\underbrace{P\circ P\cdots \circ P}_{k}(x)$.

\begin{definition}
\label{def:stability_fixedpoint}
A fixed point
$x^{*}$ of a hybrid Poincar\'{e} map $P: M\cap \MD \RA M \cap \MD$ defined in \eqref{eq:PX} is said to be
\begin{itemize}

\item \emph{stable} if for each $\epsilon>0$ there exists $\delta>0$ such that
for each $x\in M\cap \MD$, $|x-x^{*}|\leq \delta$ implies
 $|P^{k}(x)-x^{*}|\leq \epsilon$ for all $k\in\BN$;

\item \emph{globally attractive} if for each $x \!\in\! M \cap \MD$,
$\lim\limits_{k\rightarrow\infty} \! P^{k}(x) \!= \!x^{*};$

\item \emph{globally asymptotically stable} if it is
both stable and globally attractive;

\vspace{4pt}

\item \emph{locally attractive} if there exists $\mu>0$ such that
for each $x\in M\cap \MD$,
 $|x-x^{*}|\leq \mu$ implies
$\lim\limits_{k\rightarrow\infty}P^{k}(x)=x^{*};$
\item \emph{locally asymptotically stable} if it is
both stable and locally attractive.

\end{itemize}
\end{definition}


A relationship between stability of fixed points of hybrid Poincar\'{e} maps and stability of the corresponding hybrid limit cycles is established next.

\begin{theorem}\label{thm:main}
Consider a hybrid system $\MH=(C,f,D,g)$ on $\BRn$ and a
closed set $M\subset\BRn$ satisfying Assumption~\ref{ass:basic_data}.
Suppose that every maximal solution to $\MH|_{M} = (M\cap C, f, M\cap D, g)$ is complete
and
$\MH|_{M}$ has a flow periodic solution $\phi^{*}$ with
period $T^{*}>0$ 
that defines a hybrid limit cycle
$\MO\subset M\cap C$.
Then, the following statements hold:
\begin{itemize}
\item [1)] $x^{*}\in M\cap D$ is a stable fixed point of the hybrid Poincar\'{e} map $P$
 in \eqref{eq:PX}
if and only if the 
hybrid limit cycle $\MO$ of $\MH|_{M}$ generated by a flow periodic solution $\phi^{*}$ with period $T^{*}$
from
$\phi^{*}(0,0) = x^{*}$
is stable for $\MH|_{M}$,
\item [2)] $x^{*} \in M\cap D$ is a globally asymptotically stable fixed point
of the hybrid Poincar\'{e} map $P$ if and only if
$\MH|_{M}$ has a unique hybrid limit cycle $\MO$
generated by a flow periodic solution $\phi^{*}$ with period $T^{*}$
    from
$\phi^{*}(0,0) = x^{*}$ that is
globally asymptotically stable for $\MH|_{M}$ with basin of attraction containing every point in\footnote{A ``global''  property for $\MH|_{M}$ implies a ``global'' property of the original system $\MH$ only when $M$ contains $C$. For tools to establish the asymptotic stability property, see \cite{Goebel:book}.} $M \cap C$.
\end{itemize}
\end{theorem}
\begin{proof}
We first prove the sufficiency of item $1)$.
\SMF By Assumption~\ref{ass:basic_data},
every maximal solution to $\MH|_{M}$ is unique
via \cite[Proposition 2.11]{Goebel:book}. \STF
Consider the hybrid limit cycle $\MO$ generated by a flow periodic solution to $\MH|_{M}$
from
$x^{*}$ 
with $x^{*} \in M\cap D$.
Since $\MO$ is stable for $\MH|_{M}$,  given $\varepsilon>0$ there exists $\delta >0$ such that for any solution $\phi$ to $\MH|_{M}$, $|\phi(0,0)|_{\MO} \leq \delta$ implies $|\phi(t,j)|_{\MO} \leq \varepsilon $ for each $(t,j)\in \dom \phi$. Since $\phi$ is complete and $P^k(x^{*})=\phi(T_{I}(g(x^{*})),j)$ for some $j$,
in particular, we have that $|P^k(x^{*})|_\MO \leq \varepsilon$ for each $k\in {\mathbb N}$. Therefore, $x^{*}\in M \cap D$ is a stable fixed point of the hybrid Poincar\'{e} map $P$.


Next, we prove the necessity of item $1)$ as in the proof of \cite[Theorem 1]{Grizzle:2001}. Suppose that $x^{*}\in M\cap D$ is a stable fixed point of $P$. Then, $P(x^{*})=x^{*}$ due to the continuity of $P$ in \eqref{eq:PX}
and,
for any $\bar{\epsilon}>0$, there exists $\bar \delta>0$ such that
\begin{align*}
\tilde{x}\in (x^{*} + \bar \delta\BB)\cap (\X\cap\MD)
\end{align*}
implies \IfTAC{$P^{k}(\tilde{x})\in (x^{*}+\bar{\epsilon}\BB)\cap (\X\cap\MD)$ for all $k \in \mathbb{N}.$~}{$$P^{k}(\tilde{x})\in (x^{*}+\bar{\epsilon}\BB)\cap (\X\cap\MD) \quad \forall k \in \mathbb{N}.$$}Moreover, by assumption, every
maximal solution $\phi$ to $\MH|_{M}$ from
\IfTAC{$\tilde{x}\in (x^{*} + \bar \delta\BB)\cap (\X\cap\MD)$~}{\[\tilde{x}\in (x^{*} + \bar \delta\BB)\cap (\X\cap\MD)\]}is complete. Since solutions are guaranteed to exist from
$M \cap D$, there exists a \SMF complete solution $\phi$ \STF from every such
point $\tilde{x}$.
Furthermore, the distance between $\phi$ and the hybrid limit cycle $\MO$ satisfies\footnote{
Given two functions $d: M \cap C
 \to \BR_{\geq 0}$ and $g: M \cap D \to M \cap D$, the operator $\circ$ defines a function composition, i.e., $d \circ g (x) = d(g(x))$ for all $x \in M \cap D$.
}
\[
\sup\limits_{(t,j) \in \dom \phi}
|\phi(t,j)|_{\MO}\leq \sup\limits_{x\in (x^{*} + \bar \delta \BB) \cap (\X\cap\MD)} d\circ g(x).
\]
By Lemma~\ref{lem:continuity_distance}, $d$ is continuous at $x^{*}$. Since $\MO$ is transversal to $\X\cap\MD$, $\MO\cap (\X\cap\MD)$ is a singleton, $g(x^{*})\in\MO,$  and $g$ is continuous,
we have that $d\circ g$ is continuous at $x^{*}$.
Moreover, since $d\circ g(x^{*})=0,$
it follows by continuity that given any $\epsilon>0$, we can pick $\bar \epsilon$ and $\bar \delta$ such that $0<\bar{\epsilon} < \epsilon$ and
\[
\sup\limits_{x\in (x^{*}+ \bar \delta\BB) \cap (\X\cap\MD)} d\circ g(x) \leq \epsilon.
\]
Therefore, an open neighborhood of $\MO$ given by $\MV:=\{x\in \BRn: d(x)\in [0, \Blue{\bar\epsilon})\}$
is such that any solution $\phi$ to $\MH|_{M}$ from $\phi(0,0) \in \MV$ satisfies $|\phi(t,j)|_\MO \leq \epsilon$ for each $(t,j)\in \dom \phi$. Thus, the necessity of item $1)$ follows immediately.

The stability part of item $2)$ follows similarly. Sufficiency of the global attractivity part
in item $2)$ is proved as follows. Suppose the hybrid limit cycle $\MO$ generated by
a flow periodic solution 
to $\MH|_{M}$ from $x^{*}$ 
is globally attractive for $\MH|_{M}$ with basin of attraction containing every point in $M \cap C$.
Then, given $\epsilon>0$, for any solution $\phi$ to $\MH|_{M}$, there exists $\bar T>0$ such that $|\phi(t,j)|_\MO \leq \epsilon$ for each $(t,j)\in \dom \phi$ with $t+j\geq \bar T$.
Note that \Blue{$\phi$ is precompact since
$\phi$ is complete and the set $M$ is compact by Assumption~\ref{ass:basic_data}.}
Therefore,
via \cite[Lemma 2.7]{Sanfelice:2007},
$\dom\phi$ is unbounded in the $t$-direction
as Assumption~\ref{ass:basic_data}
prevents solutions from being Zeno\IfTAC{.~}{(see Remark~\ref{rk:hybrid_basic_conds}).~}
It follows that $|P^k(x^{*})|_\MO \leq \epsilon$ for sufficiently large $k$. Therefore, $x^{*}$ is a globally attractive fixed point of $P$.

Finally, we prove the necessity of the global attractivity property in item $2)$. Assume that $x^{*} \in M \cap D$ is a globally attractive fixed point.
Then,
for any $\bar{\epsilon}>0$, there exists $\bar \delta>0$ such that,
for all $k\in \mathbb{N}$,~\IfTAC{$\tilde{x}\in (x^{*}+ \bar \delta \BB)\cap (\X\cap\MD)$~}{\[\tilde{x}\in (x^{*}+ \bar \delta \BB)\cap (\X\cap\MD)\]}
implies
$\lim\limits_{k\rightarrow\infty}P^{k}(\tilde{x})=x^{*}.$
\SMF Moreover, following from Definition~\ref{def:stability_fixedpoint}, it is implied that a maximal solution $\phi$ to $\MH|_{M}$
from $x^{*}$ is complete. \STF
Then, by continuity of $d$ and $g$,
\[\lim\limits_{k\rightarrow\infty} d\circ g(P^{k}(\tilde{x}))=d\circ g(x^{*})=0,\]
from which it follows that
\IfTAC{
\begin{eqnarray}
\lim\limits_{t+j\RA\infty} |\phi(t,j)|_{\MO}
&\!\!\leq\!\!&
d\circ g(x^{*})=0.\nonumber
\end{eqnarray}}
{\begin{eqnarray}
\lim\limits_{t+j\RA\infty} |\phi(t,j)|_{\MO}
&\!\!\leq\!\!&
\lim\limits_{k\rightarrow\infty}
\sup\limits_{\tilde{x}\in (x^{*}+ \bar \delta \BB) \cap (\X\cap\MD)}
d\circ g( P^{k}(\tilde{x}) )
\nonumber\\
&\!\!\leq\!\!&
\sup\limits_{\tilde{x}\in (x^{*}+ \bar \delta \BB) \cap (\X\cap\MD)}
\lim\limits_{k\rightarrow\infty}
d\circ g( P^{k}(\tilde{x}) )
\nonumber\\
&\!\!\leq\!\!&
d\circ g(x^{*})=0.\nonumber
\end{eqnarray}}The proof is complete. 
\end{proof}

\begin{remark}
In \cite{Grizzle:2001},
sufficient and necessary conditions for stability properties of periodic orbits in impulsive systems
are established using properties of the fixed points of the corresponding Poincar\'{e} maps.
Compared to \cite{Grizzle:2001},
Theorem~\ref{thm:main} enables the use of
the Lyapunov stability tools in \cite{Goebel:book}
to certify asymptotic stability of a fixed point without even
computing the hybrid Poincar\'{e} map.
\end{remark}

At times, one might be interested only on local
asymptotic stability of the fixed point of the hybrid Poincar\'{e} map.
Such case is handled by
the following result.

\begin{corollary}\label{cor:main}
Consider a hybrid system $\MH=(C,f,D,g)$ on $\BRn$ and a closed set $\X\subset\BRn$ satisfying Assumption~\ref{ass:basic_data}.
Suppose that every maximal solution to $\MH|_{M} = (M\cap C, f, M\cap D, g)$ is complete
and
$\MH|_{M}$ has a flow periodic solution $\phi^{*}$ with
period $T^{*}>0$ 
that defines a hybrid limit cycle
$\MO\subset M\cap C$.
Then,
$x^{*}\in M\cap D$ is a locally asymptotically stable fixed point of the
hybrid Poincar\'{e} map $P$ if and only if
$\MH|_{M}$ has a unique hybrid limit cycle $\MO$ generated by a flow periodic solution $\phi^{*}$
with period $T^{*}$ 
from $\phi^{*}(0,0)= x^{*}$ that is locally asymptotically stable for $\MH|_{M}$.
\end{corollary}
\IfTAC{
\noindent The proof can be found in \cite[Corollary 6.6]{lou:TAC:2022}.
}{
\begin{proof}
The sufficiency and necessity of the stability part have been proven in Theorem~\ref{thm:main}. The sufficiency of the claim is proved as follows. Suppose the hybrid limit cycle $\MO$ generated by a flow periodic solution $\phi^{*}$ to $\MH|_{M}$ from $\phi^{*}(0,0)= x^{*}$ is locally attractive for $\MH|_{M}$. Then, given $\epsilon>0$ there exists $\delta>0$ such that
for any solution $\phi$ to $\MH|_{M}$ with
$|\phi(0,0)|_\MO \leq \delta$ and $\phi(0,0)\in \MB_{\MO}$
there exists $\bar T>0$ such that $|\phi(t,j)|_\MO \leq \epsilon$ for each $(t,j)\in \dom \phi$, \SMF $t+j \geq \bar T$. \STF
Then, it follows that $|P^k(x^{*})|_\MO \leq \epsilon$  for sufficiently large $k$. Therefore, $x^{*}\in M\cap \MD$ is a locally attractive fixed point (with basin of attraction given by $M\cap D\cap \MB_{\MO}$).

To prove the necessity of the claim, assume that $x^{*} \in M \cap D$ is a locally attractive fixed point of the hybrid Poincar\'{e} map. Then,
for any $\bar{\epsilon}>0$, there exists $\bar \delta > 0$ such that
for each
\[\tilde{x}\in (x^{*}+ \bar \delta \BB)\cap (\X\cap\MD) \subset \MB_\MO\]
we have
$\lim\limits_{k\rightarrow\infty}P^{k}(\tilde{x})=x^{*}.$
Then, by continuity of $d$ and $g$,
\[\lim\limits_{k\rightarrow\infty} d\circ g(P^{k}(\tilde{x}))=d\circ g(x^{*})=0,\]
from which it follows that
\begin{align}
\lim\limits_{t+j\RA\infty} |\phi(t,j)|_{\MO}
&\leq
\lim\limits_{k\rightarrow\infty}
\sup\limits_{\tilde{x}\in (x^{*}+\bar \delta \BB) \cap (\X\cap\MD)}
d\circ g( P^{k}(\tilde{x}) )
\nonumber\\
&\leq
\sup\limits_{\tilde{x}\in (x^{*}+ \bar \delta \BB) \cap (\X\cap\MD)}
\lim\limits_{k\rightarrow\infty}
d\circ g( P^{k}(\tilde{x}) )
\nonumber\\
&\leq
d\circ g(x^{*})=0.\nonumber
\end{align}
%
\end{proof}
}

\begin{remark}\label{compare:transv}
In \cite{Grizzle:2001} and \cite{HaddadACC:2002}, the authors extend the Poincar\'{e} method to
analyze the stability properties of periodic orbits in nonlinear systems with impulsive effects. In particular, the solutions to the systems considered therein
are right-continuous over (not necessarily closed) intervals of flow.
In particular, the models therein (as well as those in \cite{Tang:Manchester:2014})
require $\MC\cap \MD = \emptyset$, \SMF
which prevents the application of the robustness results in \cite{Goebel:book} due to the fact that the hybrid basic conditions would not hold.
On the other hand,
our results allow us to \STF
establish robustness properties of hybrid limit cycles
as shown in Section~\ref{sec:robustness}.
\end{remark}

\begin{remark}
In \cite{Tang:Manchester:2014}, within a contraction framework,
conditions guaranteeing local orbital stability of limit cycles for a class of hybrid systems are provided,
where, as a difference to the notion used here, orbital stability is solely defined as an attractivity
(or convergence) property.
Note that the case of limit cycles with multiple jumps for hybrid systems is not
explicitly analyzed in~\cite{Tang:Manchester:2014}, while
the results here are applicable to the situation where a hybrid limit cycle may contain multiple jumps within a period; see our preliminary results in \cite{lou.li.sanfelice16:TAC}. 
\end{remark}

\IfTAC{
}{
The following two examples illustrate the sufficient condition in Theorem~\ref{thm:main}
and Corollary~\ref{cor:main}
by checking the eigenvalues of the Jacobian matrix of the hybrid Poincar\'{e} map at the fixed point, respectively.
}


\begin{example}\label{exam:TCP4}
Consider the hybrid congestion control system in Example~\ref{exam:TCP3}.
A solution $\phi^{*}$ to $\MH_{\mathrm{TCP}}|_{M_{\mathrm{TCP}}}=
(M_{\textrm{\tiny TCP}}\cap\MC_{\textrm{\tiny TCP}},f_{\textrm{\tiny TCP}},M_{\textrm{\tiny TCP}}\cap\MD_{\textrm{\tiny TCP}},g_{\textrm{\tiny TCP}})$ from
$\phi^{*}(0,0)=(q_{\max}, 2Bm/(1+m))\in M_{\textrm{\tiny TCP}}\cap \MC_{\textrm{\tiny TCP}}$
is a flow periodic solution with $T^{*}=2B(1-m)/(a+ma),$  
which defines a hybrid limit cycle
$\MO\subset M_{\textrm{\tiny TCP}}\cap \MC_{\textrm{\tiny TCP}}$.
We verify the sufficient
condition 2) in Theorem \ref{thm:main} as follows.
Due to the specific form of the flow map of $\MH_{\textrm{\tiny TCP}}$,
the Jacobian of the hybrid Poincar\'{e} map has an explicit analytic form. The flow solution $\phi^f$ to the flow dynamics $\dot x=f_{\textrm{\tiny TCP}}(x)$ from $\xi$ \Cred{is} given by
\begin{equation}\label{eq:phif:TCP}
\phi^f(t,\xi)=
\left[
\begin{array}{ccc}
\xi_1 + (\xi_2-B)t+ \frac{at^{2}}{2}\\
\xi_2 + at
\end{array}
\right].
\end{equation}
From the definition of the hybrid Poincar\'{e} map and the solution of the flow dynamics
from $x=(q_{\max}, r)$
with $x\in\MD_{\textrm{\tiny TCP}}$,
it follows that
\[P_{\textrm{\tiny TCP}}(x)
=\left[
\begin{array}{cc l}
q_{\max} + (mr-B)\hat{T}+ \frac{a{\hat{T}}^{2}}{2}\\
mr + a\hat{T}
\end{array}
\right],
\]
where $\hat{T}=2(B-mr)/a$, which leads to
\IfTAC{\begin{equation}\label{eq:normal:PM}
P_{\textrm{\tiny TCP}}(x)
=(q_{\max}, 2B-mr).
\end{equation}}{\begin{equation}\label{eq:normal:PM}
P_{\textrm{\tiny TCP}}(x)
=\left[
\begin{array}{ccc}
q_{\max}\\
2B-mr
\end{array}
\right].
\end{equation}}\IfTAC{Then, the eigenvalues of the Jacobian of the hybrid Poincar\'{e} map $P_{\textrm{\tiny TCP}}$ are computed as
$\lambda_{1}=0$ and $\lambda_{2}=-m,$~}{
The Jacobian of the hybrid Poincar\'{e} map $P_{\textrm{\tiny TCP}}$ is given by
\[\mathbb{J}_{P_{\textrm{\tiny TCP}}}(x)=
\left[
\begin{array}{l c l}
0 & 0\\
0 & -m
\end{array}
\right],
\]
whose eigenvalues are
$\lambda_{1}=0$ and $\lambda_{2}=-m,$~}
which, since $m\in (0, 1)$, are inside the unit circle.
According to Theorem \ref{thm:main},
the hybrid limit cycle $\MO$ of the hybrid system $\MH_{\mathrm{TCP}}|_{M_{\mathrm{TCP}}}$
is 
asymptotically stable
with basin of attraction containing every point in $M_{\textrm{\tiny TCP}}\cap \MD_{\textrm{\tiny TCP}}$.  
\end{example}

\IfTAC{}{
\begin{example}\label{exam:Izhi:revisit4}
Consider the Izhikevich neuron system analyzed in Example~\ref{exam:Izhi:revisit2}.
Suppose the hybrid Poincar\'{e} map for $\MH_{\mathrm{I}}|_{M_{\mathrm{I}}}$ is given by $P_{\textrm{I}}$ with a fixed point $x^{*}$.
Using Corollary~\ref{cor:main} to show that
the hybrid limit cycle $\MO$ of $\MH_{\mathrm{I}}|_{M_{\mathrm{I}}}$ is locally asymptotically stable,
it suffices to check the eigenvalues of the Jacobian matrix
of the hybrid Poincar\'{e} map at the fixed point.
Due to the quadratic form in the flow map of $\MH_{\mathrm{I}}$,
an analytical expression of the Jacobian of the hybrid Poincar\'{e} map is difficult to obtain.
Instead, we 
compute the Jacobian based on a numerically approximated Poincar\'{e} map.
We consider the case of intrinsic bursting behavior with parameters  
given in \eqref{eq:Izhi:param}, and obtain a fixed point
$x^{*}\approx (30, -7.50)$ and limit cycle period $T^{*}\approx 31.22$ \Cred{\text{s}}.
The Jacobian
of the hybrid Poincar\'{e} map at the fixed point is
$$\mathbb{J}_{P_{\textrm{I}}}(x^{*})=
\begin{bmatrix}
         0    &     0\\
   0 &  -0.025
\end{bmatrix}.
$$
The eigenvalues of $\mathbb{J}_{P_{\textrm{I}}}$ are
$\lambda_{1}=0$ and $\lambda_{2}=-0.025,$
which are inside the unit circle.
Corollary~\ref{cor:main} implies that the hybrid limit cycle $\MO$ of the Izhikevich neuron model
is locally asymptotically stable.
The properties of the hybrid limit cycle $\MO$
are illustrated numerically in~\Figure~\ref{fig:1},
\Cred{where the hybrid limit cycle
is defined by the solution from the point $P_2$ that jumps at the
point $P_1$.}
Note that this hybrid limit cycle $\MO$
is locally asymptotically stable, but not globally asymptotically stable
with basin of attraction containing every point in $M_{\mathrm{I}} \cap  C_{\mathrm{I}}$.
\end{example}
}

\IfTAC{}{
\begin{example}\label{exam:Compass2}
Consider the compass gait biped system analyzed in Example~\ref{exam:Compass}.
Suppose the hybrid Poincar\'{e} map for $\MH_{\mathrm{B}}|_{M_{\mathrm{B}}}$ is given by $P_{\textrm{B}}$ with a fixed point $x^{*}$.
Using 
Corollary~\ref{cor:main}
to show that
$x^{*}$ is locally asymptotically stable for $\MH_{\mathrm{B}}|_{M_{\mathrm{B}}}$, then
it suffices to check the eigenvalues of the Jacobian matrix
of the hybrid Poincar\'{e} map at the fixed point.
Due to the nonlinear form in the flow map of $\MH_{\mathrm{B}}$,
an analytical expression of the Jacobian of the hybrid Poincar\'{e} map is difficult to obtain.
Instead, we 
compute the Jacobian based on a numerically approximated Poincar\'{e} map.
We consider the parameters given in \eqref{eq:biped:param}, and obtain a fixed point
$x^{*}\approx (0.22,-0.32,-1.79,-1.49)$ and limit cycle period $T^{*}\approx 0.734$ \Cred{\text{s}}.
The Jacobian
of the hybrid Poincar\'{e} map at the fixed point is
{\small$$\mathbb{J}_{P_{\textrm{B}}}(x^{*})=
\begin{bmatrix}
    0.0000  &  0.0000 &  -0.0000 &  -0.0000\\
    4.8611 &  -1.7156 &  -0.0948  &  0.8187\\
    5.3550 &  -1.9180 &  -0.1241  &  1.2845\\
    9.3471 &  -3.2933 &  -0.2036  &  1.9851
\end{bmatrix}.
$$}
The eigenvalues of $\mathbb{J}_{P_{\textrm{B}}}$ are
$\lambda_{1}=0.8897$, $\lambda_{2}= -0.7456$, $\lambda_{3}= -0.0000$, $\lambda_{4}= 0.0013$
with one eigenvalue located at zero and the other one locates inside the unit circle.
Corollary~\ref{cor:main} implies that the hybrid limit cycle $\MO$ of the compass gait biped system
is locally asymptotically stable.
The properties of the hybrid limit cycle $\MO$
are illustrated numerically in~\Figure~\ref{fig:biped},
\Cred{where the hybrid limit cycle
is defined by the solution from the point $P_1$ that jumps at the
point $P_2$.}
Note that this hybrid limit cycle $\MO$
is locally asymptotically stable, but not globally asymptotically stable
with basin of attraction containing every point in $M_{\mathrm{B}} \cap  C_{\mathrm{B}}$.
\end{example}
}

\section{Robustness of \Cred{Asymptotically} Stable Hybrid Limit Cycles}
\label{sec:robustness}

\IfTAC{}{In this section, we
propose results on robustness to generic perturbations,
to inflations of flow and jump sets, and to computation error of the hybrid Poincar\'{e} map.}

\subsection{Robustness to General Perturbations}

First, we present results guaranteeing robustness to generic perturbations of asymptotically stable hybrid limit cycles.
More precisely,
we consider the perturbed continuous dynamics of the hybrid system $\MH|_{M}=(\X\cap\MC,f,\X\cap\MD,g)$ given by \IfTAC{$\dot{x} = f(x+d_{1})+d_{2} \quad  x+\SM d_{3} \ST\in M\cap C,$}{\begin{equation}
\dot{x} = f(x+d_{1})+d_{2} \quad  x+\SM d_{3} \ST\in M\cap C,  \nonumber
\end{equation}} where $d_{1}$ corresponds to state noise \Cred{(e.g., measurement noise)}, $d_{2}$ captures unmodeled
dynamics \Cred{or additive perturbations}, and $d_{3}$ captures
generic disturbances on the state \Cred{when checking if the state belongs to the constraint}.
Similarly, we consider the perturbed discrete dynamics \IfTAC{$x^{+} = g(x+d_{1})+d_{2} \quad  x+\SM d_{4} \ST\in M\cap D,$}{\begin{equation}
x^{+} = g(x+d_{1})+d_{2} \quad  x+\SM d_{4} \ST\in M\cap D, \nonumber
\end{equation}} where $d_{4}$ captures generic disturbances on the state \Cred{when checking if the state belongs to the constraint $M\cap D$}.
The hybrid system $\MH|_{M}$ with such perturbations results in the perturbed hybrid system
\begin{equation}\label{sec3:eq33}
\tilde{\MH}|_{M}
\left\{
\begin{array}{cccccc}
\dot{x} \!\!&=&\!\! f(x\!+\!d_{1})\!+\!d_{2} && \!\!\!\! x+\SM d_{3} \ST\in \X\!\cap\!\MC,\\
x^{+}   \!\!&=&\!\! g(x\!+\!d_{1})\!+\!d_{2} && \!\!\!\! x+\SM d_{4} \ST\in \X\!\cap\!\MD.
\end{array}
\right.
\end{equation}
The perturbations \SM $d_{i}$ $(i=1,2,3,4)$ \ST might be state
or hybrid time dependent, but are assumed to have Euclidean norm bounded by \SM $\bar{M}_{i}\geq 0$ $(i=1,2,3,4)$, \ST and to be admissible, namely,
$\dom d_{i}$ $(i=1,2,3,4)$ is a hybrid time domain and the function $t\mapsto d_{i}(t,j)$ is measurable on $\dom d_{i}\cap (\BR_{\geq 0}\times
\{j\})$ for each $j\in\BN$.

\IfTAC{}{
Next,
we recall the following stability notion from \cite[Definition 7.10]{Goebel:book}.

\begin{definition}\label{def:KLAS}
{($\KL$ asymptotic stability)}  
Let $\MH$ be a hybrid system on $\BRn$, $\MA\subset\BRn$ be a compact set, and
$\MB_{\MA}$ be the basin of attraction of the set $\MA$.\footnote{\Cred{The}
 $\MB_{\MA}$ is the set of points $\xi\in\BRn$ such that every
complete solution $\phi$ to $\MH|_{M}$ with $\phi(0,0)=\xi$ is bounded and $\lim_{t+j\rightarrow\infty}
|\phi(t,j)|_{\MA}=0.$
}
The set $\MA$ is $\KL$ asymptotically stable on $\MB_{\MA}$ for $\MH$ if for every
proper indicator $\omega$ of $\MA$ on $\MB_{\MA}$, there exists a function $\beta\in\KL$
such that \Cred{every solution $\phi\in\MS_{\MH}(\MB_{\MA})$ satisfies}
\begin{equation}\label{KLbound}
\omega(\phi(t,j))\leq
\beta(\omega(\phi(0,0)), t+j)
\quad \forall (t,j)\in \dom \phi.
\end{equation}
\end{definition}
}

\IfTAC{
}
{The next result establishes that
local asymptotic stability of $\MO$ and
Assumption~\ref{ass:basic_data} guarantee a $\KL$ bound as in \eqref{KLbound}; namely, local asymptotic stability of $\MO$ is uniform.}

\begin{theorem}\label{thm:KLAS}
Consider a hybrid system $\MH=(C,f,D,g)$ on $\BRn$ and a closed set $\X\subset\BRn$ satisfying Assumption~\ref{ass:basic_data}.
If $\MO$ is a locally asymptotically stable hybrid limit cycle
for $\MH|_{M}$ with basin of attraction $\MB_{\MO}$,
then $\MO$ is $\mathcal{K}\mathcal{L}$ asymptotically stable\IfTAC{\footnote{\Red{See \cite[Definition 7.10]{Goebel:book} for a definition of $\mathcal{K}\mathcal{L}$ asymptotic stability.}}}{} on the basin of attraction $\MB_{\MO}$ of the set $\MO$.
\end{theorem}
\begin{proof}
First, it is proved in Lemma \ref{lem:closeness_orbit} that
$\MO$ is a compact set.
Second, note that for a hybrid system $\MH$ on $\BRn$ and a closed set $\X\subset\BRn$ satisfying Assumption~\ref{ass:basic_data},
$\MH|_{M}$ is well-posed \cite[Definition 6.29]{Goebel:book}. Then, it is also nominally well-posed.
Therefore, according to \cite[Theorem 7.12]{Goebel:book},
$\MB_{\MO}$ is open and $\MO$ is $\mathcal{K}\mathcal{L}$ asymptotically stable on $\MB_{\MO}$.
\end{proof}

The following result establishes that the stability of $\MO$ for $\MH|_{M}$
is robust to the class of perturbations defined above.

\begin{theorem}\label{thm:robust}
Consider a hybrid system $\MH=(C,f,D,g)$ on $\BRn$ and a closed set $\X\subset\BRn$ satisfying Assumption~\ref{ass:basic_data}.
If $\MO$ is a locally asymptotically stable compact set for $\MH|_{M}$ with basin of attraction $\MB_{\MO}$,
then for every proper indicator $\omega$ of $\MO$ on $\MB_{\MO}$ there exists
$\tilde{\beta}\in \KL$ such that
for every $\varepsilon>0$ and every compact set $K\subset \MB_\MO$,
there exist \SM $\bar{M}_i>0,$ $i\in \{1,2,3,4\},$ such that for any
admissible perturbations $d_i,$ $i\in \{1,2,3,4\},$ with
Euclidean norm bounded by $\bar{M}_i$, respectively, \ST
every solution $\tilde\phi$ to $\tilde{\MH}|_{M}$ with $\tilde\phi(0,0)\in K$
satisfies
$$\omega (\tilde\phi(t,j)) \leq
\tilde{\beta}( \omega( \tilde\phi(0,0) ), t+j)+\varepsilon
\quad \forall (t,j)\in\dom{\tilde\phi}.$$
\end{theorem}
\begin{proof}
\Cred{Following \cite[Section 6.4]{Goebel:book}},
we introduce the following perturbed hybrid system $\MH|_{M}^{\rho}$ with 
\Cred{constant} $\rho > 0$:
\begin{equation}\label{eq:rho:perturBB}
\MH|_{M}^{\rho} \left\{
  \begin{array}{cccclcccc}
    \dot{x} \ \in F_{\rho}(x) && x \in \MC_{\rho}, \\
    x^{+}  \! \in G_{\rho}(x) && x \in \MD_{\rho},
  \end{array}
  \right.
\end{equation}
where 
\begin{equation}
{\footnotesize
\begin{array}{rlll}
\MC_{\rho} \!\!\!\!&:=\{x\in\BRn : (x+\rho\BB)\cap (\X\cap\MC) \neq \emptyset\},\nonumber\\
F_{\rho}(x) \!\!\!\!& := \overline{\text{co}}f((x+\rho\BB)\cap (\X\cap\MC))+\rho\BB\quad \Cred{\forall x\in\BRn},\nonumber\\
\MD_{\rho} \!\!\!\!& :=\{x\in\BRn : (x+\rho\BB)\cap (\X\cap\MD) \neq \emptyset\},\nonumber\\
G_{\rho}(x)  \!\!\!\!& := \{v\!\in\!\BRn \!:\! v\!\in\! \eta \!+\! \rho\BB, \eta\!\in\! g((x \!+\! \rho\BB)\!\cap\! (\X\!\cap\!\MD) )\}
\; \Cred{\forall x\!\in\!\BRn}. \nonumber
\end{array}
}
\end{equation}
Then, every solution to $\tilde{\MH}|_{M}$
with admissible perturbations \SM $d_i$ having Euclidean norm bounded by $\bar{M}_i,$ $i\in \{1,2,3,4\},$ respectively, \ST
is a solution to the hybrid system $\MH|_{M}^{\rho}$ with \SM
$\rho \geq \max\{\bar{M}_{1},\bar{M}_{2},$ $\bar{M}_{3},\bar{M}_{4}\}$, \ST
which corresponds to an outer perturbation of $\MH|_{M}$ and satisfies
\Cred{the convergence property \cite[Assumption 3.25]{Ricardo:book:2021}. 
Then, the claim follows by \cite[Theorem 3.26]{Ricardo:book:2021}} and the fact that every solution to $\tilde{\MH}|_{M}$ is a solution to \eqref{eq:rho:perturBB}.
In fact, using \cite[Theorem 3.26]{Ricardo:book:2021},
for every proper indicator $\omega$ of $\MO$ on $\MB_{\MO}$
there exists $\tilde{\beta}\in\KL$ such that for each compact set
$K\subset \MB_{\MO}$ and each $\varepsilon>0$, there exists $\rho^{*}>0$
such that for each $\rho\in (0, \rho^{*}]$, every solution $\phi_{\rho}$ to \eqref{eq:rho:perturBB}
from $K$ satisfies
$\omega(\phi_{\rho}(t,j)) \leq
\tilde{\beta}( \omega(\phi_{\rho}(0,0)), t+j )+\varepsilon
$
for each $(t,j)\in\dom{\phi_{\rho}}$.
The proof concludes using the
relationship between the solutions to $\tilde{\MH}|_{M}$ and \eqref{eq:rho:perturBB},
and picking $\bar{M}_i,$ \IfTAC{}{$i\in \{1,2,3,4\},$}  such that $\max\{\bar{M}_{1},\bar{M}_{2},$ $\bar{M}_{3},\bar{M}_{4}\} \in (0,\rho^*]$.
\end{proof}

\begin{remark}
Robustness results of stability of compact sets for
general hybrid systems are available in \cite{Goebel:book}.
Since $\MO$ is an asymptotically stable compact set for $\MH|_{M}$,
Theorem \ref{thm:robust}
is novel
in the context of the literature of Poincar\'{e} maps.
In particular, if one was to write the systems in \cite{Grizzle:2001} and \cite{HaddadACC:2002} within the
framework of \cite{Goebel:book},
then one would not be able to apply the results
on robustness for hybrid systems in \cite{Goebel:book} since the hybrid basic conditions
would not be satisfied and the hybrid limit cycle may not be given
by a compact set.
Furthermore, through an application of \cite[Lemma 7.19]{Goebel:book}, it can be shown that
the hybrid limit cycle is robustly $\mathcal{K}\mathcal{L}$ asymptotically stable on  $\MB_{\MO}$.
\end{remark}

\begin{remark}
Recently, the authors in \cite{Hamed:Gregg:2017, Hamed:Gregg:2019} present static or dynamic decentralized (event-based) controllers
for robust stabilization of hybrid periodic orbits against
possible disturbances and established results on
$H_2/H_{\infty}$ optimal decentralized event-based control design.
In contrast to our work, they use input-to-state stability
for robust stability properties of hybrid periodic orbits with respect to disturbance inputs in the discrete dynamics. 
Note that the results in \cite{Hamed:Gregg:2017, Hamed:Gregg:2019}
consider possible disturbances only \Cred{on the} discrete dynamics
and are only suitable for nonlinear impulsive systems that have jumps on switching surfaces. 
On the other hand, in this paper, we establish conditions for robustness of
hybrid limit cycles that allow \Cred{disturbances}
in the continuous/discrete dynamics
and are applicable for hybrid dynamical systems
with nonempty intersection between the flow set and the jump set.
\end{remark}

\begin{remark}
\Cred{Very recently, the authors in \cite{Choi:2022} propose a reachability-based approach to
compute regions-of-attraction for hybrid limit cycles in
a class of hybrid systems with a switching surface
and bounded disturbance.
Note that the approach in \cite{Choi:2022}
deals with bounded disturbance only on the continuous dynamics
and is only suitable for hybrid systems that have jumps on switching surfaces.}
\end{remark}

\IfTAC{}{
\begin{example}\label{exam:Izhi:Robust}
Consider the Izhikevich neuron system in Example~\ref{exam:Izhi:revisit1}.
\Cred{We illustrate} Theorem~\ref{thm:main} for the hybrid system $\MH_{\mathrm{I}}$
by plotting the solutions from the initial condition $(-55, -6)$,
when an additive perturbation \Cred{$d_{2}=(d_{2v},d_{2w})$} affects the jump map.
The noise is injected as unmodeled dynamics on the jump map as
\Cred{$d_{2}=(d_{2v},d_{2w})=(\rho\sin(t), 0),$
where a variety values for $\rho$ are used to verify} robustness.
\Figure~\ref{fig3:exam1} shows the phase plots
for the perturbed solution (red line) and normal solution (blue line).
It is found that the hybrid limit cycle $\MO$ is robust to the additive perturbation $d_{2}$
when \Cred{$\rho\in (0,0.4)$} as shown in \Figure~\ref{fig:sub1},
while $\MO$ is not robust to the additive perturbation $d_{2}$
when $\rho=0.42$ as shown in \Figure~\ref{fig:sub2}.

Next, more simulations are performed to quantify the relationship between
$\rho^{*}$ (the maximal value of the perturbation parameter $\rho$)
and $\varepsilon$ (the desired level of closeness to $\MO$).
Given a compact set
$K:= [-57, -53] \times [-6.20, -5.80]$,
and different desired region radiuses
$\varepsilon\in\{0.3, 0.6, 0.9, 1.2, 1.5\}$,
the simulation results are shown in Table~\ref{tab:margin:Izhi}, which
indicates that the relationship between $\rho^{*}$ and $\varepsilon$ can be approximated as
$\rho^{*}\approx 0.004\varepsilon$.
As it can be seen, the larger admissible convergence error
the larger the perturbation parameter $\rho^{*}$ can be.
This study validates Theorem~\ref{thm:robust}.
\begin{figure}[!ht]
\centering
\subfigure[$\rho=0.24$]{
\label{fig:sub1}
\psfrag{x1}[][][1.0]{$v$ (mV)}
\psfrag{y1}[][][1.0]{$w$}
\includegraphics[width=\figwidth]{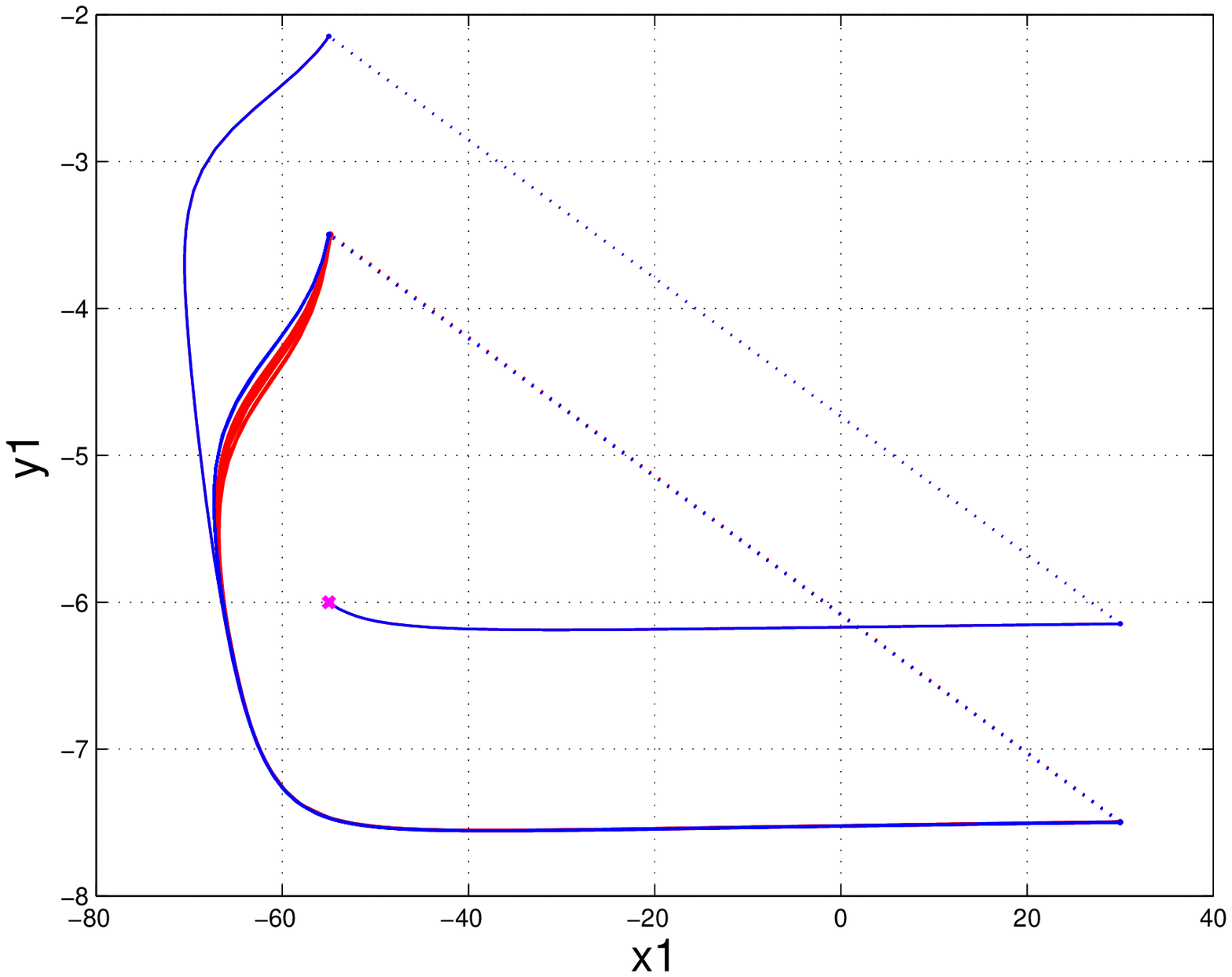}
}
\
\subfigure[$\rho=0.42$]{
\label{fig:sub2}
\psfrag{x1}[][][1.0]{$v$ (mV)}
\psfrag{y1}[][][1.0]{$w$}
\includegraphics[width=\figwidth]{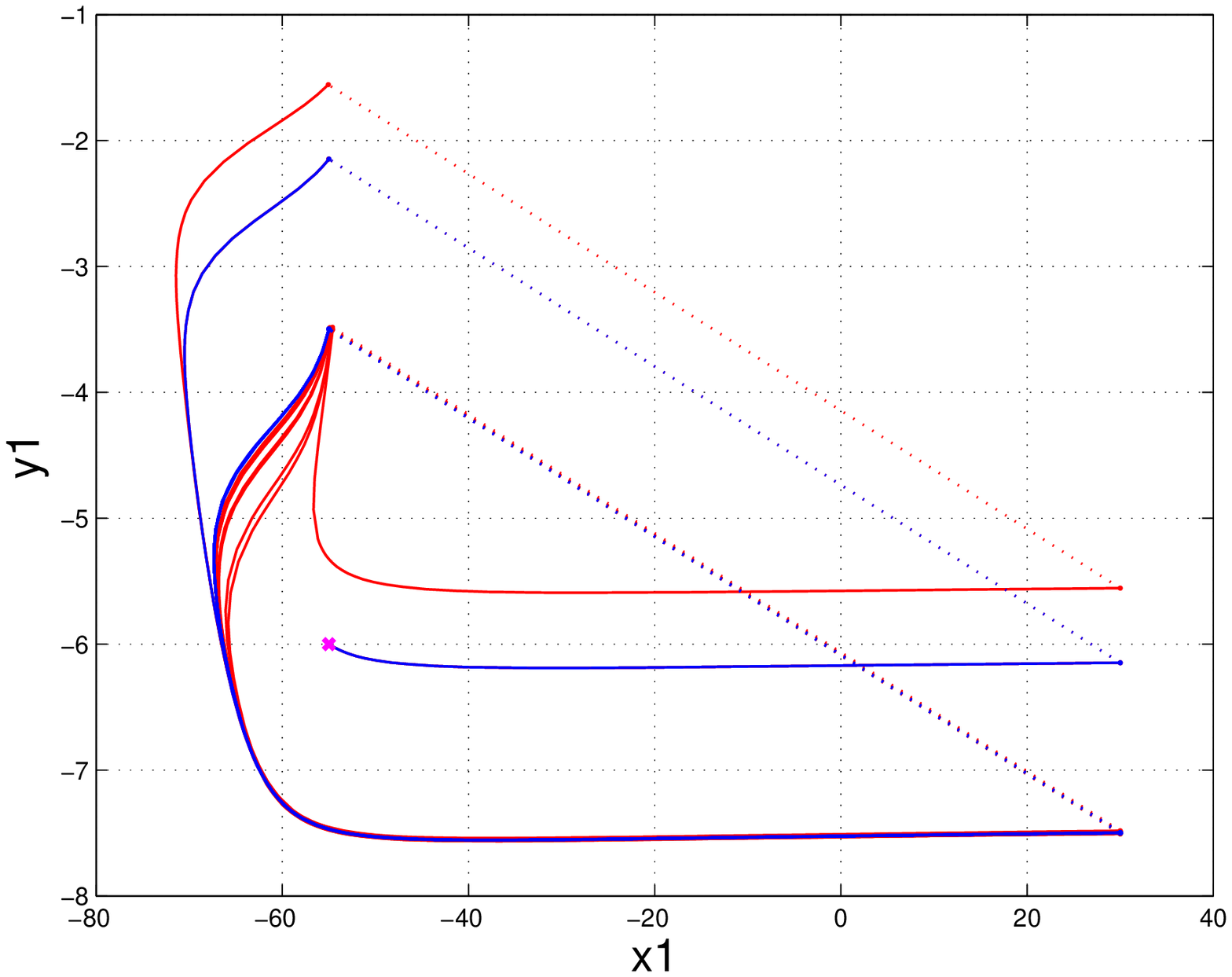}
}
\caption{Trajectories with initial condition $(-55, -6)$ in Example \ref{exam:Izhi:Robust}.
\Cred{The perturbed solutions (red line) and normal solutions (blue line) are given for different values of
additive perturbation $d_{2}=(\rho\sin(t), 0).$
(a) When $\rho=0.24$, the hybrid limit cycle $\MO$ is robust to small state perturbation;
(b) when $\rho=0.42$, $\MO$ is not robust to the additive perturbation $d_{2}$.}
}
\label{fig3:exam1}
\end{figure}
\begin{table}[!ht]
\tabcolsep 5mm
\caption{Relationship between $\rho^{*}$ and $\varepsilon$}
\begin{center}
\begin{tabular}{cccc}
\hline
$\rho^{*}$ & $\varepsilon$ & $\rho^{*}/\varepsilon$ \\ \hline
    0.0010  &  0.30 &   0.0031\\
    0.0019  &  0.60 &   0.0030\\
    0.0032  &  0.90  &  0.0035\\
    0.0042  &  1.20  &  0.0034\\
    0.0058  &  1.50  &  0.0039\\
    \hline
\end{tabular}
\end{center}\label{tab:margin:Izhi}
\end{table}
\end{example}
}

\IfTAC{\vspace{-2mm}}{}
\subsection{Robustness to Inflations of $C$ and $D$}
\label{sec:RobustInflations}

We consider the following specific parametric perturbation on $h$, in both the flow and jump sets,
with $\epsilon>0$ denoting the parameter:
the perturbed flow set is an inflation of the original flow set while the condition
$h(x)=0$ in the jump set is replaced by $h(x)\in [-\epsilon,\epsilon]$.
The resulting system is denoted as $\MH|_{M}^{\epsilon}$ and is given by
\begin{equation}\label{sec:HMrho}
  \MH|_{M}^{\epsilon} \left\{
  \begin{array}{cccclcccc}
    \dot{\bx} & = & f(\bx) && \bx \in \MC_\epsilon\cap M, \\
    \bx^{+}   & = & g(\bx) && \bx \in \MD_\epsilon\cap M, \\
  \end{array}
  \right.
\end{equation}
where the flow set and the jump set are replaced by
$\MC_\epsilon=\{ x\in\BRn: h(x)\geq -\epsilon \}$ and $\MD_\epsilon=\{ x\in\BRn: h(x)\in [-\epsilon,\epsilon], L_{f}h(x)\leq 0 \}$,
respectively, while the flow map and jump map are the same as for $\MH|_{M}$.
We have the following result, whose proof follows from the proof of Theorem~\ref{thm:robust}.
\begin{theorem}\label{col:robustSR}
Consider a hybrid system $\MH=(C,f,D,g)$ on $\BRn$ and a closed set $\X\subset\BRn$ satisfying Assumption~\ref{ass:basic_data}.
If $\MO$ is a locally asymptotically stable compact set for $\MH|_{M}$ with basin of attraction $\MB_{\MO}$,
then there exists $\tilde{\beta}\in \KL$ such that,
for every $\varepsilon>0$ and each compact set $K\subset \MB_\MO$,  
there exists $\bar \epsilon>0$ such that for each $\epsilon\in (0,\bar \epsilon]$
every solution $\phi$ to $\MH|_{M}^{\epsilon}$ in \eqref{sec:HMrho} with $\phi(0,0)\in K$ 
satisfies
\begin{equation}\label{eqn:KLboundParametricPerturbation}
|\phi(t,j)|_{\MO} \leq
\tilde{\beta}( |\phi(0,0)|_{\MO}, t+j)+\varepsilon
\quad \forall (t,j)\in\dom{\phi}.
\end{equation}
\end{theorem}

Theorem~\ref{col:robustSR} implies that
the asymptotic stability property of the hybrid limit cycle $\MO$ is robust to
a parametric perturbation on $h$.
\Red{Note that
the $\mathcal{KL}$ bound \eqref{eqn:KLboundParametricPerturbation}
is obtained when
the parametrically perturbed system $\MH|_{M}^{\epsilon}$ in \eqref{sec:HMrho} should also exhibit a hybrid limit cycle.}
At times, a relationship between
the maximum value $\bar\epsilon$ of the perturbation and the factor $\varepsilon$ in the semiglobal and practical $\KL$ bound in
\eqref{eqn:KLboundParametricPerturbation} can be established numerically.
Next, Theorem~\ref{col:robustSR} and this relationship are illustrated
in the TCP congestion control example.

\begin{example}\label{exam:TCP:robust2}
Let us revisit the hybrid congestion control system \eqref{eq:TCP3}
in \IfTAC{Section \ref{sec:motive_ex},}{Example~\ref{exam:TCP2},}
where, now, the flow set $\MC_{\textrm{\tiny TCP}}$ and the jump set $\MD_{\textrm{\tiny TCP}}$ are replaced by
$\MC_{\textrm{\tiny TCP}}^{\epsilon}=\{x\in\BR^{2}:  q_{\max}-q\geq -\epsilon  \}$ \RRed{and}
$\MD_{\textrm{\tiny TCP}}^{\epsilon}=\{x\in\BR^{2}: q_{\max}-q\in [-\epsilon,\epsilon], r\geq B\},$
respectively.
To validate \Blue{Theorem}~\ref{col:robustSR}, multiple simulations are performed
to show a relationship between $\bar{\epsilon}$, the maximal value of the perturbation
parameter $\epsilon$, and $\varepsilon$, the desired level of closeness to the hybrid limit cycle $\MO$.
Given the compact set
$K= [0.68, 0.72] \times [0.58, 0.64]$
and different
desired level $\varepsilon \in \{0.01, 0.02, 0.03, 0.04\}$  of closeness to the hybrid limit cycle,
\IfTAC{it}{the simulation results are shown in Table~\ref{tab:margin:TCP:SR},
which} indicates that the relationship between $\bar{\epsilon}$ and $\varepsilon$ can be approximated as
$\bar{\epsilon}\approx 2.8 \varepsilon$.
\IfTAC{}{
%
\begin{table}[!ht]
\tabcolsep 5mm
\caption{Relationship between $\bar{\epsilon}$ and $\varepsilon$}
\begin{center}
\begin{tabular}{cccc}
\hline
$\bar{\epsilon}$ & $\varepsilon$ & $\bar{\epsilon}/\varepsilon$ \\ \hline
    0.022  &  0.01  &  2.20\\
    0.056  &  0.02  &  2.80\\
    0.078  &  0.03  &  2.60\\
    0.107  &  0.04  &  2.68\\ \hline
\end{tabular}
\end{center}
\label{tab:margin:TCP:SR}
\end{table}
}
\end{example}

\IfTAC{
\vspace{-3mm}
\subsection{Robustness to Computation Error of Hybrid Poincar\'{e} Map}
}
{\subsection{Robustness to Computation Error of The Hybrid Poincar\'{e} Map}
}

The \Cred{hybrid Poincar\'{e} map} defined in \eqref{eq:PX} indicates
the evolution of a trajectory of a hybrid system from a point on the jump set $M\cap D$
to another point in the same set $M\cap D$.
As stated in Theorem~\ref{thm:main} and Corollary~\ref{cor:main}, stability of hybrid limit cycles can be verified by checking
the eigenvalues of the Jacobian of the hybrid Poincar\'{e} map at its fixed point.
However, errors in the computation of the hybrid Poincar\'{e} map
may influence the statements made about \Cred{asymptotic} stability.
Typically, the hybrid Poincar\'{e} map is computed numerically by
discretizing the flows, using
integration schemes such as Euler, Runge-Kutta,
and multi-step methods \cite{Ricardo:2010:simulator},
\Cred{which unavoidably lead to an approximation of Poincar\'{e} maps}.

Following the ideas in
\IfTAC{\cite{Ricardo:2010:simulator}}{\cite{Ricardo:2010:simulator,Stuart94}}
about perturbations
introduced by computations,
the discrete-time system associated with the (exact) hybrid Poincar\'{e} map $P$ in \eqref{eq:PX}
is given by\footnote{\Cred{By some abuse of notation,
though it is not hybrid, we label as $\MH_{P}$ the discrete-time system in \eqref{sec:exactPM}
associated to the Poincar\'{e} map $P$
and we use $j\in \BN$ as time instead of $(0,j)$ for $\MH_{P}$ afterwards.
}}
\begin{equation}\label{sec:exactPM}
\MH_{P}: \bx^{+} = P(\bx) \quad  \bx\in \X\cap\MD,
\end{equation}
which we treat as a hybrid system without flows.
As argued above, due to unavoidable errors in computations and computer implementations,
only approximations of the map $P$ and \Cred{of} the sets $M$ and $D$ are
available.
In particular, given a point $\bx\in \X\cap\MD$,
the value of the step size, \Cred{denoted $s>0$,} used in the computation of  $P$ \Cred{at a point $x$}
affects the precision of the resulting approximation, which, in turn,
may prevent the solution to \eqref{sec:exactPM} to remain in $ \X\cap\MD$
and be complete.
Due to this, we denote by $P_{s}$ \Cred{the results of computing} $P$, and
by $M_s$ and $D_s$ the approximations of $M$ and $D$, respectively.
\Cred{With some abuse of notation,} the discrete-time system associated with $P_{s}, M_s,$ and $D_s$
is defined as
\begin{equation}\label{sec:eqPM}
\MH_{P_{s}}: \bx^{+} = P_{s}(\bx) \quad  \bx\in M_s\cap\MD_{s}.
\end{equation}

\Cred{The approximations of $P_{s}, M_s,$ and $D_s$} are assumed to satisfy the following properties.

\begin{assumption}\label{ass:data:Ps}
Given $M\subset\BRn$ and $\MH=(C,f,D,g)$,
the function $P_s:\reals^n \to \reals^n$ \Cred{parameterized by $s>0$} is such that, for some
continuous function $\varrho:\reals^n \to \reals_{\geq 0}$,
there exists $s^*>0$ such that,
\Cred{for all $x \in M\cap D$,}
\begin{eqnarray} \label{eqn:PsProperty}
P_s(x) \in P_\varrho(x)  \qquad \forall s  \in  (0,s^*]
\end{eqnarray}
where \IfTAC{$P_\varrho(x):=\{v \in \reals^n: v \in g + \varrho(g)\BB,\;
g \in P(x + \varrho(x)\ball) \}  
$}{$$
P_\varrho(x):=\{v \in \reals^n: v \in g + \varrho(g)\ball,\;
g \in P(x + \varrho(x)\ball) \}  
$$} and the set $M_s \cap D_s$ satisfies, \Cred{for any positive sequence $\{s_i\}_{i=1}^{\infty}$ such that $s_{i}\searrow 0$,
\begin{eqnarray} \label{eqn:MsCapDsProperty}
\limsup\limits_{i\rightarrow\infty} M_{s_i} \cap D_{s_i} \subset M \cap D.
\end{eqnarray}}
\end{assumption}
\begin{remark}
The property in \eqref{eqn:PsProperty} is a consistency condition
on the integration scheme used to compute the flows
involved in \eqref{eq:PX}. For instance, when the forward
Euler method is used to approximate those flows,
the numerical values of $\phi$ are generated using the scheme $x+ s f(x)$, which,
under Lipschitzness of $f$ and boundedness of solutions (and its derivatives) to $\dot x = f(x)$,
is convergent of order $1$; in particular,
the error between $P_s$ and $P$ is $O(s)$, which implies that \eqref{eqn:PsProperty} holds
for some function
\IfTAC{$\varrho$.}{$\varrho$;
see, e.g., \cite[Chapter 3.2]{AscherPetzold98}.}
 Vaguely, the property in \eqref{eqn:MsCapDsProperty} holds when
a distance between $M_s \cap D_s$ and $M \cap D$
approaches zero as the step size vanishes, which
is an expected property as precision improves with a decreasing step size.
\Blue{Condition \eqref{eqn:MsCapDsProperty} is satisfied when, for small enough $s>0$, $M_{s} \cap D_{s}$
is contained in an outer perturbation of $M \cap D$.
Very often, the jump set $M \cap D$ can be implemented accurately in the computation of the hybrid Poincar\'{e} map, i.e., it
may be possible to take \IfTAC{$M_{s}=M$ and $D_{s}=D$.}{$M_{s}=M$ and $D_{s}=D$, as illustrated in Example~\ref{exam:TCP:robust5} below.}}
\end{remark}


The following closeness result
between solutions to $\MH_{P}$ and $\MH_{P_s}$ holds.

\begin{theorem}\label{thm:PM:closeness}
(\Blue{closeness between} solutions and approximations on compact domains)
Consider a hybrid system $\MH=(C,f,D,g)$ on $\BRn$ and a closed set $\X\subset\BRn$ satisfying Assumption~\ref{ass:basic_data}.
Assume the computed Poincar\'{e} map $P_{s}$ approximating $P$ and
the sets $M_s$ and $D_s$ approximating $M$ and $D$, respectively,
satisfy Assumption~\ref{ass:data:Ps}.
Then, for every compact set $K\subset M \cap D$, 
every
$\varepsilon>0$, and every simulation horizon $J\in \BN$, there exists $s^{*}>0$ with the following property:
there exists $\delta^{*}>0$ such that for each $\delta\in (0,\delta^{*}]$,
for each $s\in (0, s^{*}]$ and any solution
$\phi_{P_{s}} \in \MS_{\MH_{P_{s}}}( K+\delta\BB )$    
there exists a solution
$\phi_{P} \in \MS_{\MH_{P}}(K)$    
\Cred{with $\dom\phi_{P}\subset\BN$}
such that $\phi_{P_{s}}$ and $\phi_{P}$ are \Cred{$(J, \varepsilon)$-close}.
\footnote{\Cred{
See \cite[Definition 3.2]{Ricardo:2010:simulator} for a definition of $(T,J,\varepsilon)$-close
to quantify the distance between hybrid arcs (and solutions).
Here, 
it is just the hybrid case but with $t=0$.}}
\end{theorem}
\IfTAC{
\noindent A proof can be found in \cite[Theorem 7.12]{lou:TAC:2022}.
}{
\begin{proof}
By Assumption~\ref{ass:basic_data} and the definition of $P$ in \eqref{eq:PX},
\Cred{$P$ is well-defined and,
by Lemma~\ref{lem:continuity_impact_time}, continuous} on the closed set $M\cap D$.
Then, the hybrid system $\MH_{P}$ without flows
satisfies the hybrid basic conditions \Cred{A1) and A2) in Section~\ref{subsec:hybrid}}.
Next, to show the convergence property between $\MH_{P}$ and $\MH_{P_{s}}$,
let us first consider an outer perturbation of $\MH_{P}$.
\Blue{Given $\delta>0$, let $s\in (0,\delta]$.
Obviously, the step size $s$ explicitly depends on $\delta$
and approaches zero as $\delta \searrow 0$.}
The outer perturbation
\IfTAC{}{(see \cite[Example 5.3]{Goebel:2006} for more details)}
of $\MH_{P}$
for a state dependent perturbation determined by the constant $\delta$  
and
a continuous function $\varrho:\reals^n \to \reals_{\geq 0}$ is given by
\begin{equation}\label{eqn:OP}
\MH_{P_{\delta}}: \bx^{+} \in P_{\delta}(\bx)\quad x\in D_{\delta},
\end{equation}
where
$P_{\delta}(x) := \big\{ v \in \reals^n: v \in g + \delta\varrho(g)\ball,\;
g \in P( x +$ $\delta\varrho(x)\ball) \big\},$
$D_{\delta}:=\{ x\in\BRn: x+\delta\varrho(x)\BB \cap (M\cap D)
\neq \emptyset \}.$
%
%

Then, by \Blue{\cite[Lemma 5.1]{Ricardo:2010:simulator}},
the outer perturbation
$\MH_{P_{\delta}}$ of $\MH_{P}$ has the convergence property.
Consequently,
the closeness between solutions to $\MH_{P}$ and $\MH_{P_{\delta}}$ follows from \cite[Theorem 3.4]{Ricardo:2010:simulator}.
Using Assumption~\ref{ass:data:Ps},
for every compact set $K\subset\BRn,$
the solutions $\phi_{P_{s}} \in \MS_{\MH_{P_{s}}}( K+\delta\BB )$
are solutions to the perturbed hybrid system
$\MH_{P_{\delta}}$. Hence, the properties of solutions to the perturbed hybrid system
$\MH_{P_{\delta}}$ also hold for those of $\MH_{P_{s}}$.
The proof concludes by exploiting
the closeness property between solutions to $\MH_{P}$ and $\MH_{P_{s}}$.
\end{proof}
}

Inspired by \cite[Theorem 5.3]{Ricardo:2010:simulator},
the following stability result shows that when
Assumption~\ref{ass:data:Ps} holds,
asymptotic stability of the fixed point of $P$ (assumed to be unique) is preserved
under the computation of $P$.
\begin{theorem}\label{thm:PM:SPS}
(stability preservation under computation error of $P$)
Consider a hybrid system $\MH=(C,f,D,g)$ on $\BRn$ and a closed set $\X\subset\BRn$ satisfying Assumption~\ref{ass:basic_data}.
Assume that the computed Poincar\'{e} map $P_{s}$ approximating $P$ and
the sets $M_s$ and $D_s$ approximating $M$ and $D$, respectively,
satisfy Assumption~\ref{ass:data:Ps},
and that $x^{*}$ is a unique
globally asymptotically stable 
fixed point of $P$.
Then, $x^{*}$ is a unique semiglobally practically asymptotically stable 
fixed point of $P_{s}$ with basin of attraction containing every point in $M\cap D,$
i.e.,
there exists $\tilde{\beta}\in \KL$ such that,
for every $\varepsilon>0$, each compact set $K\subset M_{s} \cap D_{s}$,
and every simulation horizon $J\in \BN$, there exists $s^{*}>0$ such that, for each $s\in (0, s^{*}]$, every
solution $\phi_{P_s} \in \MS_{\MH_{P_{s}}}(K)$
to $\MH_{P_{s}}$ 
satisfies for each
$j\in \dom\phi_{P_{s}}$
\begin{equation}
|\phi_{P_{s}}(j)-x^{*}| \leq
\tilde{\beta}( |\phi_{P_{s}}(0)-x^{*}|, j)+\varepsilon. \nonumber
\end{equation}
\end{theorem}
\begin{proof}
Since the hybrid system $\MH_{P}$ without flows
satisfies the hybrid basic conditions \Cred{A1) and (A2) in Section~\ref{subsec:hybrid}}
and
$x^{*}$ is a unique globally asymptotically stable
fixed point of $P$,  
by \cite[Theorem 3.1]{Ricardo:2010:simulator},
there exists $\tilde{\beta}\in \KL$ such that
\Cred{each solution
$\phi_{P} \in \MS_{\MH_{P}}(M \cap D)$  to $\MH_{P}$ 
satisfies}
\begin{equation}
|\phi_{P}(j)-x^{*}| \leq
\tilde{\beta}( |\phi_{P}(0)-x^{*}|, j)\quad \forall j\in \dom\phi_{P}. \nonumber
\end{equation}
\IfTAC{\Blue{Given $\delta>0$, let $s\in (0,\delta]$.}~}{}Given a compact set $K\subset M_{s} \cap D_{s}$
and a simulation horizon $J\in \BN$,
by the assumptions,
\cite[Lemma 5.1]{Ricardo:2010:simulator} implies that
\IfTAC{\Blue{for a state dependent perturbation determined by the constant $\delta$
and a continuous function $\varrho:\reals^n \to \reals_{\geq 0}$,
the outer perturbation $\MH_{P_{\delta}}$ 
of $\MH_{P}$ given by
\begin{equation}\label{eqn:OP}
\MH_{P_{\delta}}: \bx^{+} \in P_{\delta}(\bx)\quad x\in D_{\delta},
\end{equation}
where
$P_{\delta}(x) := \big\{ v \in \reals^n: v \in g + \delta\varrho(g)\ball,\;
g \in P( x +$ $\delta\varrho(x)\ball) \big\},$
$D_{\delta}:=\{ x\in\BRn: x+\delta\varrho(x)\BB \cap (M\cap D)
\neq \emptyset \},$}}
{the outer perturbation $\MH_{P_{\delta}}$ of $\MH_{P}$ in \eqref{eqn:OP}}
satisfies the convergence property in \cite[Definition 3.3]{Ricardo:2010:simulator}.
Then, using $K$ above,
\cite[Theorem 3.5]{Ricardo:2010:simulator} implies that
for each $\varepsilon>0$ there exists $\delta^{*}>0$ such that
for each $\delta\in (0,\delta^{*}]$,
every solution $\phi_{P_{\delta}} \in \MS_{\MH_{P_{\delta}}}( K+\delta\BB )$ to $\MH_{P_{\delta}}$   
satisfies for each $j\in\dom\phi_{P_{\delta}}$
\begin{equation}
|\phi_{P_{\delta}}(j)-x^{*}| \leq
\tilde{\beta}( |\phi_{P_{\delta}}(0)-x^{*}|, j)+\varepsilon. \nonumber
\end{equation}
By Assumption~\ref{ass:data:Ps},
the properties of solutions to \IfTAC{}{the perturbed hybrid system}
$\MH_{P_{\delta}}$ also hold  for  solutions $\phi_{P_{s}}$.
The result follows by this preservation and the
$\KL$ bound of solutions to $\MH_{P}$.
\end{proof}

Note that the property in Theorem~\ref{thm:PM:SPS} holds for small
enough step size $s$. The step size bound $s^{*}$ decreases with
the desired level of closeness to $x^{*}$, which is given by $\varepsilon$.
The next result shows that
the computed Poincar\'{e} map $P_{s}$ has
a semiglobally asymptotically stable \IfTAC{(semi-GAS)}{} compact set $\MA_{s}$
with basin of attraction containing every point in $M\cap D$
that reduces to a singleton $\{ x^{*} \}$ as $s$ approaches zero.
\begin{theorem}\label{thm:PM:continuity}
(continuity of asymptotically stable fixed points)
Consider a hybrid system $\MH=(C,f,D,g)$ on $\BRn$ and a closed set $\X\subset\BRn$ satisfying Assumption~\ref{ass:basic_data}.
%
Assume that $x^{*}$ is a unique
globally asymptotically stable fixed point of the hybrid Poincar\'{e} map $P$
and the computed Poincar\'{e} map $P_{s}$ approximating $P$ and
the sets $M_s$ and $D_s$ approximating $M$ and $D$, respectively,
satisfy Assumption~\ref{ass:data:Ps}.
Then, there exists $s^{*}>0$ such that for each $s\in (0, s^{*}]$,
the computed Poincar\'{e} map $P_{s}$ has
a \IfTAC{semi-GAS}{semiglobally asymptotically stable}
compact set $\MA_{s}$
with basin of attraction containing every point in $M\cap D$
satisfying
\Blue{\IfTAC{$\lim_{s\searrow 0} \MA_{s}=x^{*}.$}{$$\lim_{s\searrow 0} \MA_{s}=x^{*}.$$}}
\end{theorem}
\begin{proof}
%
%
%
Let $K$ be any compact set such that for some
$\varepsilon>0,$ $x^{*}+2\varepsilon\BB\subset K\subset \BRn$.
Using $K$ as above and an arbitrary simulation horizon
$J\in \BN$,
consider the perturbed system $\MH_{P_{\delta}}$ in \eqref{eqn:OP}
and define $\tilde{\MH}_{\tilde{P}_{\delta}}$ with
\begin{align}
\tilde{P}_{\delta}(x)=
\left\{
\begin{aligned}
P_{\delta}(x)\cup \{ x^{*} \} &\quad x\in D_{\delta}\\
 \{x^{*}\}\qquad &\quad x\in \BRn\backslash D_{\delta}
\end{aligned}
\right.\nonumber
\end{align}
and $\tilde{D}_{\delta}=\BRn$.
Using $K$ and $\varepsilon$ as above,
\cite[Theorem 3.5]{Ricardo:2010:simulator} implies that
for each $\varepsilon>0$ there exists $\delta^{*}>0$ such that
for each $\delta\in (0,\delta^{*}]$,
every solution $\phi_{\tilde{P}_{\delta}} \in \MS_{\tilde{\MH}_{\tilde{P}_{\delta}}}(K)$ to $\tilde{\MH}_{\tilde{P}_{\delta}}$
satisfies for each \Cred{$j\in\dom \phi_{\tilde{P}_{\delta}}$}
\begin{equation}\label{eqn:Hdelta:KL}
|\phi_{\tilde{P}_{\delta}}(j)-x^{*}| \leq
\tilde{\beta}( |\phi_{\tilde{P}_{\delta}}(0)-x^{*}|, j)+\varepsilon.
\end{equation}

For a simulation horizon
$J\in \BN$, let $\mathrm{Reach}_{J,\tilde{\MH}_{\tilde{P}_{\delta}}}(x^{*}+2\varepsilon\BB)$
be the reachable set of $\tilde{\MH}_{\tilde{P}_{\delta}}$ from $x^{*}+2\varepsilon\BB$ up to $J,$ i.e.,
\IfConf{
\begin{align}
\mathrm{Reach}_{J,\tilde{\MH}_{\tilde{P}_{\delta}}}\!\!(x^{*}\!+\!2\varepsilon\BB)\!:=
& \{ \phi_{\tilde{P}_{\delta}}\!(j)\!:\!
\phi_{\tilde{P}_{\delta}} \text{ is a solution to } \tilde{\MH}_{\tilde{P}_{\delta}}, \nonumber\\
& \hspace{-7mm}\phi_{\tilde{P}_{\delta}}(0)\in x^{*}+2\varepsilon\BB, j\in\dom\phi, j\leq J\}. \nonumber
\end{align}
}{
\begin{align}
\mathrm{Reach}_{J,\tilde{\MH}_{\tilde{P}_{\delta}}}\!\!(x^{*}\!+\!2\varepsilon\BB)\!:=
& \{ \phi_{\tilde{P}_{\delta}}\!(j)\!:\!
\phi_{\tilde{P}_{\delta}} \text{ is a solution to } \tilde{\MH}_{\tilde{P}_{\delta}}, \nonumber\\
& \hspace{0.5cm}\phi_{\tilde{P}_{\delta}}(0)\in x^{*}+2\varepsilon\BB, j\in\dom\phi, j\leq J\}. \nonumber
\end{align}}
%
Now, following a similar step as in the proof of \cite[Theorem 5.4]{Ricardo:2010:simulator},
let
$$B_{\varepsilon}:=\overline{\mathrm{Reach}_{\infty,\tilde{\MH}_{\tilde{P}_{\delta}}}(x^{*}+2\varepsilon\BB)}.$$
By \eqref{eqn:Hdelta:KL}, $B_{\varepsilon}$ is bounded.
Moreover, since $B_{\varepsilon}$ is closed by definition,
it follows that it is compact.
Next, we show that it is forward invariant.
Consider a solution
$\phi_{\tilde{P}_{\delta}} \in \MS_{\tilde{\MH}_{\tilde{P}_{\delta}}}(B_{\varepsilon})$
to $\tilde{\MH}_{\tilde{P}_{\delta}}$.
Assume that there exists $j'\in \dom\phi_{\tilde{P}_{\delta}}$
for which $\phi_{\tilde{P}_{\delta}}(j')\notin B_{\varepsilon}$.
By definition of $B_{\varepsilon}$, since $\phi_{\tilde{P}_{\delta}}(0)\in B_{\varepsilon}$, the solution $\phi_{\tilde{P}_{\delta}}$
belongs to $B_{\varepsilon}$ for each $j\in \dom\phi_{\tilde{P}_{\delta}}$. This is a contradiction.
Next, we show that solutions to $\tilde{\MH}_{\tilde{P}_{\delta}}$ starting from $K$ converge to $B_{\varepsilon}$ uniformly.
\eqref{eqn:Hdelta:KL}
implies that for the given $K$ and $\varepsilon$, there exists $N>0$ such that for
every solution $\phi_{\tilde{P}_{\delta}} \in \MS_{\tilde{\MH}_{\tilde{P}_{\delta}}}(K)$ to $\tilde{\MH}_{\tilde{P}_{\delta}}$
and for each $j\in\dom \phi_{\tilde{P}_{\delta}}$, $j\geq N$:\IfTAC{~$|\phi_{\tilde{P}_{\delta}}(j)-x^{*}| \leq 2\varepsilon.$~}{\begin{equation}\label{eqn:}
|\phi_{\tilde{P}_{\delta}}(j)-x^{*}| \leq 2\varepsilon.\nonumber
\end{equation}}Then, since $B_{\varepsilon}$ is compact, forward invariant, and uniformly
attractive from $K$, by
\cite[Theorem 3.26]{Ricardo:book:2021},
$B_{\varepsilon}$
is a \IfTAC{semi-GAS}{semiglobally asymptotically stable} set for $\tilde{\MH}_{\tilde{P}_{\delta}}$.
By the construction of $\tilde{\MH}_{\tilde{P}_{\delta}}$ and Assumption~\ref{ass:data:Ps},
semiglobal asymptotic stability of $B_{\varepsilon}$ for $\MH_{P_{s}}$
with basin of attraction containing every point in $M\cap D$
follows.~\IfTAC{Finally, note that $B_{0}=\{x^{*}\}$ and that as $\varepsilon \RA 0$,}{

Finally, note that $B_{0}=\{x^{*}\}$ and that as $\varepsilon \RA 0$,}
$\Blue{\lim_{\varepsilon\searrow 0} B_{\varepsilon}=x^{*}}$.
By \eqref{eqn:Hdelta:KL}, 
$\varepsilon \searrow 0$ implies $\delta \searrow 0$.
Moreover, 
from the proof of Theorem~\ref{thm:PM:closeness},
we have $s \searrow 0$ as $\delta \searrow 0$.
It follows that
$s \searrow 0$ as $\varepsilon \searrow 0$.
Therefore, the result follows by $\MA_{s}=B_{\varepsilon}$.
\end{proof}

\IfTAC{}{
The following example illustrates that the Euler integration scheme for
differential equations satisfies the continuity property in Theorem~\ref{thm:PM:continuity}.

\begin{example}\label{exam:TCP:robust5}
Consider the hybrid congestion control system in Example~\ref{exam:TCP3}.
To approximate the hybrid Poincar\'{e} map at $x=(q,r)\in\MD_{\textrm{\tiny TCP}}$,  
we numerically compute solutions from $g_{\textrm{\tiny TCP}}(x)=(g_{1}(x), g_{2}(x))=(q_{\max}, mr)$
by discretizing the flows of the hybrid system
using Euler integration scheme with step size $s$.
From the definition of the hybrid Poincar\'{e} map and the flow from $g_{\textrm{\tiny TCP}}(x)$,
it follows that, using \eqref{eq:phif:TCP},
the computed Poincar\'{e} map after $k$ steps of size $s$
with $ks\geq T^{*}$ and $(k-1)s<T^{*}$ is
\IfTAC{
\begin{equation*}\label{sec:HPsTCP}
\MH_{P_{\textrm{\tiny TCP}_{s}}}\!:
\bx^{+}\!=\! P_{\textrm{\tiny TCP}_{s}}(\bx)=
\left[\!\!
\begin{array}{ccc}
g_{\textrm{\tiny P}_{s}}\\   
mr+aks
\end{array}
\!\!\right]
\;\;  \bx \!\in \! M_{\textrm{\tiny TCP}_{s}}\!\!\cap\MD_{\textrm{\tiny TCP}_{s}},
\end{equation*}
}
{
\begin{equation*}\label{sec:HPsTCP}
\MH_{P_{\textrm{\tiny TCP}_{s}}}:
\bx^{+} = P_{\textrm{\tiny TCP}_{s}}(\bx)=
\left[\!\!
\begin{array}{ccc}
g_{\textrm{\tiny P}_{s}}\\   
mr+aks
\end{array}
\!\!\right]
\;\;  \bx \!\in \! M_{\textrm{\tiny TCP}_{s}}\!\!\cap\MD_{\textrm{\tiny TCP}_{s}},
\end{equation*}}where $g_{\textrm{\tiny P}_{s}}:=q_{\max}+(mr-B)ks+\frac{k^2}{2}as^{2}$,
$M_{\textrm{\tiny TCP}_{s}}=M_{\textrm{\tiny TCP}}$ and $\MD_{\textrm{\tiny TCP}_{s}}=\MD_{\textrm{\tiny TCP}}$ (i.e., the approximation of the sets has no error)
satisfy Assumption~\ref{ass:data:Ps}.

To derive the fixed point $x_s^*$ of $P_{\textrm{\tiny TCP}_{s}}$,
it follows from the form of $\MH_{P_{\textrm{\tiny TCP}_{s}}}$ that $x_s^*=(q,r)$ is such that
$(mr-B)ks + \frac{k^2}{2} a s^{2}=0$ and $mr+aks=r$.
By solving these two equations, we obtain
$r=2B/(m+1)$.
Therefore, when $ks=2(B-mr)/a$,
the fixed point $x_{s}^{*}$ of $P_{\textrm{\tiny TCP}_{s}}$
can be computed as $x_{s}^{*}=(q_{\max}, 2B/(m+1))$,
which is equivalent to the fixed point $x^{*}$ of the hybrid Poincar\'{e} map $P_{\textrm{\tiny TCP}}$ in \eqref{eq:normal:PM}.
In this case, using the fact that $r=2B/(m+1)$, we have
$ks=2(B-mr)/a=2B(1-m)/(a+ma)=T^{*}$ satisfying the conditions
$ks\geq T^{*}$ and $(k-1)s<T^{*}$ for each $k \in \{1,2,3,\dots\}$.
In addition,
we have $|x_{s}^{*}-x^{*}| \RA 0 \; \mathrm{as}\; s\searrow 0$,
which illustrates Theorem~\ref{thm:PM:continuity}.
\end{example}
}

\IfTAC{\vspace{-2mm}}{}

\section{Conclusion}\label{sec:conclu}

\IfTAC{
Notions and tools for the analysis of existence and stability of hybrid limit cycles in hybrid dynamical systems were proposed.
Necessary conditions were established for the existence of hybrid limit cycles.
The Zhukovskii stability notion for hybrid systems was introduced
and
 a sufficient condition relying on Zhukovskii stability of
the hybrid system was established for the existence of hybrid limit cycles.
Sufficient and necessary conditions for the stability of hybrid limit cycles were presented.
Moreover, comparing to previous results in the literature, we established conditions for robustness of
hybrid limit cycles with respect to small perturbations and to computation error of the hybrid Poincar\'{e} map,
which is a very challenging problem in systems with impulsive effects.
Examples were included to aid the reading and
illustrate the concepts and the methodology of applying the new results.
Future work includes exercising the presented conditions on
systems of higher dimension and more intricate dynamics,
and hybrid control design for asymptotic stabilization of limit cycles
as well as their robust implementation.
}
{
Notions and tools for the analysis of existence and stability of hybrid limit cycles in hybrid dynamical systems were proposed.
Necessary conditions
were established for the existence of hybrid limit cycles.
The Zhukovskii stability notion for hybrid systems was introduced and
a relationship between Zhukovskii
stability and the incremental graphical stability was presented.
A sufficient condition relying on Zhukovskii stability of
the hybrid system was established for the existence of hybrid limit cycles.
In addition to nominal results, the key novel contributions included
an approach relying on incremental graphical stability
for the nonexistence of hybrid limit cycles.
To investigate the stability properties of the hybrid limit cycles,
we also constructed a time-to-impact function inspired by those introduced \cite{Grizzle:2001,Maghenem:ACC:2020,Hamed:Gregg:2017}.
Based on these constructions,
sufficient and necessary conditions for the stability of hybrid limit cycles were presented.
Moreover, comparing to previous results in the literature, we established conditions for robustness of
hybrid limit cycles with respect to small perturbations and to computation error of the hybrid Poincar\'{e} map,
which is a very challenging problem in systems with impulsive effects.
An extension effort that characterizes the robust stability properties for the situation where
a hybrid limit cycle may contain multiple jumps within a period
can be found in \cite{lou.li.sanfelice16:TAC,Lou:ADHS15}.  
Examples were included to aid the reading and
illustrate the concepts and the methodology of applying the new results.
Future work includes exercising the presented conditions on 
systems of higher dimension and more intricate dynamics,
and hybrid control design for asymptotic stabilization of limit cycles
as well as their robust implementation.
}

\IfTAC{} 
{  

\section{Appendix}

The following theorems are used in the proof of Theorem~\ref{thm:exist2}
and we present them here for completeness.

\begin{theorem}\label{appx:thm1}
(Tubular Flow Theorem, \cite[Chapter 2, Theorem 1.1]{Palis:1982})
Let $f$ be a vector field of class $\mathcal{C}^{r}$,
$r\geq 1$, on $U\subset \BRn$
and let $v\in U$ be a regular point of $f$.
Let $\Xi=:\{(\xi_{1},\cdots,\xi_{n})\in \BRn: |\xi_i|<1, i=1,2,\cdots,n\}$ and let $f_{\Xi}$ be the vector field on $\Xi$
defined by $f_{\Xi}(\xi)=(1,0,\cdots,0).$  Then there exists a $\mathcal{C}^{r}$ diffeomorphism
$H: \mathcal{N}_{v}\rightarrow \Xi$, for some neighborhood $\mathcal{N}_v$ of $v$ in $U,$ taking trajectories of $\dot x = f(x)$ to
trajectories of $\dot \xi = f_{\Xi}(\xi)$.
\end{theorem}

\begin{theorem}\label{appx:thm2}
(Brouwer's Fixed Point Theorem, \cite[Corollary 1.1.1]{Florenzano:2003})
Let $X$ be a nonempty compact convex subset of $\BRn$
and $P: X\rightarrow X$ a continuous (single-valued) mapping. Then there exists a $q\in X$
such that $P(q)=q$.
\end{theorem}
}

\end{document}